\documentclass[11pt,a4paper]{article}
\usepackage{xcolor}
\usepackage{fullpage}
\usepackage{amsmath}
\usepackage{amssymb}
\usepackage{graphicx}
\usepackage{amsthm}
\usepackage{bm}
\usepackage{float}
\usepackage{soul}
\usepackage{booktabs}
\usepackage[font=scriptsize,bf]{caption}
\usepackage{subfigure}
\usepackage{wrapfig}

\usepackage[linesnumbered,ruled,vlined]{algorithm2e}
\usepackage[margin=1in]{geometry}
\usepackage{tikzit}

\tikzstyle{basic}=[fill=white, draw=black, shape=circle]
\tikzstyle{square}=[fill=white, draw=black, shape=rectangle]
\tikzstyle{big dashed}=[fill=white, draw=black, shape=circle, minimum width=1cm, dashed]
\tikzstyle{vertical ellipse dashed}=[fill=none, draw=blue, minimum width=0.75cm, minimum height=3cm, ellipse, dashed, tikzit shape=rectangle, tikzit draw=blue, tikzit fill=white]
\tikzstyle{small vertical ellipse dashed}=[fill=none, draw=blue, shape=circle, tikzit fill=white, tikzit draw=blue, dashed, minimum width=0.75cm, minimum height=1.5cm, tikzit shape=rectangle, ellipse]
\tikzstyle{tiny vertical ellipse dashed}=[fill=none, draw=blue, shape=circle, tikzit fill=white, ellipse, dashed, minimum width=0.75cm, minimum height=1cm, tikzit shape=rectangle]
\tikzstyle{red}=[fill=red, draw=black, shape=circle]
\tikzstyle{green}=[fill={rgb,255: red,0; green,128; blue,128}, draw=black, shape=circle]
\tikzstyle{blue}=[fill=blue, draw=black, shape=circle]
\tikzstyle{huge dashed}=[fill=white, draw=black, shape=circle, dashed, minimum width=2cm]
\tikzstyle{medium}=[fill=white, draw=black, shape=circle, minimum width=1cm]
\tikzstyle{pale green}=[fill={rgb,255: red,173; green,231; blue,0}, draw=black, shape=circle, minimum width=1cm]
\tikzstyle{horizontal ellipse dashed}=[fill=white, draw=black, tikzit draw=magenta, tikzit shape=rectangle, minimum width=3cm, minimum height=0.75cm, ellipse, dashed]
\tikzstyle{minsize}=[fill=white, draw=black, shape=circle, minimum width=0.75cm]
\tikzstyle{horizontal ellipse green}=[fill={rgb,255: red,191; green,255; blue,0}, draw=black, tikzit draw={rgb,255: red,191; green,255; blue,0}, tikzit shape=rectangle, minimum width=3cm, minimum height=0.75cm, ellipse, dashed]
\tikzstyle{horizontal ellipse blue}=[fill={rgb,255: red,107; green,203; blue,255}, draw=black, tikzit draw=blue, tikzit shape=rectangle, minimum width=3cm, minimum height=0.75cm, ellipse, dashed]
\tikzstyle{smallblack}=[fill=black, draw=black, shape=circle, inner sep=0 pt, minimum size=3 pt]
\tikzstyle{smallSquare}=[fill=white, draw=black, shape=rectangle, inner sep=0 pt, minimum size=6 pt]
\tikzstyle{smallCircle}=[fill=white, draw=black, shape=circle, inner sep=0 pt, minimum size=20 pt]
\tikzstyle{big vertical ellipse dashed}=[fill=none, draw=blue, shape=circle, tikzit shape=rectangle, ellipse, dashed, minimum width=0.95cm, minimum height=3.7cm]
\tikzstyle{smallred}=[fill=red, draw=red, shape=circle, inner sep=0 pt, minimum size=3 pt]
\tikzstyle{smallblue}=[fill=blue, draw=blue, shape=circle, inner sep=0pt, minimum size=3pt]
\tikzstyle{small green}=[fill={rgb,255: red,0; green,107; blue,61}, draw={rgb,255: red,0; green,107; blue,61}, shape=circle, inner sep=0pt, minimum size=3pt]
\tikzstyle{med red}=[fill=red, draw=red, shape=circle, inner sep=0pt, minimum size=5pt]
\tikzstyle{med blue}=[fill=blue, draw=blue, shape=circle, inner sep=0pt, minimum size=5pt]
\tikzstyle{med green}=[fill={rgb,255: red,0; green,107; blue,61}, draw={rgb,255: red,0; green,107; blue,61}, shape=circle, inner sep=0pt, minimum size=5pt]
\tikzstyle{caterpillar_nodes}=[fill={rgb,255: red,255; green,200; blue,112}, draw=black, shape=circle, minimum width=1cm]
\tikzstyle{clique}=[fill={rgb,255: red,108; green,113; blue,158}, draw=black, shape=circle, minimum width=1.3cm]
\tikzstyle{expander}=[fill={rgb,255: red,175; green,183; blue,255}, draw=black, shape=circle, minimum width=4cm]
\tikzstyle{P}=[fill={rgb,255: red,255; green,123; blue,123}, draw=black, shape=circle, minimum width=3cm]
\tikzstyle{caterpillar_node_smaller}=[fill={rgb,255: red,255; green,200; blue,112}, draw=black, shape=circle, minimum width=1.2cm]

\tikzstyle{directed}=[->, line width=1pt]
\tikzstyle{undirected}=[-, line width=1pt]
\tikzstyle{directed red}=[draw=red, ->, line width=1pt]
\tikzstyle{directed green}=[draw={rgb,255: red,0; green,128; blue,128}, ->, line width=1pt]
\tikzstyle{directed blue}=[draw=blue, ->, line width=1pt]
\tikzstyle{directed purple}=[draw={rgb,255: red,128; green,0; blue,128}, ->, line width=1pt]
\tikzstyle{undirected red}=[-, draw=red, line width=1pt]
\tikzstyle{undirected green}=[-, draw={rgb,255: red,0; green,107; blue,61}, line width=1pt]
\tikzstyle{undirected blue}=[-, draw=blue, line width=1pt]
\tikzstyle{undirected purple}=[-, draw={rgb,255: red,128; green,0; blue,128}, line width=1pt]
\tikzstyle{undirected dashed}=[-, line width=1pt, dashed]
\tikzstyle{orange dashed}=[-, draw={rgb,255: red,255; green,128; blue,0}, dashed, line width=1.5pt]
\tikzstyle{directed dash}=[->, dashed]
\tikzstyle{blue dashed}=[-, draw=blue, dashed, line width=1pt]
\tikzstyle{green dashed}=[-, draw={rgb,255: red,0; green,162; blue,0}, dashed, line width=1pt]
\tikzstyle{blue filled}=[-, fill={blue!20}, draw=blue, line width=1pt, opacity=0.5, tikzit fill=white]
\tikzstyle{red filled}=[-, fill={red!20}, line width=1pt, draw=red, opacity=0.5, tikzit fill=white]
\tikzstyle{green filled}=[-, line width=1pt, draw={rgb,255: red,0; green,107; blue,61}, opacity=0.5, tikzit fill=white, fill={rgb,255: red,149; green,255; blue,179}]
\tikzstyle{orange filled}=[-, fill={orange!20}, draw=orange, line width=1pt, opacity=0.5, tikzit fill=white]
\tikzstyle{undirected dashed}=[-, draw=black, dashed, line width=1pt]
\tikzstyle{thick}=[-, line width=3pt]
\tikzstyle{red dashed}=[-, line width=1pt, dashed, draw=red]

\numberwithin{equation}{section}

\usepackage[pagebackref=false,citecolor=newblue,colorlinks=true,linkcolor=newblue,bookmarks=false]{hyperref}

\usepackage{datetime}
\usepackage{xcolor}
\usepackage[utf8]{inputenc}
\usepackage{graphicx}

\usepackage{tikz}

\usepackage[utf8]{inputenc} 
\usepackage[T1]{fontenc}    
\usepackage{xcolor}
\usepackage{url}            
\usepackage{booktabs}       
\usepackage{amsfonts,amsmath}       
\usepackage{nicefrac}       
\usepackage{microtype}      
\usepackage{amsthm}
\usepackage{amsmath}
\usepackage{amsfonts}
\usepackage{mathtools}
\usepackage{pdfpages}

\usepackage{xcolor}
\usepackage{tikz}
\usepackage{wrapfig}
\usepackage{pgffor}
\usepackage{float}
\usepackage{textcomp}
\usepackage{caption}

\usepackage{amssymb}
\usepackage{multicol}
\usepackage{url}

\usepackage{hyperref}

\usepackage{enumitem}

\definecolor{newblue}{rgb}{0.2,0.2,0.6} 

\newcommand{\cost}{\mathsf{cost}}
\newcommand{\COST}{\mathsf{COST}}
\newcommand{\tree}{\mathcal{T}}
\newcommand{\OPT}{\mathsf{OPT}}
\newcommand{\leaves}{\mathsf{leaves}}
\newcommand{\parent}{\mathrm{parent}}
\newcommand{\sibling}{\mathrm{sib}}
\newcommand{\HC}{\textsf{HC}}
\newcommand{\C}{B}

\newcommand{\WCOST}{\mathsf{WCOST}}
\newcommand{\WOPT}{\mathsf{WOPT}}

\newcommand{\sparsity}{\mathsf{sparsity}}

\newcommand{\SpecPart}{\texttt{Spectral Partitioning} }

\newcommand{\TSCD}{\mathcal{T}_{\mathrm{SCB}}} 
\newcommand{\treeCC}{\mathcal{T}_{\mathsf{MS}}} 
\newcommand{\treeDT}{\mathcal{T}} 

\newcommand{\parentsf}{\mathsf{parent}}
\newcommand{\sizesf}{\mathsf{size}}
\newcommand{\depthsf}{\mathsf{depth}}

\newcommand{\bkt}[1]{\ensuremath{B\lp #1 \rp}}
\newcommand{\bkting}[1]{\ensuremath{\mathcal{B}_{#1}}}
\newcommand{\buckets}{\ensuremath{\mathcal{B}}}
\newcommand{\heavyBkts}{\ensuremath{\mathcal{B}_{\mathrm{H}}}}
\newcommand{\lightBkts}{\ensuremath{\mathcal{B}}_{\mathrm{L}}}

\newcommand{\CGap}{C_{\ref{lem:MS22+}}}

\newcommand{\SpecClust}{\ensuremath{\mathsf{Spectral Clustering}}}

\newcommand{\dmax}{\Delta}
\newcommand{\dmin}{\delta}
\newcommand{\davg}{d}

\newcommand{\degfracS}{\eta_S}

\newcommand{\phiin}{\Phi_\mathrm{in}}
\newcommand{\phiout}{\Phi_\mathrm{out}}

\newcommand{\phiIn}{\phi_{\mathrm{in}}}
\newcommand{\phiOut}{\phi_{\mathrm{out}}}

\newcommand{\lp}{\left (}
\newcommand{\rp}{\right )}

\newcommand{\rsp}{\right ]}

\usepackage{caption}        

\newtheorem{theorem}{Theorem}[section]

\newtheorem{lemma}[theorem]{Lemma}
\newtheorem*{lemma*}{Lemma}
\newtheorem*{theorem*}{Theorem}
\newtheorem{definition}[theorem]{Definition}

\newtheorem{fact}{Fact}[section]

\newtheorem{claim}{Claim}[section]

\newcommand{\Vol}[1]{\mathrm{vol}\!\left(\!\,#1\,\!\right)}
\newcommand{\vol}{\mathrm{vol}}

\providecommand{\abs}[1]{\lvert#1\rvert}

\newcommand{\barr}{\overline}

\DeclarePairedDelimiter\ceil{\lceil}{\rceil}

\newcommand{\T}{\mathbb{T}}

\newcommand{\poly}{\operatorname{poly}}
\newcommand{\polylog}{\operatorname{polylog}}
\renewcommand{\deg}{\mathrm{deg}}

\newcommand{\core}{\mathsf{CORE}}

\newcommand{\argmax}{\operatorname{argmax}}

\renewcommand{\leq}{\leqslant}
\renewcommand{\geq}{\geqslant}

\renewcommand{\tilde}{\widetilde}
\renewcommand{\epsilon}{\varepsilon}

\begin{document}

\author{
Steinar Laenen\footnote{School of Informatics,    University of Edinburgh, UK. \url{steinar.laenen@ed.ac.uk}}
\and 
Bogdan-Adrian Manghiuc\footnote{School of Informatics,    University of Edinburgh, UK. \url{b.a.manghiuc@sms.ed.ac.uk}. }
\and He Sun\footnote{School of Informatics, University of Edinburgh, UK. \url{h.sun@ed.ac.uk}. This work is supported by an EPSRC Early Career Fellowship~(EP/T00729X/1).}  }

\title{Nearly-Optimal Hierarchical Clustering  for  Well-Clustered Graphs\footnote{A preliminary version of the work  appeared at the 40th International Conference on Machine Learning (ICML~'23).}}

\date{}

\maketitle
\begin{abstract}
    This paper presents two  efficient hierarchical clustering~(\textsf{HC}) algorithms with respect to Dasgupta's cost function. For any input graph $G $ with a clear cluster-structure, our designed algorithms run in nearly-linear time in the input size of $G$, and return an $O(1)$-approximate \textsf{HC} tree with respect to Dasgupta's cost function.  We    compare the performance of our algorithm against the previous state-of-the-art  on   synthetic and real-world datasets and show that our designed algorithm   produces comparable or better \HC\ trees with much lower running time. 
\end{abstract}

\thispagestyle{empty}

\newpage

\tableofcontents
    
\thispagestyle{empty}

\newpage

\setcounter{page}{1}

\newpage
\section{Introduction}\label{sec:introduction}
Hierarchical clustering~(\textsf{HC}) is the recursive partitioning of a dataset into increasingly smaller clusters via an effective binary tree representation, and has been  employed as a standard package in data analysis with widespread applications in practice.
Traditional \HC\ algorithms are typically based on agglomerative heuristics and, due to the lack of a clear objective function, there was limited work on their analysis. 
Dasgupta~\cite{dasgupta2016cost}   introduced   a simple cost function for hierarchical clustering, and this work has  inspired a number of algorithmic studies on hierarchical clustering.

In this paper we study efficient hierarchical clustering for graphs  with a clear structure of clusters. We prove that, under two different conditions of an input graph $G$ that characterise its cluster-structure,  one can construct  in nearly-linear time\footnote{We say that, for an input graph $G$ with $n$ vertices and $m$ edges, an algorithm  runs in nearly-linear time if the algorithm's running time is $O(m\cdot\log^cn)$ for some constant $c$.  For simplicity we  use  $\tilde{O}(\cdot)$ to hide this $\log^cn$ factor.} an $O(1)$-approximate \HC\ tree $\tree$ of $G$ with respect to Dasgupta's cost. Our results show that, while it's \textsf{NP}-hard to construct an $O(1)$-approximate \HC\ tree for general graphs assuming the Small Set Expansion Hypothesis~\cite{charikar2017approximate}, an $O(1)$-approximate \HC\ tree can be constructed in nearly-linear time   for a wide range of graph instances occurring in practice. This nearly-linear time complexity of our designed algorithms represents a significant improvement over the previous state-of-the-art on the same problem~\cite{manghiuc_sun_hierarchical}.

Our designed two algorithms share the same framework at the high level: we apply   spectral clustering~\cite{nips/NgJW01}   to partition an input graph $G$ into $k$ clusters $P_1, \ldots, P_k$,  and    further partition each cluster by grouping the vertices in every $P_i~(1\leq i\leq k)$ with respect to their degrees. We call the resulting vertex sets   \emph{degree buckets}, and  show that the Dasgupta cost of \HC\ trees constructed on  degree buckets is low. These  intermediate trees constructed on every bucket  can therefore  form the basis of our final tree. To construct the final \HC\ tree on $G$, we merge the    degree bucket trees based on the following two approaches:
\begin{itemize}
\item  Our first approach treats every degree bucket of vertices in $G$ as a single vertex of another ``contracted'' graph $H$ with much fewer vertices. Thanks to  the small size of $H$, we  apply the recursive sparsest cut algorithm and construct a tree on $H$ in a top-down fashion. The structure of this tree on $H$ determines how the degree bucket trees are merged when constructing our final tree of $G$.  
\item  For our second approach, we show that under a regularity assumption on the vertex degrees, it suffices to merge the clusters in a ``caterpillar'' style. This simpler construction allows our second algorithm to run in nearly-linear time for a larger number of clusters $k$ compared with the first algorithm.
\end{itemize}

To demonstrate the significance of our work, we experimentally compare our first algorithm against the previous state-of-the-art  and a well-known linkage heuristic (\texttt{AverageLinkage}) on both synthetic and real-world datasets. Our experimental results show  
that the trees constructed from our algorithm and \texttt{AverageLinkage} achieve similar cost values, which are much lower than the ones constructed from \cite{CAKMT17} and \cite{manghiuc_sun_hierarchical}.
 Moreover, our algorithm runs significantly faster than the other   three  tested algorithms.

\section{Related Work}
Addressing  the lack of an   objective function for hierarchical clustering, Dasgupta~\cite{dasgupta2016cost} introduced a simple cost function to measure the quality of an \textsf{HC} tree, and   proved several properties of the cost function. Dasgupta further showed that a recursive sparsest cut algorithm can be applied to construct an $O(\log^{3/2}n)$-approximate \textsf{HC} tree.  
  Charikar and   Chatziafratis~\cite{charikar2017approximate} improved the analysis of constructing \HC\ trees based on the sparsest cut problem,   and proved  that an $\alpha$-approximate algorithm for the sparsest cut problem can be employed to construct an $O(\alpha)$-approximate \textsf{HC} tree. This implies that, by applying the state-of-the-art for the sparsest cut problem~\cite{ARV09},  an $O(\sqrt{\log n})$-approximate \textsf{HC} tree can be constructed  in polynomial-time. 

 It is known  that, assuming the Small Set Expansion Hypothesis, it is \textsf{NP}-hard to construct an $O(1)$-approximate \HC\ tree for general graphs~\cite{charikar2017approximate}. Hence, it is natural to examine the conditions of input graphs under which an  $O(1)$-approximate \HC\ tree can be constructed in polynomial-time. Cohen-Addad et al.~\cite{CAKMT17} studied a hierarchical extension of the classical stochastic block model~(\textsf{SBM}) and showed that, for graphs randomly generated from this model, there is an \textsf{SVD} projection-based algorithm~\cite{mcsherry2001spectral} that, together with linkage heuristics,     constructs a $(1 + o(1))$-approximate \textsf{HC} tree with high probability. 
Manghiuc and Sun~\cite{manghiuc_sun_hierarchical} studied hierarchical clustering for well-clustered graphs and proved that, when there is a cluster-structure of an input graph $G$, an $O(1)$-approximate \HC\ tree of $G$ can be constructed in polynomial-time; their designed algorithm is based on the graph decomposition algorithm by Oveis Gharan and   Trevisan~\cite{GT14}, and has high time complexity.

 There are recent studies of 
   hierarchical clustering in different models of computation. For instance,  Kapralov et al.~\cite{KKLZ22} studied the problem of learning the hierarchical cluster structure of graphs in a semi-supervised setting. 
Their presented algorithm runs in sub-linear time and, under  some clusterability conditions of an input graph $G$ with $k$ clusters, their algorithm  $O(\sqrt{\log k})$-approximates Dasgupta's cost of an optimal \textsf{HC} tree. 
This work is incomparable to ours:  the objective of their work is to approximate Dasgupta's cost of an \textsf{HC} tree, while the output of our algorithms is  a complete \textsf{HC} tree.  

Finally, there are  studies of hierarchical clustering under different objective functions. 
  Moseley and   Wang~\cite{moseley2017approximation} studied the dual of Dasgupta's cost function.
   This objective, and a dissimilarity objective by  Cohen-Addad et al.~\cite{cohen2018hierarchical}, have    
   received considerable attention~\cite{alon2020hierarchical, charikar2019better_than_AL, chatziafratis2020bisect}. It is important to notice that an $O(1)$-approximate \textsf{HC} tree can be constructed efficiently for general graphs under these objectives, suggesting the fundamental difference in the computational complexity of constructing an \textsf{HC} tree under different objectives. This is the main reason  for us to entirely focus on Dasgupta's cost function in this work.

\section{Preliminaries}\label{sec:preliminaries}
This section lists  the background knowledge used in our paper, and is organised as follows: In Section~\ref{sec:notation} we list the basic notation and  facts in spectral graph theory. 
Section~\ref{sec:hc} discusses   hierarchical clustering and Dasgupta's  cost function. In Section~\ref{sec:contracted graphs} we introduce the notion of contracted graphs, and we finish the section with a brief introduction to spectral clustering in 
Section~\ref{sec:sc}.

\subsection{Notation\label{sec:notation}}

We always assume that $G=(V,E,w)$ is an undirected graph with $|V| = n$ vertices, $|E| = m$ edges, and weight function $w: V\times V\rightarrow \mathbb{R}_{\geq 0}$. 
For any edge $e = \{u, v\} \in E$, we write $w_e$ or $w_{uv}$ to express the   weight of $e$. Let $w_{\mathrm{min}}$ and $w_{\mathrm{max}}$ be the minimum and maximum edge weight of $G$, respectively; we further assume that $w_{\mathrm{max}}/w_{\mathrm{min}} \leq c\cdot n^{\gamma}$ for some constants $\gamma>0$ and $c$, which are independent of the input size.

For a vertex $u \in V$, we denote its  \emph{degree}  by $d_u \triangleq \sum_{v \in V} w_{uv}$. 
We   use $\delta_G, \Delta_G$, and $d_G$ for the minimum, maximum and average degrees in $G$ respectively, where $d_G \triangleq \sum_{u \in V} d_u / n$. For any $S\subset V$, we use $\delta_G(S), \Delta_G(S)$, and $d_G(S)$ to represent the minimum, maximum, and average degrees of the vertices of $S$ in $G$.

For any two subsets $S, T \subset V$, we define the \emph{cut value} between $S$ and $T$ by  $w(S, T) \triangleq \sum_{e \in E(S, T)} w_e$, where $E(S, T)$ is the set of edges between $S$ and $T$.
For any $G=(V,E, w)$ and  set $S \subseteq V$, the \emph{volume} of $S$ is $\vol_G(S)\triangleq \sum_{u\in S} d_u$,
and we write $\vol(G)$ when referring to $\vol_G(V)$. Sometimes we drop the subscript $G$ when it is clear from the context.  For any non-empty subset $S\subset V$, we define $G[S]$ to be the induced subgraph on $S$.

For any input graph $G=(V,E, w)$ and any $S\subset V$, let the conductance of   $S$ in $G$ be
\[
\Phi_G(S) \triangleq \frac{w(S, V\setminus S) }{\vol(S)}.
\] We define  the conductance of $G$ by
\[
 \Phi_G\triangleq \min_{\substack{S\subset V\\ \vol(S) \leq \vol(V)/2}} \Phi_G(S).
\]

For any non-empty subset $S \subset V$ we   refer to $\Phi_G(S)$ as the \emph{outer conductance} of $S$ with respect to $G$, and $\Phi_{G[S]}$ as the \emph{inner conductance} of $S$.  

Our analysis is based on the spectral properties of graphs, and here  we list the basics of  spectral graph theory. 
For a graph $G = (V, E, w)$,  let $\mathbf{D}\in\mathbb{R}^{n\times n}$ be the diagonal matrix defined by $\mathbf{D}_{uu} = d_u$ for all $u \in V$. We  denote by  $\mathbf{A}\in\mathbb{R}^{n\times n}$  the  \emph{adjacency matrix}  of $G$, where $\mathbf{A}_{uv} = w_{uv}$ for all $u, v \in V$. The \emph{normalised Laplacian matrix} of $G$ is defined as $\mathcal{L} \triangleq \mathbf{I} - \mathbf{D}^{-1/2} \mathbf{A} \mathbf{D}^{-1/2}$, where $\mathbf{I}$ is the $n \times n$ identity matrix.  The normalised Laplacian $\mathcal{L}$ is symmetric and real-valued, and   has $n$ real eigenvalues which we   write as  $\lambda_1 \leq \ldots \leq \lambda_n$. It is known that  $\lambda_1=0$ and $\lambda_n\leq~2$~\cite{chung1997spectral}.

For any integer $k\geq 2$, we call 
 subsets of vertices $A_1,\ldots, A_k$ a $k$-way partition of $G$ if $\bigcup_{i=1}^k A_i = V$ and 
 $A_i\cap A_j=\emptyset$ for different $i$ and $j$. We   define the \emph{$k$-way expansion} of $G$ by
\begin{equation*}
    \rho(k) \triangleq \min_{\mathrm{partitions \:} A_1, \dots, A_k} \max_{1\leq i \leq k} \Phi_G(A_i).
\end{equation*}
The 
celebrated higher-order Cheeger inequality~\cite{higherCheeg} states that it holds for any graph $G$ and $k \geq 2$ that 
\begin{equation}\label{eq:Higher Cheeger}
    \frac{\lambda_k}{2} \leq \rho(k) \leq O(k^3) \sqrt{\lambda_k}.
\end{equation}

\subsection{Hierarchical Clustering\label{sec:hc}}

A \emph{hierarchical clustering (\textsf{HC}) tree} of a given graph $G$ is a binary tree $\tree$ with $n$ leaf nodes such that each leaf corresponds to exactly one vertex $v \in V(G)$.
Let $\tree$ be an \textsf{HC} tree of   a graph $G = (V, E, w)$, and  $N \in \tree$ be an arbitrary internal node\footnote{We consider any non-leaf node of $\tree$ an \emph{internal node}. We   always use the term \emph{node(s)} for the nodes of $\tree$ and the term \emph{vertices} for the elements of the vertex set $V$.} of $\tree$. 

We denote $\tree[N]$ to be the subtree of $\tree$ rooted at $N$,  $\leaves \lp \tree[N]\rp$ to be the set of leaf nodes of $\tree[N]$,
and  $\parent_{\tree}(N)$ to be the parent of node $N$ in $\tree$.

 In addition, each internal node $N\in \tree$   induces a unique vertex set $C \subseteq V$ formed by the vertices corresponding to $\leaves(\tree[N])$.  
For ease of presentation, we   sometimes abuse the notation  and write $N \in \tree$ for both the internal node of $\tree$ and the corresponding subset of vertices in $V$.

To measure the quality of an \textsf{HC} tree $\tree$ with similarity weights, Dasgupta~\cite{dasgupta2016cost} introduced the cost function defined by
\[
    \COST_{G}(\tree) \triangleq \sum_{e = \{u,v\} \in E} w_e \cdot \abs{\leaves \lp \tree[u \vee v]\rp},
\]
where $u \vee v$ is the lowest common ancestor of $u$ and $v$ in $\tree$. Sometimes, it is convenient to consider the cost of an edge $e = \{u, v\} \in E$ in $\tree$ as $\cost_{G}(e) \triangleq w_e \cdot \abs{\leaves(\tree[u \vee v])}$.
Trees that achieve a better hierarchical clustering have a lower cost, and the objective of \textsf{HC} is to construct trees with the lowest cost based on the following consideration: for any pair of vertices $u, v \in V$ that corresponds to an edge of high weight $w_{uv}$ (i.e., $u$ and $v$ are highly similar), a ``good'' \textsf{HC} tree would  separate $u$ and $v$ lower in the tree, thus reflected in a small size of $|\leaves(\tree[u \vee v])|$. 
We denote by $\OPT_G$ the minimum cost of any \textsf{HC} tree of $G$, i.e., $\OPT_G = \min_{\tree} \COST_G(\tree)$, and use the notation $\tree^*$ to refer to  an \emph{optimal} tree achieving the minimum cost. We say that an \textsf{HC} tree $\tree$ is an \emph{$\alpha$-approximate} tree if $\COST_G(\tree)\leq \alpha\cdot \OPT_G$ for some $\alpha \geq 1$.

\subsection{Contracted Graphs}\label{sec:contracted graphs}
Our work is based on \emph{contracted} graphs, which were   introduced by Kapralov et al.~\cite{KKLZ22} in the context of hierarchical clustering.

\begin{definition}[Contracted Graph, \cite{KKLZ22}]
Let $G=(V, E, w)$ be a weighted graph, and  $\mathcal{A}=\{A\}_{i=1}^k$ be a partition of $V$. 
We say that the vertex and edge-weighted graph $H=([k], \genfrac(){0pt}{2}{[k]}{2}, W^*, w^*)$ is a contraction of $G$ with respect to $\mathcal{A}$ if for every $i,j \in [k]$ we have that $W^*(i,j) = w(A_i, A_j)$ and for every $i \in [k]$ we have $w^*(i) = |A_i|$. 
We denote the contraction of $G$ with respect to $\mathcal{A}$ as $G / \mathcal{A}$.
\end{definition}

Note that contracted graphs are \emph{vertex-weighted}, i.e.,   every vertex $u \in V(H)$ has a corresponding weight. To measure the quality of an \HC\ tree $\tree$ on a vertex-weighted graph $H=(V, E, W, w)$,  we define the \emph{weighted Dasgupta's cost} of $\tree$ on $H$ as
\begin{equation*}
\WCOST_H(\tree) \triangleq \sum_{e=\{u,v\} \in E} W(u, v)  \sum_{z \in \leaves\lp\tree[u \vee v]\rp} w(z).
\end{equation*}
We denote by $\WOPT_H$ the minimum cost of any \textsf{HC} tree of $H$, i.e., $\WOPT_H = \min_{\tree} \WCOST_H(\tree)$.

For any set $S \subset V(H)$ we define the sparsity of the cut $(S, V \setminus S)$ in $H$ as 
\begin{equation}\label{eq:weighted_sparsity}
\sparsity_H(S) \triangleq \frac{W(S, V \setminus S)}{w(S) \cdot w(V \setminus S)},    
\end{equation}
where $w(S) \triangleq \sum_{v \in S} w(v)$. The vertex-weighted sparsest cut of $G$ is the cut with the minimum sparsity. 

We call the vertex-weighted variant of the recursive sparsest cut algorithm the \textsf{WRSC} algorithm, which is described as follows:  
Let $\alpha \geq 1$, and $H=(V, E, W, w)$ be a vertex and edge-weighted graph. Let $(S, V \setminus S)$ be a vertex-weighted sparsest cut of $H$. The \textsf{WRSC} algorithm on $H$ is a recursive algorithm that   finds a cut $(T, V \setminus T)$ satisfying $\sparsity_H(T) \leq \alpha \cdot \sparsity_H(S)$, and   recurs on the induced subgraphs $H[T]$ and $H[V \setminus T]$.  Kapralov et al.~\cite{KKLZ22}   showed that the approximation guarantee of this algorithm for constructing \HC\ trees follows from the one for non-vertex-weighted graphs~\cite{charikar2017approximate}, and their result is summarised as follows:

\begin{lemma}[\cite{KKLZ22}]\label{lem:wcost_rsc_approx}
Let $H=(V, E, W, w)$ be a vertex and edge-weighted graph. Then, the \textsf{WRSC} algorithm achieves an $O(\alpha)$-approximation for the weighted Dasgupta's cost of $H$, where $\alpha$ is the approximation ratio of the sparsest cut algorithm used in    \textsf{WRSC}.
\end{lemma}

We emphasise that we only use the \emph{combinatorial} properties of vertex-weighted graphs. As such we don't consider their    Laplacian matrices and the corresponding spectra.

\subsection{Spectral Clustering\label{sec:sc}}
Another key component used in our analysis is spectral clustering, which is one of the most popular clustering algorithms used in practice~\cite{nips/NgJW01, spielmanSpectralPartitioningWorks1996}. 
For any input graph $G=(V, E, w)$ and $k \in \mathbb{N}$, spectral clustering consists of the following three steps: (1) compute the eigenvectors $f_1 \ldots f_k$ of $\mathcal{L}_G$, and embed each $u \in V(G)$ to the point $F(u)\in \mathbb{R}^k$ based on $f_1 \ldots f_k$;
(2) apply $k$-means on the embedded points $\{F(u)\}_{u \in V(G)}$; (3) partition $V$ into $k$ clusters $P_1 \ldots P_k$ based on the output of $k$-means.

To analyse the   performance of spectral clustering, one can examine the scenario in which there is a large gap between $\lambda_{k+1}$ and $\rho(k)$. By the higher-order Cheeger inequality~\eqref{eq:Higher Cheeger}, we know that a low value of $\rho(k)$ ensures that the vertex set $V$ of $G$ can be partitioned into $k$ subsets~(clusters), each of which has conductance upper bounded by $\rho(k)$; on the other hand, a high value of $\lambda_{k+1}$ implies that  any $(k+1)$-way partition of $V$ would introduce some $A\subset V$ with conductance $\Phi_G(A)\geq \rho(k+1)\geq \lambda_{k+1}/2$.

Based on this observation, a sequence of works~\cite{MT, MS22,  peng_partitioning_2017} showed that, assuming the presence of a large gap between $\lambda_{k+1}$ and $\rho(k)$,
  spectral clustering  returns clusters $P_1, \ldots P_k$ of low outer conductance $\Phi_G(P_i)$ for each $1 \leq i \leq k$. We remark that spectral clustering can be implemented in nearly-linear time~\cite{peng_partitioning_2017}.

\section{Hierarchical Clustering for Well-Clustered Graphs: Previous Approach}  \label{sec:ms}
 Our presented new algorithms are based  on the work of Manghiuc and Sun~\cite{manghiuc_sun_hierarchical} on the same problem, and this section gives a brief overview of their approach.  We consider a graph $G = (V, E, w)$ to have $k$ well-defined clusters if $V(G)$ can be partitioned into disjoint subsets $\{A_i\}_{i=1}^k$ such that (i) there's a sparse cut between $A_i$ and $V \setminus A_i$, formulated as $\Phi_G(A_i) \leq \phiout$ for any $1 \leq i \leq k$,  and (ii) each $G[A_i]$ has high inner conductance $\Phi_{G[A_i]} \geq \phiin$. Then, for an input graph $G$ with a clear cluster-structure, the algorithm by  \cite{manghiuc_sun_hierarchical}, which we call the \textsf{MS} algorithm in the following, constructs an   $O(1)$-approximate \textsf{HC} tree of $G$ in polynomial-time.  At a very high level, the \textsf{MS} algorithm is based on  the following  two components:
\begin{itemize}
  \item They first show that, when an input graph $G$ of $n$ vertices and $m$ edges has high conductance, an $O(1)$-approximate \textsf{HC} tree   of $G$  can be   constructed based on the degree sequence of $G$, and the algorithm runs in $O(m+ n\log n)$ time. We use $\T_{\deg}$ to express such trees constructed   from   the degree sequence of $G$.
  
  \item They  combine the first result with the algorithm proposed in  \cite{GT14}, which decomposes an input graph into a set of clusters $A_1,\ldots, A_{\ell}$, where every $A_i$ has low outer-conductance $\Phi_G(A_i)$ and   high inner-conductance $\Phi_{G[A_i]}$ for any $1\leq i\leq \ell$. 
\end{itemize}
With  these two components, one might intuitively think that an $O(1)$-approximate \textsf{HC} tree of $G$ can be easily constructed by (i) finding    clusters $\{A_i\}$ of high conductance, (ii) constructing an $O(1)$-approximate \HC\ tree  $\T_i=\T_{\deg}(G[A_i])$ for every $G[A_i]$, and (iii)   merging the constructed $\{\T_i\}$  in an ``optimal'' way. However,  Manghiuc and Sun~\cite{manghiuc_sun_hierarchical} show by counterexample that this is not sufficient and, in order to achieve an $O(1)$-approximation,  further decomposition of every $\{A_i\}$ would be necessary. Specifically, they  adjust the algorithm of  \cite{GT14},
and further decompose every  vertex set $A_i$ of high-conductance into  smaller subsets, which they call  \emph{critical nodes}. They show that these critical nodes can be 
 carefully merged to construct an $O(1)$-approximate \HC\ tree for  well-clustered graphs. On the downside, as the \textsf{MS} algorithm is heavily  based \cite{GT14},
 the time complexity of their algorithm is   $\tilde{O} \lp k^3 m^2n \cdot \lp w_{\max}/w_{\min}\rp\rp$, which limits the application of their algorithm on   large-scale datasets.

\section{Algorithms}\label{sec:our_algorithm}
 
This section presents our hierarchical clustering algorithms for well-clustered graphs. It consists of two subsections, each of which corresponds to one algorithm.

\subsection{The First Algorithm}

In this subsection we present a nearly-linear time algorithm that, given a well-clustered graph $G$ with a constant number of clusters as input, constructs an $O(1)$-approximate \HC\ tree of $G$ with respect to Dasgupta's cost. Our result is as follows:

\begin{theorem}\label{thm:main1}
Given a connected graph $G=(V,E,w)$ and some constant $k$ as input such that $\lambda_{k+1}=\Omega(1)$,  $\lambda_{k} = O(k^{-12})$, and 
$w_\mathrm{max} / w_\mathrm{min} = O(n^\gamma)$ for a constant $\gamma>1$,   there is a nearly-linear time algorithm that constructs an \textsf{HC} tree $\mathcal{T}$ of 
$G$ satisfying $\COST_G(\treeDT) = O\lp 1 \rp \cdot \OPT_G$.
\end{theorem}

In comparison to the previous algorithms for hierarchical clustering on well-structured  graphs~\cite{CAKMT17,manghiuc_sun_hierarchical}, the  advantages of our algorithm are its simplicity and nearly-linear time complexity, which is optimal up to a poly-logarithmic factor. 

\subsubsection{Overview of the Algorithm}

We first describe the algorithm behind Theorem~\ref{thm:main1}. 
Our  algorithm consists of four steps.   For any input graph $G=(V,E,w)$ and parameter $k$, our algorithm first runs spectral clustering and obtains clusters $P_1,\ldots, P_k$.

The second step of our algorithm is a   degree-based bucketing procedure. We set $\beta\triangleq 2^{k(\gamma+1)}$, and define for any   $P_i$ returned from spectral clustering  and $u\in P_i$ the set 
\begin{equation}\label{eq:bubeta}
\bkt{u} \triangleq \left\{v \in P_i : d_u \leq d_v < \beta \cdot d_u \right\}.
\end{equation}
Moreover, as $B(u)$ consists of all the vertices $v \in P_i$ with     $d_u\leq d_v < \beta \cdot d_u$, we generalise the  definition of $B(u)$  and define for any $j\in\mathbb{Z}$ that 
\[
B^j(u) = \left\{ v \in P_i : \beta^{j} \cdot d_u \leq d_v < \beta^{j+1} \cdot d_u\right\}.
\]
By definition, it holds that $B(u)=B^0(u)$, and $\{B^j(u)\}_j$ forms a partition of $P_i$ for any $u\in P_i$. We illustrate the bucketing procedure in Figure~\ref{fig:bucketing_main}.  In our algorithm, we set  $u^{(i)}~(1\leq i\leq k)$ to be a vertex in $P_i$ with minimum degree  and   focus on $\left\{B^j\left(u^{(i)}\right)\right\}_j$, the partition of $P_i$ which is only based on non-negative values of $j$.  Let   $
\bkting{i} \triangleq \left\{B^j\left(u^{(i)}\right)\right\}_j$, and  $\mathcal{B} \triangleq \bigcup_{i=1}^k \mathcal{B}_i$.

\begin{figure}[h]
\centering
    \begin{minipage}{0.5\textwidth}
    \centering
    \resizebox{7.8cm}{!}{%
    \begin{tikzpicture}
	\begin{pgfonlayer}{nodelayer}
		\node [style=none] (2) at (-7, 0.5) {};
		\node [style=none] (3) at (5, 0.5) {};
		\node [style=none] (20) at (-1, -2) {};
		\node [style=none] (21) at (-3, -2) {};
		\node [style=none] (22) at (-5, -2) {};
		\node [style=none] (23) at (1, -2) {};
		\node [style=none] (24) at (3, -2) {};
		\node [style=none] (25) at (5, -2) {};
		\node [style=none] (27) at (-7, -2) {};
		\node [style=none] (46) at (-1, -4) {};
		\node [style=none] (47) at (1, -4) {};
		\node [style=none] (48) at (3, -4) {};
		\node [style=none] (50) at (-1, -4) {};
		\node [style=none] (53) at (1, -3.75) {};
		\node [style=none] (54) at (1, -4.25) {};
		\node [style=none] (55) at (-1, -3.75) {};
		\node [style=none] (56) at (-1, -4.25) {};
		\node [style=none] (57) at (-1, -4.75) {$d_u$};
		\node [style=none] (64) at (-1, -4.25) {};
		\node [style=none] (65) at (1, -4.75) {$d_u \beta$};
		\node [style=none] (67) at (-7, 1.25) {};
		\node [style=smallred] (68) at (-1, 0.5) {};
		\node [style=smallred] (69) at (-1.25, 0.5) {};
		\node [style=smallred] (70) at (-1.5, 0.5) {};
		\node [style=smallred] (71) at (-0.75, 0.5) {};
		\node [style=smallred] (74) at (2.5, 0.5) {};
		\node [style=smallred] (75) at (3.5, 0.5) {};
		\node [style=smallred] (76) at (4.25, 0.5) {};
		\node [style=smallred] (78) at (-0.5, 0.5) {};
		\node [style=smallred] (79) at (-3.25, 0.5) {};
		\node [style=smallred] (80) at (-4.25, 0.5) {};
		\node [style=smallred] (81) at (-5.5, 0.5) {};
		\node [style=smallred] (82) at (-6.25, 0.5) {};
		\node [style=smallred] (83) at (-6.25, -3) {};
		\node [style=smallred] (84) at (-5.75, -3) {};
		\node [style=smallred] (85) at (-4.25, -3) {};
		\node [style=smallred] (86) at (-3.75, -3) {};
		\node [style=smallred] (87) at (-2.25, -3) {};
		\node [style=smallred] (88) at (-1.75, -3) {};
		\node [style=smallred] (89) at (-2, -2.5) {};
		\node [style=smallred] (90) at (-0.25, -3) {};
		\node [style=smallred] (91) at (0.25, -3) {};
		\node [style=smallred] (92) at (0, -2.5) {};
		\node [style=smallred] (93) at (2.25, -3) {};
		\node [style=smallred] (94) at (3.75, -3) {};
		\node [style=smallred] (95) at (4.25, -3) {};
		\node [style=smallred] (97) at (1.5, 0.5) {};
		\node [style=smallred] (98) at (1.75, -3) {};
		\node [style=none] (99) at (-1, 1.5) {};
		\node [style=none] (100) at (-1, 1.5) {};
		\node [style=none] (101) at (5, -4) {};
		\node [style=none] (102) at (-7, -4) {};
		\node [style=none] (103) at (-3, -3.75) {};
		\node [style=none] (104) at (-3, -4.25) {};
		\node [style=none] (105) at (-3, -4.75) {$d_u \beta^{-1}$};
		\node [style=none] (106) at (-1, 1) {};
		\node [style=none] (107) at (5, 1.5) {};
		\node [style=none] (108) at (5, 1.5) {};
		\node [style=none] (109) at (5, 1) {};
		\node [style=none] (110) at (-7, 1.5) {};
		\node [style=none] (111) at (-7, 1.5) {};
		\node [style=none] (112) at (-7, 1) {};
		\node [style=none] (113) at (-7, 2) {$\dmin_G(P_i)$};
		\node [style=none] (114) at (-1, 2) {$d_u$};
		\node [style=none] (115) at (5, 2) {$\dmax_G(P_i)$};
		\node [style=none] (116) at (0, -0.5) {$B(u)$};
		\node [style=none] (117) at (0, -1) {};
		\node [style=none] (118) at (0, -1.5) {};
		\node [style=none] (119) at (0, -1.5) {};
		\node [style=none] (120) at (0, -1) {};
		\node [style=smallred] (121) at (5, 0.5) {};
		\node [style=smallred] (122) at (-7, 0.5) {};
		\node [style=smallred] (123) at (-6, -2.5) {};
		\node [style=smallred] (124) at (4, -2.5) {};
		\node [style=smallred] (125) at (-0.25, 0.5) {};
	\end{pgfonlayer}
	\begin{pgfonlayer}{edgelayer}
		\draw [style=undirected] (2.center) to (3.center);
		\draw [style=undirected, bend right=90, looseness=2.50] (27.center) to (22.center);
		\draw [style=undirected, bend right=90, looseness=2.50] (22.center) to (21.center);
		\draw [style=undirected, bend right=90, looseness=2.50] (21.center) to (20.center);
		\draw [style=undirected, bend right=90, looseness=2.50] (23.center) to (24.center);
		\draw [style=undirected, bend right=90, looseness=2.50] (24.center) to (25.center);
		\draw [style=undirected] (47.center) to (53.center);
		\draw [style=undirected, in=90, out=-90] (47.center) to (54.center);
		\draw [style=undirected] (50.center) to (55.center);
		\draw [style=undirected] (50.center) to (56.center);
		\draw [style=undirected, in=270, out=-90, looseness=2.50] (20.center) to (23.center);
		\draw [style=directed dash] (50.center) to (102.center);
		\draw [style=directed dash] (50.center) to (101.center);
		\draw [style=undirected] (103.center) to (104.center);
		\draw [style=directed] (100.center) to (106.center);
		\draw [style=directed] (108.center) to (109.center);
		\draw [style=directed] (111.center) to (112.center);
		\draw [style=directed] (120.center) to (119.center);
	\end{pgfonlayer}
\end{tikzpicture}
    }\\
    \end{minipage}%
\caption{The  illustration of  our bucketing procedure induced by a vertex $u$.} \label{fig:bucketing_main}
\end{figure}
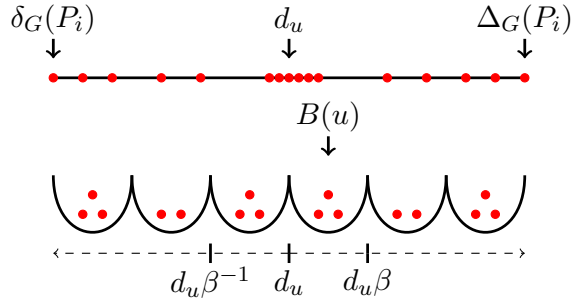

In the third step, our algorithm constructs an arbitrary balanced binary tree $\tree_{B }$ for every bucket $B \in \mathcal{B}$, and sets $\mathbb{T}$ to be the collection of our constructed balanced trees $\tree_{B }$. 

For the fourth step, our algorithm constructs the contracted graph $H$ defined by $
H\triangleq G\big/ \mathcal{B}$,
and applies the \textsf{WRSC} algorithm to construct an \textsf{HC} tree $\mathcal{T}_{G/\mathcal{B}}$ of $H$. Our algorithm finally   replaces every leaf node of $\mathcal{T}_{G/\mathcal{B}}$ corresponding to  a set $B\in\mathcal{B}$ by an arbitrary balanced tree $\mathcal{T}_{B}$ on $G\left[B\right]$. See Algorithm~\ref{alg:spectral_clustering_degree_bucketing} for the formal description of our algorithm.

\begin{algorithm}
    \DontPrintSemicolon
    
    \KwData{$G=(V,E, w)$, $k \in \mathbb{N}$ such that $\lambda_{k+1} > 0$}
    \KwResult{$\treeDT$}
    
    $   \{P_1,\ldots,P_k\}\leftarrow\mathsf{SpectralClustering}(G,k)$\label{alg:specWRSC:line:step1}\;
    $\beta \leftarrow 2^{k(\gamma + 1)}$\;
    $\mathbb{T} \leftarrow \emptyset$\;

    \For{every $P_i$}
    {
        Order the vertices of $P_i$   increasingly with respect to  their degrees\label{alg:specWRSC:line:step2_start}\;
        $u^{(i)}\leftarrow$ vertex $u\in P_i$ with minimum degree \;
        $ \bkting{i} \leftarrow \left\{B^j\left(u^{(i)}\right)\right\}_j$\label{alg:specWRSC:line:step2_end}\;

        \For{every $B \in \bkting{i}$ \label{alg:specWRSC:line:step3_start}}
        {
            Construct an arbitrary     balanced   $\tree_{B}$  of $G\left[B\right]$ \;
            $\mathbb{T} \leftarrow \mathbb{T} \cup \tree_{B} $\label{alg:specWRSC:line:step3_end}\;
        }
    }
    
    $\mathcal{B} = \bigcup_{i=1}^k \mathcal{B}_i$\label{alg:specWRSC:line:step4_start}\;
    
    $H\leftarrow G\big/\mathcal{B}$\label{alg:specWRSC:line:construct_contracted}\;
    $\tree \leftarrow \mathsf{WRSC}(H)$\label{alg:specWRSC:line:WRSC}\;
    
    \For{ every  $\tree_{B} \in \T $}
    {
        Replace the leaf node of $\tree$   corresponding to $B $ by   $\tree_{B}$\label{alg:specWRSC:line:step4_end}\;
    }
    
    \Return $\treeDT$\;
    
\caption{Spectral Clustering with Degree Bucketing and \textsf{WRSC} (\textsf{SpecWRSC})}
\label{alg:spectral_clustering_degree_bucketing}
\end{algorithm}

It is important to notice that, as every output $P_i~(1\leq i\leq k)$ from spectral clustering doesn't necessarily have high inner-conductance, our work shows that in general the inner-conductance condition of vertex sets is not needed in order to construct an $O(1)$-approximate \HC\ tree. This is significant in our point of view, since ensuring the high inner-conductance   of certain vertex sets is the main reason behind the high time complexity of the \textsf{MS} algorithm. Moreover, the notion of critical nodes introduced in \cite{manghiuc_sun_hierarchical}, which are subsets of a cluster with high conductance,  is replaced by the degree-based bucketing $\mathcal{B}$ of every output cluster $P_i$ from spectral clustering; this $\mathcal{B}$ can be constructed in nearly-linear time.

\subsubsection{Proof Sketch of Theorem~\ref{thm:main1}} We first analyse the approximation ratio of our constructed tree, and  show why these simple  four steps   suffice to construct an $O(1)$-approximate \HC\ tree for well-clustered graphs. By the algorithm description, our constructed $\mathcal{T}$ is based on merging different $\tree_{B}$, and   it holds that \begin{align}
     \COST_G \lp \mathcal{T} \rp  &     = \sum_{i=1}^k \sum_{B \in  \mathcal{B}_i } \COST_{G\left[B \right]}\lp \tree_{B}\rp  + \sum_{\substack{B, B' \in \buckets \\ B \neq B'}} \sum_{\substack{e=\{u,v\} \\ u \in B, v \in B' }}\cost_G(e).\label{eq:sum_cost_degree_buckets_nodegree}
\end{align}
The first step of our analysis is to upper bound the total contribution of the internal edges of all the buckets, and
 our result is as follows:
\begin{lemma}\label{lemma:upper bound inner nodegree}
   It holds that 
\[
        \sum_{i=1}^k \sum_{B \in \buckets_{i}} \COST_{G[B]} \left( \tree_{B} \right) =   O\lp  \beta \cdot k^{23}/\lambda_{k+1}^{10}\rp \cdot \OPT_G.
\]
\end{lemma}

By Lemma~\ref{lemma:upper bound inner nodegree},  the total cost induced from the edges inside every bucket can be \emph{directly} used when constructing the final tree $\tree$. This is crucial for our analysis since our defined buckets can be constructed in nearly-linear time. In comparison to this step, the  \textsf{MS} algorithm relies on finding critical nodes, which is computationally much more expensive.

Next, we analyse the second term of \eqref{eq:sum_cost_degree_buckets_nodegree}. For ease of presentation, let $\widetilde{E}$ be the set of edges crossing different buckets, and $\tree_{G/ \mathcal{B}}$ the tree returned from the \textsf{WRSC} algorithm; remember that every leaf of $\tree_{G / \mathcal{B}}$ corresponds to a vertex set in $\mathcal{B}$. By construction, we know that the weight of 
every edge $e\in\widetilde{E}$ contributes to the edge weight in $G / \mathcal{B}$, and it holds that  
\begin{equation}\label{eq:equiv:wcost_unweightedcost}
    \sum_{e \in \tilde{E}} \cost_G(e) = \WCOST_{G / \buckets}(\mathcal{T}_{G/ \mathcal{B}}).
\end{equation}
Moreover, since  $\mathcal{T}_{G\setminus \mathcal{B}}$ is constructed by performing the \textsf{WRSC} algorithm on $G\setminus \mathcal{B}$, we have by Lemma~\ref{lem:wcost_rsc_approx} that 
\begin{equation}\label{eq:wopt_cost_crossing_edges_nodegree}
\WCOST_{G / \buckets}(\tree_{G / \mathcal{B}}) = O\lp \alpha\rp \cdot \WOPT_{G / \buckets},
\end{equation}
where $\alpha$ is the approximation ratio achieved by the \textsf{WRSC} algorithm. 

Our next lemma is the key to the overall analysis, and upper bounds $\WOPT_{G / \buckets}$ with respect to $\OPT_G$. 
\begin{lemma}\label{lemma:upper bound WOPT nodegree}
   It holds that 
\[
        \WOPT_{G / \buckets} =  O\lp  \beta \cdot k^{23}/\lambda_{k+1}^{10}\rp \cdot \OPT_G.
\]
\end{lemma}
Combining   \eqref{eq:equiv:wcost_unweightedcost}, \eqref{eq:wopt_cost_crossing_edges_nodegree} with Lemma~\ref{lemma:upper bound WOPT nodegree}, we have 
 \begin{equation} \label{eq:upperbound_crossing_edges}
\sum_{e \in \tilde{E}} \cost_G(e) = O\lp  \alpha \cdot \beta \cdot k^{23}/\lambda_{k+1}^{10}\rp \cdot \OPT_G.
\end{equation}
Next, we study $\alpha$, which is the only term in the approximation ratio of \eqref{eq:upperbound_crossing_edges} that is not necessarily a constant.  Recall that our choice of $\gamma$ and $\beta$ satisfies
\[
\frac{w_{\max}}{w_{\min}} = O(n^{\gamma})
\]
and $\beta = 2^{k (\gamma+1)}$. Hence,  it holds that 
\[
\frac{\Delta_G}{\delta_G}  = O\left(n^{\gamma+1}\right),
\]
and the total number of buckets in $\mathcal{B}$ is upper bounded by 
\[
k\cdot \max\{1, \log_{\beta} n^{\gamma+1} \}\leq k + \frac{k(\gamma+1)}{\log\beta} \cdot \log n = k + \log n.
 \]
 Thanks to this, there are at most $k + \log n$ vertices in $G/\mathcal{B}$, and a sparsest cut of $G/\mathcal{B}$ can be found in $O(n)$ time by enumerating all of its possible subsets; as such we can set $\alpha=1$ in the analysis. We highlight that this is another advantage of our bucketing step: with careful choice of parameters, we only need to study a contracted graph with $O(\log n)$ vertices, whose   sparsest cut can be found in linear time. 
 Since the \textsf{WRSC} algorithm only  computes the   vertex-weighted sparsest cut $O(\log n)$ times,  the overall running time of \textsf{WRSC} is $O(n \cdot \log n)$.
 We remark that  
 we don't need to deal with the general sparsest cut problem,  
 for which most approximation algorithms are based on complicated  optimisation techniques~(e.g., \cite{ARV09}).

 Combining  \eqref{eq:sum_cost_degree_buckets_nodegree} with Lemmas~\ref{lemma:upper bound inner nodegree}, \ref{lemma:upper bound WOPT nodegree} as well as the fact that $\alpha=1$ proves the approximation guarantee of Theorem~\ref{thm:main1}.

Next, we sketch the running time analysis  of Algorithm~\ref{alg:spectral_clustering_degree_bucketing}. The first step (Line~\ref{alg:specWRSC:line:step1}) applies spectral clustering, and takes $\widetilde{O}(m)$ time~\cite{peng_partitioning_2017}. The second step (Lines \ref{alg:specWRSC:line:step2_start}--\ref{alg:specWRSC:line:step2_end}) is a degree-based bucketing procedure for all the clusters, the time complexity of which is 
\[
\sum_{i=1}^k O\left(  |P_i|\log|P_i| \right)
  = O(n\log n).
\]
 The third step~(Lines~\ref{alg:specWRSC:line:step3_start}--\ref{alg:specWRSC:line:step3_end}) of the algorithm
 constructs  an arbitrary binary tree $\tree_{B}$ for every bucket $B \in \buckets$, and we show that this step takes $\widetilde{O}(m)$ time. The actual construction of the trees in this step isn't difficult, but we need to store several attributes for each internal node as we build the tree, such that  we can compute Dasgupta's cost in nearly-linear time as well. We refer the reader to the appendix for the  detailed discussion of this step. In the last step~(Lines~\ref{alg:specWRSC:line:step4_start}--\ref{alg:specWRSC:line:step4_end}), the algorithm merges the trees $\tree_{B}$ based on the output of \textsf{WRSC}, which can be implemented in   $O(n \log n)$ time as discussed above. Combining all these steps gives us the nearly-linear time complexity of  Algorithm~\ref{alg:spectral_clustering_degree_bucketing}.

\subsubsection{Proof Sketch of Lemma~\ref{lemma:upper bound inner nodegree} and Lemma~\ref{lemma:upper bound WOPT nodegree}} 
Our key approach to breaking the running time barrier of \cite{manghiuc_sun_hierarchical} is that,  instead of approximating the optimal tree $\tree^*$, we approximate the tree $\tree_{\textsf{MS}}$ constructed by the algorithm of \cite{manghiuc_sun_hierarchical}. We know from their result that one can upper bound the cost of $\treeCC$ with respect to $\OPT_G$, i.e., 
\begin{equation*}
\COST_G(\treeCC) = O\lp k^{22}/\lambda_{k+1}^{10}\rp \cdot \OPT_G.
\end{equation*}
We take the existence of the tree $\treeCC$ for granted, and perform two transformations to construct a tree $\treeCC''$, while ensuring that  the total introduced cost from the two transformations can be upper bounded. In the first step, we identify how every bucket $B \in \buckets$ is spread throughout $\treeCC$, and make sure that all the separate components that make up $B$ are isolated in our new tree. We call the resulting tree $\treeCC'$,  and prove that 
\begin{equation*}
\COST_G\lp\treeCC'\rp \leq \COST_G\lp\treeCC\rp + O(k^{21}/\lambda_{k+1}^{10}) \cdot \OPT_G.
\end{equation*}

For the second transformation, we carefully adjust the tree $\treeCC'$, such that the currently isolated components that make up every bucket $B \in \buckets$ get grouped together into the same subtree.  We call the resulting tree $\treeCC''$, and   bound its cost by
\begin{equation*}
    \COST_G\lp\treeCC''\rp \leq \COST_G\lp\treeCC'\rp + O(\beta \cdot k^{23}/\lambda_{k+1}^{10}) \cdot \OPT_G.
\end{equation*}
Taking these transformations together, we get that
\begin{equation}\label{eq:bound_treeprimeprime}
\COST_G\lp\treeCC''\rp \leq O(\beta \cdot k^{23}/\lambda_{k+1}^{10}) \cdot \OPT_G.
\end{equation}

Notice that we can perform these two transformations without an explicit construction of $\treeCC$, but we still end up with a tree $\treeCC''$ with bounded Dasgupta cost.  Moreover, since   every bucket $B \in \buckets$  in $\tree_{\textsf{MS}}''$ is separated,  $\treeCC''$ is also a tree on the contracted graph $G / \buckets$; this  allows us to upper bound $\WOPT_{G / \buckets}$ with $\COST_G(\treeCC'')$. Combining this with  \eqref{eq:bound_treeprimeprime}  proves  Lemma~\ref{lemma:upper bound WOPT nodegree}. 
With the same approach, we prove that
  Lemma~\ref{lemma:upper bound inner nodegree} 
holds as well.

We remark that this is another conceptual    contribution of our paper:  although \cite{manghiuc_sun_hierarchical} shows that a similar proof technique can be generalised to construct $O(1)$-approximate \HC\ trees for \emph{expander} graphs, we show that such  proof technique can be used to obtain $O(1)$-approximate \HC\ trees for \emph{well-clustered} graphs.

\subsection{The Second Algorithm}\label{sec:second_result}

In this subsection, we present another nearly-linear time hierarchical clustering algorithm for well-clustered graphs. Here, we assume that the degrees of vertices inside every cluster are \emph{almost balanced}, i.e., it holds for an optimal partitioning $S_1, \ldots S_k$ corresponding to $\rho(k)$ that the degrees inside each $S_i$ are upper bounded by a parameter $\degfracS \in \mathbb{R}^+$. With this condition, our second algorithm achieves the same approximation guarantee as Algorithm~\ref{alg:spectral_clustering_degree_bucketing} even with $k=O(\log^cn)$ for some constant $c$. This result is summarised as follows: 
\begin{theorem}\label{thm:main2}
 Let $G=(V,E,w)$ be a   graph with $k$ clusters $\{S_i\}_{i=1}^k$ such that $\lambda_{k+1} = \Omega(1)$, $\max_i (\Delta_G(S_i)/\delta_G(S_i)) \leq \eta_S$, $\rho(k) \leq k^{-4}$, and $\Phi_{G[S_i]} \geq \Omega(1)$ for $1 \leq i \leq k$. Then, there is a nearly-linear time algorithm that constructs an \textsf{HC} tree $\TSCD$ of 
$G$ satisfying $\COST_G(\TSCD) = O\lp 1 \rp \cdot \OPT_G$.

\end{theorem}

 The algorithm behind Theorem~\ref{thm:main2} is similar with Algorithm~\ref{alg:spectral_clustering_degree_bucketing}, and is described in   Algorithm~\ref{alg:spectral_degree_clustering}. However, our second algorithm has two significant changes compared with Algorithm~\ref{alg:spectral_clustering_degree_bucketing}. 
\begin{enumerate}
    \item We adjust the bucketing step by setting the bucketing parameter $\beta=\degfracS$, where $\degfracS$ is an upper bound for $\max_i \left(\dmax_G(S_i) / \dmin_G(S_i)\right)$. Moreover, instead of bucketing the vertices starting at the vertex $u^{(1)}$ of minimum degree inside $P_i$, we carefully choose for each $P_i$ the vertex $u_i^* \triangleq \argmax_{u \in P_i} \Vol{\bkt{u}}$ whose induced bucket $\bkt{u_i^*}$ has the highest volume inside $P_i$. Based on this, we set  $\bkting{u_i^*} \triangleq \left\{B^j\left(u_i^*\right)\right\}_j$ and $\mathcal{B} \triangleq \bigcup_{i=1}^k \bkting{u_i^*}$.
    \item After constructing the balanced trees $\tree_B$ for every $B \in \mathcal{B}$, instead of applying \textsf{WRSC},   we concatenate the trees $\tree_B$ in a simple ``caterpillar'' style according to the sizes $|B|$ of the buckets $B \in \mathcal{B}$.
\end{enumerate}

\begin{algorithm}
    \DontPrintSemicolon
    
    \KwData{$G=(V,E,w)$, $k\in\mathbb{N}$, $\degfracS \in \mathbb{R}^+$}
    \KwResult{$\TSCD$}
    
    $P = \{P_1,\ldots,P_k\} \gets \mathsf{SpectralClustering}(G,k)$\;
    $\beta \leftarrow \degfracS$\;
    Initialize $\mathbb{T} \leftarrow \emptyset$\;
    
    \For{$P_i \in P$}
    {
        Order the vertices of $P_i$ increasingly with respect to their degrees\;
        $\bkt{u} \leftarrow \left\{v \in P_i : d_u \leq d_v < \beta \cdot d_u \right\}$\;
        $u_i^* \leftarrow \argmax_{u \in P_i} \vol(B(u))$\label{alg:line:choosing ui*}\;
        $ \bkting{u_i^*} \leftarrow \left\{B^j\left(u_i^*\right)\right\}_j$\;
    
        \For{$B \in \bkting{u_i^*}$}
        {
            Let $\tree_{B}$ be any balanced $\HC$ tree on $G[B]$\;
            $\mathbb{T} \leftarrow \mathbb{T} \cup \tree_{B}$\;
        }
    }
    
    Let $\mathbb{T} = \{\tree_{1}, \dots, \tree_{|\mathbb{T}|}\}$ be such that $|\tree_i| \leq |\tree_{i+1}|$, for all $1 \leq i < |\mathbb{T}|$\;
    
    Initialize $\TSCD \leftarrow \tree_{1}$\label{alg:line:begin caterpillar}\;
    
    \For{ $i \leftarrow 2, \dots, |\mathbb{T}|$}
    {
        $\TSCD \leftarrow \TSCD \vee \tree_{i}$\label{alg:line:end caterpillar}\;
    }
    
    \Return $\TSCD$\;
    
    \caption{Spectral Clustering with Degree Bucketing and Caterpillar Construction}\label{alg:spectral_degree_clustering}
\end{algorithm}

To explain the first change, notice that every cluster $P_i$ returned by spectral clustering has a large overlap with its corresponding optimal cluster $S_i$. Moreover, the degrees of the vertices inside $S_i$ are within a factor of $\degfracS$ of each other.
Therefore, if an arbitrary vertex $u \in P_i$ is chosen as representative,
the bucketing $\bkt{u}$ of $P_i$ might equally divide the vertices in $P_i \cap S_i$ into two consecutive buckets, which is undesirable due to the high induced cost of the crossing edges.
To circumvent this issue, we   choose $u_i^*$ as the vertex to induce the bucketing and prove that this specific choice of $u_i^*$ guarantees that the bucketing $\bkting{u_i^*}$ contains one bucket $B \in \bkting{u_i^*}$ largely overlapping with $S_i$. This greatly reduces the number of crossing edges in our constructed  \HC\ tree and hence improves the approximation guarantee.

To explain the second change, we notice  that the nearly-linear running time of  \textsf{WRSC} relies on the condition that $k=O(1)$, which might not hold necessarily for the new setting. We prove that, as long as  $\max_i \left(\dmax_G(S_i) / \dmin_G(S_i)\right)$ is upper bounded by a constant, it is sufficient to construct a caterpillar tree based on the sizes $|B|$ of the buckets $B \in \mathcal{B}$. We prove that our algorithm runs in nearly-linear time, and the output tree achieves $O(1)$-approximation. See the appendix for the detailed analysis of  Algorithm~\ref{alg:spectral_degree_clustering} and  the proof of Theorem~\ref{thm:main2}

\section{Experiments}\label{sec:experiments}
 We experimentally evaluate the performance of our designed \texttt{SpecWRSC}  algorithm, and compare it against  the previous state-of-the-art and a well-known  linkage algorithm. Specifically, the   \texttt{SpecWRSC} algorithm is compared against the following three algorithms:
 \begin{enumerate}
\item the \texttt{Linkage}$++$ algorithm proposed in~\cite{CAKMT17};
\item the \texttt{MS}  algorithm proposed in~\cite{manghiuc_sun_hierarchical};  
\item  the \texttt{AverageLinkage}  algorithm.
 \end{enumerate}
 Even though there are no   theoretical approximation guarantees of \texttt{AverageLinkage} with respect to Dasgupta's cost function, we include it in our   evaluation for reference due to its excellent performance in practice.

All the tested  algorithms were implemented in Python 3.9 and experiments were performed using a Lenovo ThinkPad T15G, with an Intel(R) Xeon(R) W-10855M CPU@2.80GHz processor and 126 GB RAM. All of the reported costs and running time below are averaged over $5$ independent runs. Our code can be downloaded from \href{https://github.com/SteinarLaenen/Nearly-Optimal-Hierarchical-Clustering-for-Well-Clustered-Graphs}{https://github.com/steinarlaenen/nearly-optimal-hierarchical-clustering-for-well-clustered-graphs}.

\subsection{Results on Synthetic Data} 

We first compare the  performance of our algorithm against the others described before on synthetic data. 

\subsubsection{Stochastic Block Model} We look at graphs generated  according to the standard Stochastic Block Model (\textsf{SBM}).
We set the number of clusters as $k=5$, and the number of vertices in each cluster $\{P_i\}_{i=1}^5$ as $n_k$. For each pair of vertices $u \in P_i$ and $v \in P_j$, we add an edge $\{u,v\}$ with probability $p$ if $i=j$, and add an edge with probability $q$ if $i \neq j$. 

To compare the running time, we incrementally increase the number of vertices in each cluster $\{P_i\}_{i=1}^5$ from $n_k = 100$ to $n_k = 3{,}000$, in increments of 200, and we fix $p=0.1$ and $q = 0.002$. For each algorithm, we measure the running time in seconds, and the results are plotted in Figure~\ref{fig:running time comparison}(a), confirming our algorithm's 
low time complexity compared with our competitors.
 To further highlight the fast running time of our algorithm, we increase the size of each cluster to $20{,}000$, which results in 
  the total number of nodes $n=100{,}000$. We find that \textsf{SpecWRSC} returns an \HC\ tree in $14{,}437$ seconds ($\sim 4$ hours), while all the other algorithms timeout after 12 hours of compute time.

To compare the quality of the constructed \HC\ trees on \textsf{SBM}s, we fix the number of vertices inside each cluster $\{P_i\}_{i=1}^5$ as $n_k = 1{,}000$,    the value $q=0.002$, and consider different values of $p\in[0.06, 0.2]$. As shown in Figure~\ref{fig:running time comparison}(b), \texttt{SpecWRSC} performs better than all the other algorithms.

\begin{figure}
\centering
\begin{minipage}{0.4\textwidth}
\centering
    \includegraphics[width=6cm]{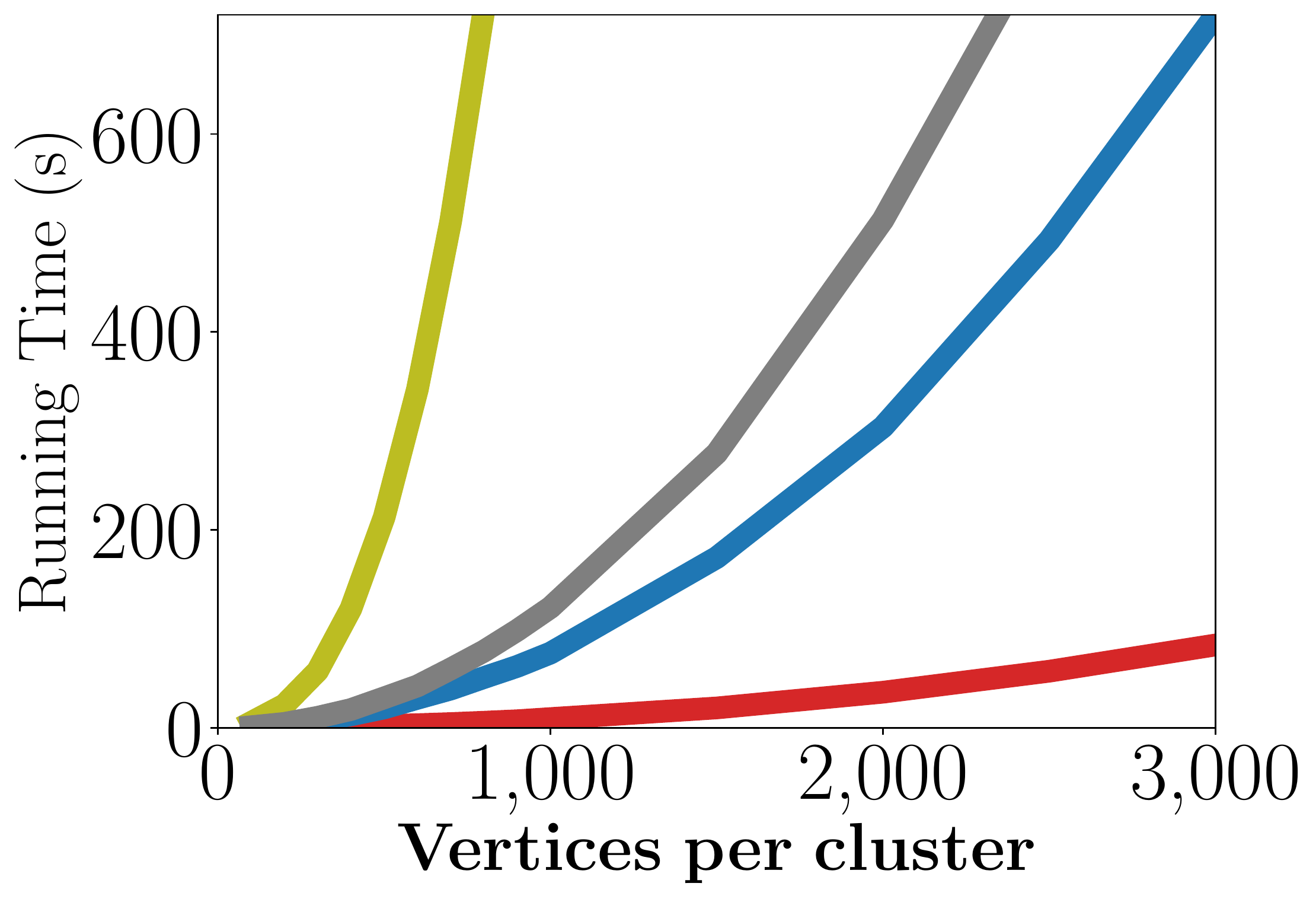}
    \centering
    \small{(a)}
\end{minipage}%
\begin{minipage}{0.4\textwidth}
\centering
    \includegraphics[width=6cm]{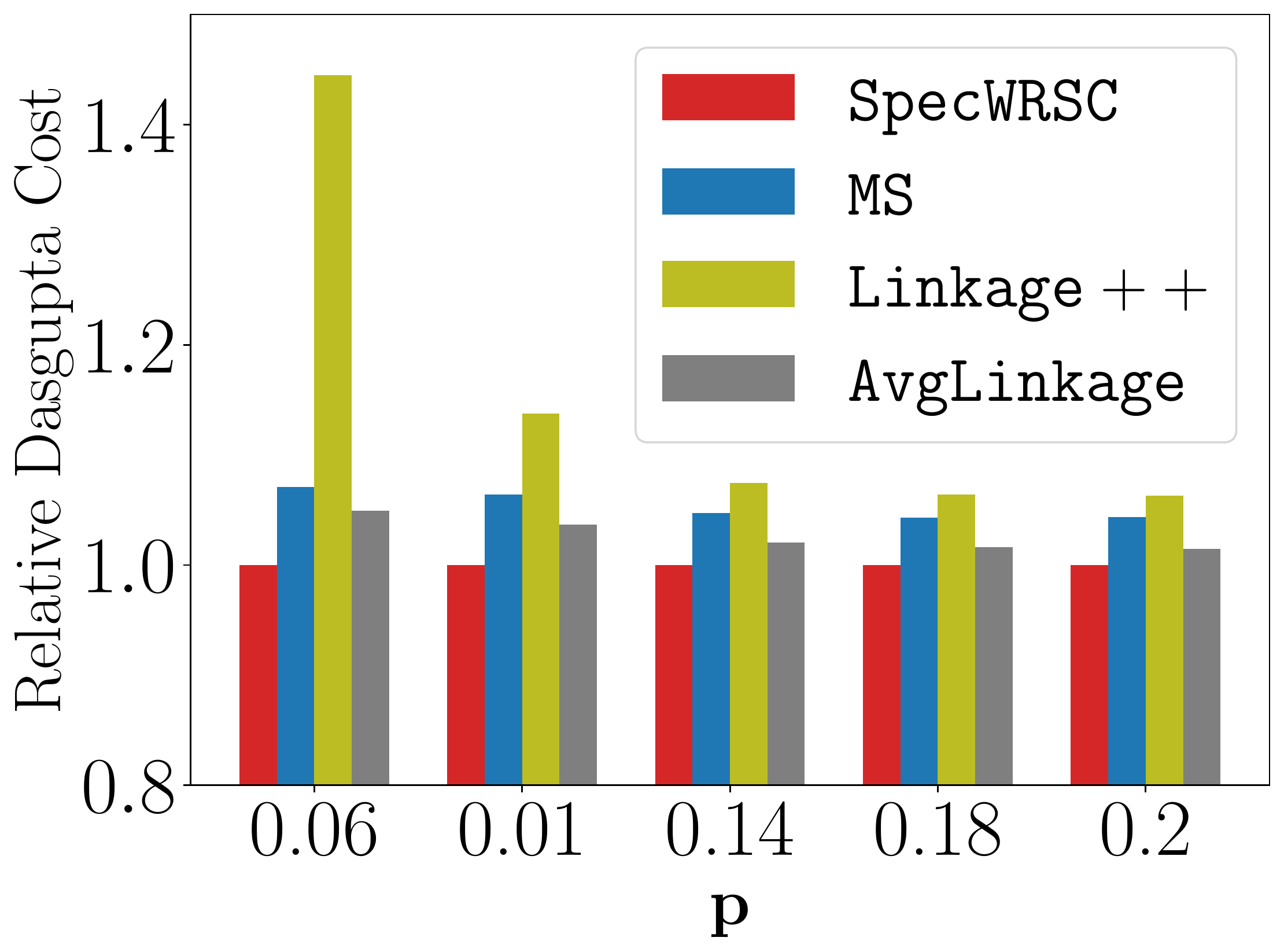}
    \centering
    \small{(b)}
\end{minipage}
\caption{Results for \textsf{SBM}s. In Figure~(a) the $x$-axis represents the number of vertices inside each cluster, and the  $y$-axis represents the algorithms' running time in seconds. In Figure~(b), the $x$-axis represents different values of $p$, while the $y$-axis represents the cost of the algorithms' returned \textsf{HC} trees normalised by the cost of \texttt{SpecWRSC}.}\label{fig:running time comparison} 
\end{figure}

\subsubsection{Hierarchical Stochastic Block Model}
Next, we consider graphs generated according to a hierarchical stochastic block model \textsf{(HSBM)} \cite{cohen2018hierarchical}. \textsf{HSBM} is similar to  the \textsf{SBM} but   assumes the existence of a hierarchical structure between the clusters. We set the number of vertices in each cluster $\{P_i\}_{i=1}^5$ as $n_k = 600$. For each pair of vertices $u \in P_i$ and $v \in P_j$,  we assume that $u$ and $v$ are connected by an edge with probability $p$ if $i = j$;  otherwise $u$ and $v$ are connected by an edge with probability $q_{i, j}$ defined as follows: (i) for all $i \in \{1, 2, 3\}$ and $j \in \{4, 5\}, q_{i, j} = q_{j, i} = q_{\mathrm{min}}$; (ii) for $i \in \{1, 2\}$, $q_{i, 3} = q_{3, i} = 2 \cdot q_{\mathrm{min}}$; (iii) $q_{4, 5} = q_{5, 4} = 2 \cdot q_{\mathrm{min}}$; (iv) $q_{1, 2} = q_{2, 1} = 3 \cdot q_{\mathrm{min}}$. We fix the value $q_{\mathrm{min}} = 0.0005$ and consider different values of $p \in [0.04, 0.2]$. This choice of hyperparameters bears similarity with the setting tested in~\cite{CAKMT17, manghiuc_sun_hierarchical}. 
The result of our experiments is plotted in Figure~\ref{fig:hierarchical_comparison}(a). Again, we find that our algorithm outperforms \texttt{MS} and \texttt{Linkage}$++$, and  has similar performance as  \texttt{AverageLinkage}. 

\subsection{Result  on Real-World Data} 
To evaluate the performance of our algorithm on real-world datasets, we first  follow the sequence of recent work on hierarchical clustering~\cite{abboud2019subquadratic,CAKMT17, menon2019online,manghiuc_sun_hierarchical,roy2017hierarchical}, all of which are based on the following 4 datasets from the Scikit-learn library \cite{scikit-learn} and the UCI ML repository~\cite{UCIML}: Iris (\textsf{Ir}), Wine (\textsf{Wi}), Cancer (\textsf{Ca}), and Newsgroup (\textsf{Ng})\footnote{Due to the large size of NewsGroup, we consider only a subset consisting of ``comp.graphics'', ``comp.os.ms-windows.misc'', ``comp.sys.ibm.pc.hardware'', ``comp.sys.mac.hardware'', ``rec.sport.baseball'', and ``rec.sport.hockey''.}. For each dataset, we construct the similarity graph based on the Gaussian kernel, in which the $\sigma$-value is chosen according to the standard heuristic~\cite{nips/NgJW01}.  The setup  of parameters for each dataset is listed in Table~\ref{tab:real_world_params}.

The cost of the \HC\ trees returned by each algorithm is reported in Figure~\ref{fig:hierarchical_comparison}(b), and the figure shows that 
our algorithm performs better than \texttt{MS} and \texttt{Linkage}$++$ and matches the performance of \texttt{AverageLinkage}. We further report the running time of the tested algorithms in Table~\ref{tab:real_world_params}, and this shows that \texttt{SpecWRSC} has the lowest running time among the four tested algorithms. Moreover,  we observe that \texttt{SpecWRSC} performs better when the spectral gap $\lambda_{k+1} / \lambda_k$ is large, which is in line with our theoretical analysis.

\begin{table}[t]
\begin{center}
\begin{tabular}{@{}lllll@{}}
\toprule
                       & \textsf{Ir}      & \textsf{Wi}    & \textsf{Ca}      & \textsf{Ng}        \\ 
\textit{Parameters}    &         &        &         &           \\ \midrule
$n$      & $50$    & $180$  & $558$   & $3{,}516$    \\

$k$    & $3$     & $5$    & $5$     & $5$       \\ 
$\sigma$                  & $0.3$   & $0.88$ & $0.88$  & $130$     \\
$\lambda_{k+1} / \lambda_k$ & 5.27 & 4.96 &  1.45 & 1.01 \\

\midrule
\textit{Running Time} &         &        &         &           \\ \midrule
\texttt{AvgLinkage}                     & $0.14$  & $0.22$ & $2.49$  & $198.34$  \\
\texttt{MS}                     & $0.88$  & $2.00$ & $20.36$ & $1{,}088.87$ \\
\texttt{Linkage++}              & $0.61$  & $0.90$ & $9.18$  & $944.05$  \\
\texttt{SpecWRSC}               & $\mathbf{0.136}$ & $\mathbf{0.16}$ & $\mathbf{1.69}$  & $\mathbf{128.17}$ 
\end{tabular}
\end{center}
\caption{Setup of the parameters for the tested real-world datasets, and the running time of the tested algorithms. For each dataset, we list the number of nodes $n$, the number of clusters $k$,   the $\sigma$-value used to construct the similarity graph based on the Gaussian kernel, and the spectral gap  $\lambda_{k+1} / \lambda_k$ of the constructed  normalised Laplacian. The running time is  reported in seconds, and the  lowest one is  highlighted.}\label{tab:real_world_params}
\end{table}

To further  demonstrate the effectiveness of our algorithm on larger real-world datasets, we evaluate our algorithm on the GEMSEC Facebook dataset~\cite{rozemberczki2019gemsec}. This dataset represents a page-page graph of verified Facebook pages, in which every vertex
 represents an  official Facebook page, and every edge represents the  mutual likes between pages. We focus on the largest subset of this dataset, i.e.,  the artist category. This graph contains $50{,}515$ nodes and $819{,}306$ edges. Out of all the tested algorithms, $\texttt{SpecWRSC}$ is the only one that terminates within 12 hours of compute time. Setting $k=20$, \texttt{SpecWRSC} returns a tree in $1{,}223$ seconds, with a Dasgupta cost of $1.62 \cdot 10^{10}$.

\begin{figure}
\centering
\begin{minipage}{0.4\textwidth}
\centering
    \includegraphics[width=6cm]{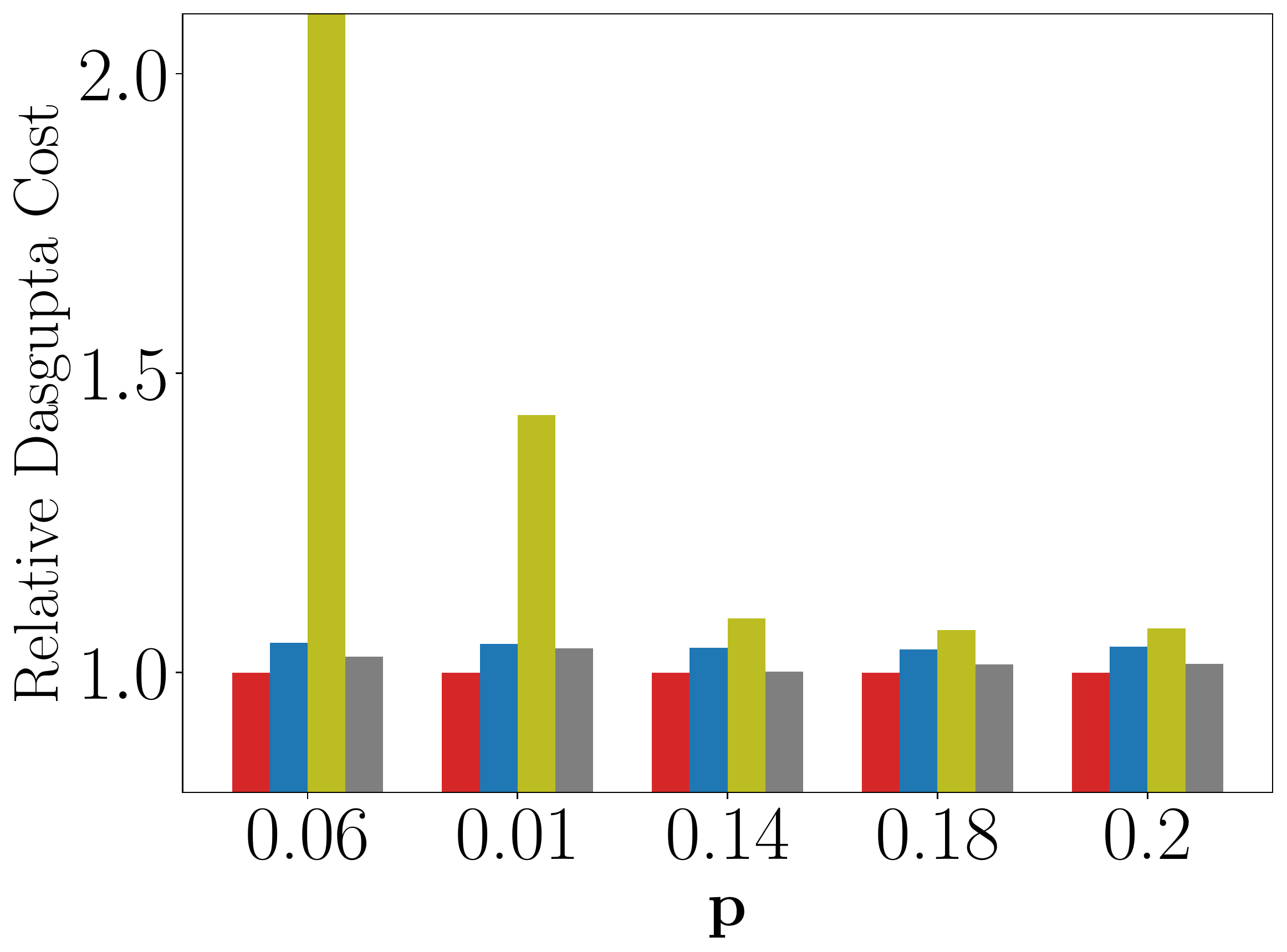}
    \centering
    \small{(a)}
\end{minipage}%
\begin{minipage}{0.4\textwidth}
\centering
    \includegraphics[width=6cm]{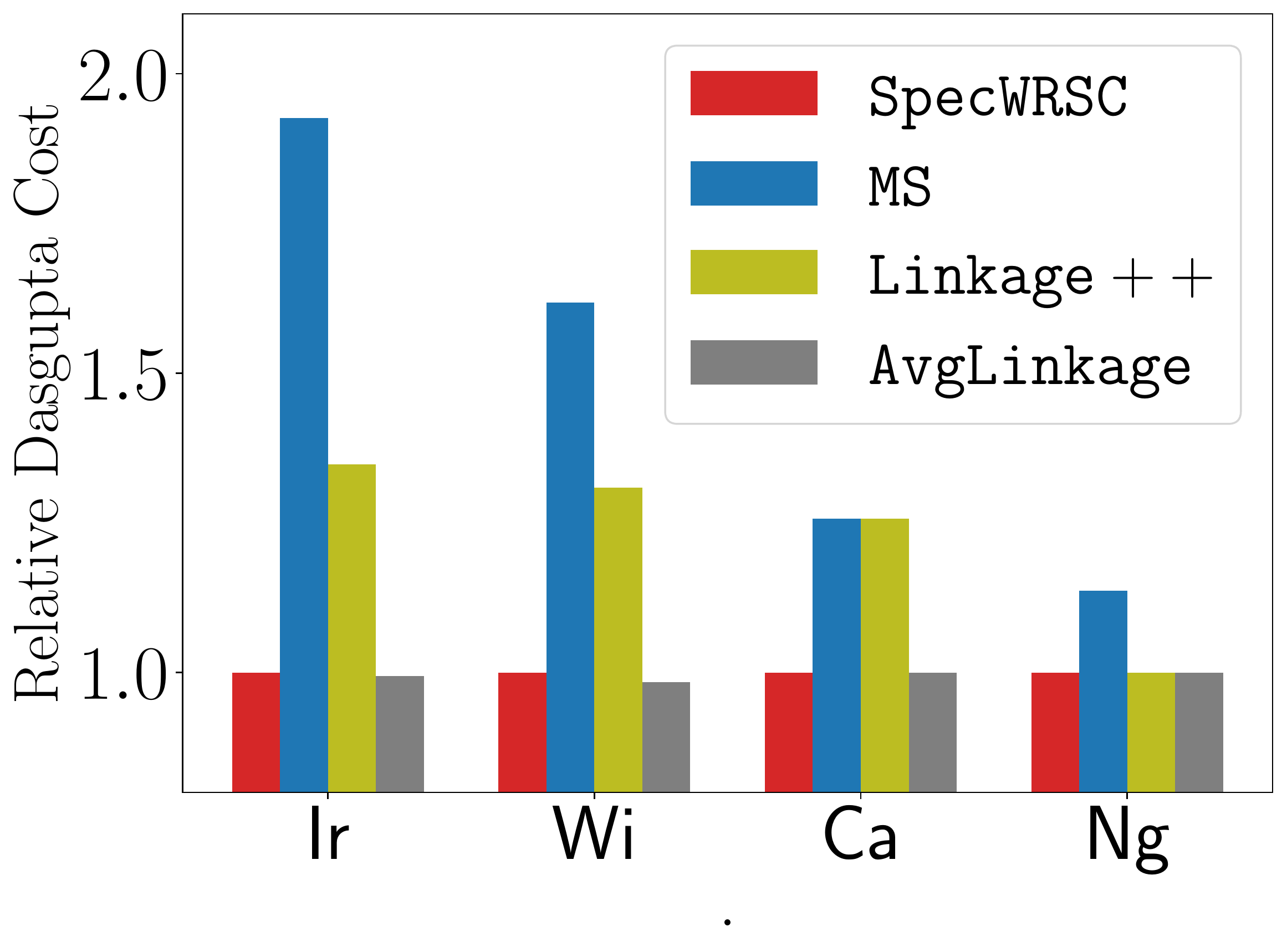}
    \centering
    \small{(b)}
\end{minipage}
\caption{(a) Results for \textsf{HSBM}s, and (b) results on the real-world datasets. In Figure~(a) the $x$-axis represents different values of $p$, while the $y$-axis represents the cost of the algorithms' returned \textsf{HC} trees normalised by the cost of \texttt{SpecWRSC}. In Figure~(b) the $x$-axis represents the different real-world datasets we evaluate on.} \label{fig:hierarchical_comparison}
\end{figure}

These  experimental results together  demonstrate that our designed algorithm not only has excellent running time but also constructs hierarchical clustering with a cost lower than or similar to the ones constructed by the previous algorithms.

\bibliographystyle{alpha}

\bibliography{main.bib}

\appendix

\section{Additional Background}\label{sec:additional_background}

 In this section we list additional background notion and facts used in our analysis, and the section is organised as follows: In Section~\ref{appendix:background_hc} we introduce additional background on hierarchical clustering. In Section~\ref{appendix:background_spectralclustering} we provide a more detailed discussion on spectral clustering, and finally in Section~\ref{appendix:background_strong_decomp_lemma} we present and prove a new decomposition lemma of well-clustered graphs.

\subsection{Hierarchical Clustering}\label{appendix:background_hc}

We first examine the upper and lower bounds of the cost of \HC\ trees.  The first fact holds trivially by the fact that every edge's contribution towards the cost function is at most $n$.

\begin{fact}\label{fact:trivial_upper_bound}
It holds for any \textsf{HC} tree $\tree$ of $G$ that
\[
    \COST_G(\tree) \leq \frac{n \cdot \Vol{G}}{2}.
\]
\end{fact}

Secondly, it is known that $\OPT_G$ can be lower bounded with respect to the degree distribution of $G$.

\begin{lemma}[\cite{manghiuc_sun_hierarchical}]\label{lem:Lower Bound Cost Expander}
    It holds for any  graph $G$ that
\[
    \OPT_G \geq \frac{2\Phi_G}{9} \cdot \max \left\{\frac{\vol(G)^2}{\Delta_G} , \delta_G \cdot n^2 \right\} = \frac{2\Phi_G}{9} \cdot n\cdot \vol(G) \cdot \max \left\{ \frac{d_G}{\Delta_G}, \frac{\delta_G}{d_G}\right\}.
\]
\end{lemma}

Next, we introduce the notion of a \emph{dense branch}, which is a useful tool to analyse the cost of \HC\ trees.
 \begin{definition}[Dense branch, \cite{manghiuc_sun_hierarchical}]\label{def:Dense Branch}
    Given a graph $G$ with an \textsf{HC} tree $\tree$, the \emph{dense branch} is the path $(A_0, A_1, \dots, A_k)$ in $\tree$ for some $k \in \mathbb{Z}_{\geq 0}$, such that the following hold:
    \begin{enumerate}
        \item $A_0$ is the root of $\tree$;
        \item $A_k$ is the node such that $\vol(A_k) > \vol(G)/2$ and both of its children have volume at most $\vol(G)/2$.
    \end{enumerate}
\end{definition}

 It is important to note that the dense branch of $\tree$ is unique, and consists of all the nodes $A_i$ with $\vol(A_i) > \vol(G)/2$. Moreover, for every pair of consecutive nodes $A_i, A_{i+1}$ on the dense branch, $A_{i+1}$ is the child of $A_i$ of the \emph{higher} volume.

 The next lemma presents a lower bound on the cost of a tree $\tree$ using the dense branch, which will be extensively used   in our analysis.
 
 \begin{lemma}[Lower bound of $\COST_G(\tree)$ based on the dense branch,  \cite{manghiuc_sun_hierarchical}]\label{lem:Cost light nodes}
    Let $G$ be a graph of conductance $\Phi_G$, and   $\tree$   an arbitrary \textsf{HC} tree of $G$. Suppose $(A_0, \dots, A_k)$ is the dense branch of $\tree$ for some $k \in \mathbb{Z}_{\geq 0}$, and suppose each node $A_i$ has sibling $B_i$, for all $1 \leq i \leq k$. Then,
    the following lower bounds of 
    $\COST_G(\tree)$ hold:
    \begin{enumerate}
        \item $\COST_G(\tree) \geq \frac{\Phi_G}{2}\sum_{i=1}^k |A_{i-1}| \cdot \vol(B_i)$;
        \item $\COST_G(\tree) \geq \frac{\Phi_G}{2} \cdot |A_k| \cdot \vol(A_k)$.
    \end{enumerate}
\end{lemma}

One of the two main results from \cite{manghiuc_sun_hierarchical} is an $O(1)$-approximation on Dasgupta's cost for expander graphs. Their result is stated in the following theorem.
 
 \begin{theorem}[\cite{manghiuc_sun_hierarchical}]\label{thm:degree}
Given any graph $G=(V,E,w)$ with  conductance $\Phi_G$ as input, there is an algorithm that runs in $O(m + n \log n)$ time, and returns an \HC\ tree $\T_{\deg}(G)$ of $G$ that satisfies 
$\COST_G(\T_{\deg}) = O\left(1/\Phi_G^4\right)\cdot\OPT_G$.
\end{theorem}

 Throughout our discussion, we always use $\T_{\deg}(G)$ to represent the \HC\ tree of $G$ that is constructed from Theorem~\ref{thm:degree}, i.e., the tree constructed based on the degree sequence of $G$. In particular, for any partition $\{C_1,\ldots, C_{\ell}\}$, one can apply Theorem~\ref{thm:degree} to every induced graph $G[C_i]$ for any $1\leq i\leq \ell$, and for simplicity we write $\T_i\triangleq \T_{\deg}( G[C_i] )$  in the following discussion. We next define critical nodes with respect to every $\T_i$.

 \begin{definition}[Critical nodes,  \cite{manghiuc_sun_hierarchical}]\label{def:Critical Nodes}
        Let $\T_i = \T_{\deg}(G[C_i])$ be  defined as above for any  non-empty subset $C_i \subset V$. Suppose $(A_0, \dots, A_{r_i})$ is the dense branch of $\T_i$ for some $r_i \in \mathbb{Z}_{+}$, $B_j$ is the sibling of $A_j$, and let $A_{r_i+1}, B_{r_i+1}$ be the two children of $A_{r_i}$. We define $\mathcal{S}_i \triangleq \{B_1, \dots, B_{r_i + 1}, A_{r_i + 1}\}$ to be the set of \emph{critical nodes} of $C_i$. Each node $N \in \mathcal{S}_i$ is a \emph{critical node}.
\end{definition}
 We remark that each critical node $N \in \mathcal{S}_i~(1\leq i\leq \ell)$ is an internal node of maximum size in $\T_i$ that is not in the dense branch. Moreover, every $\mathcal{S}_i$ is a partition of $C_i$.

 The following lemma lists some  facts of the  critical nodes of the trees constructed from  Theorem~\ref{thm:degree}.
 
  \begin{lemma}[\cite{manghiuc_sun_hierarchical}]\label{lem:critical node properties}
     Let $\mathcal{S}_i = \{B_1, \dots, B_{r_i + 1}, A_{r_i + 1}\}$ be the set of critical nodes of $\T_i$. Then the following statements hold:
\begin{enumerate}[label=(Q\arabic*)]
    \item  $|B_j| = 2 \cdot|B_{j+1}|$ for all $j\geq2$, 
    \item  $\vol_{G[C_i]}(B_j) \leq 2 \cdot \vol_{G[C_i]}(B_{j+1})$ for all $j\geq1$, 
    \item  $|A_{i_{\max}}| = |B_{i_{\max}}|$.
\end{enumerate}
 \end{lemma}

For convenience in our notation, we also define a \emph{critical sibling} node of critical nodes.
 \begin{definition}[Critical Sibling Node]\label{def:sibling}
    Let $\mathcal{S}_i = \{B_1, \dots, B_{r_i + 1}, A_{r_i + 1}\}$ be the set of critical nodes of $\T_i$. We define a \emph{sibling node} of a critical node $B_j$ as $\sibling_{\T_i}(B_j) \triangleq B_{j+1}$ for $1 \leq j \leq r_i$, and $\sibling_{\T_i}(B_{r_i+1}) \triangleq A_{r_i+1}$ otherwise.
\end{definition}
 We remark that the only critical node without a sibling is $A_{r_i+1}$. Furthermore, due to the degree based construction in Theorem~\ref{thm:degree},  it holds for all the other critical nodes $N \in \mathcal{S}_i \setminus A_{r_i+1}$  that $\Vol{N} \leq 2 \cdot \Vol{\sibling_{\T_i}(N)}$ and $|N| \leq 2 \cdot |\sibling_{\T_i}(N)|$.

 Finally, the following inequality will be extensively in our proof of Theorem~\ref{thm:main2}:

\begin{lemma}
\label{lem:lower_bound_sum_SiVolSi}
Let $G=(V,E, w)$ be a graph with a partition $S_1,\ldots,S_k$, such that $\phiout \leq 1/2$ and $\Phi_{G[S_i]}\geq \Phi_{\mathrm{in}}$ for any $1\leq i \leq k$.   Then it holds that
    \[
    \sum_{i=1}^k |S_i| \cdot \Vol{S_i} \leq \frac{18 \cdot \degfracS}{\phiin} \cdot \OPT_G,
    \]
where $\degfracS$ is an upper bound of $\max_i (\Delta(S_i)/\delta(S_i))$.
\end{lemma}

\begin{proof}
We can trivially lower bound $\OPT_G$ by only considering the internal edges $e \in S_i$ of the individual clusters $G[S_i]$, and then applying the lower bound from Lemma~\ref{lem:Lower Bound Cost Expander}:
\begin{align*}
    \OPT_G &\geq \sum_{i=1}^k \sum_{e \in E[G[S_i]]} \cost_{G[S_i]}(e) \geq \sum_{i=1}^k \OPT_{G[S_i]} \nonumber \\
    &\geq \sum_{i=1}^k \frac{2 \cdot \Phi_{G[S_i]}}{9} \cdot |S_i| \cdot \Vol{G[S_i]} \cdot \frac{\davg_{G[S_i]}}{\dmax_{G[S_i]}} \nonumber\\
    &\geq \frac{2 \cdot \phiin}{9} \cdot \sum_{i=1}^k |S_i| \cdot \frac{\Vol{S_i}}{2} \cdot \frac{1}{2 \cdot \degfracS} \nonumber.
\end{align*}
Here, the second inequality follows by Lemma~\ref{lem:Lower Bound Cost Expander} and the last inequality holds because of the assumptions in the Lemma and the fact that \[\Vol{S_i} = \Vol{G[S_i]} + w(S_i,  V\setminus S_i) \leq 2 \cdot \Vol{G[S_i]}\] when $\phiout \leq 1/2$, and because \[
\frac{\davg_{G[S_i]}}{ \dmax_{G[S_i]}} \geq \frac{\davg_{G[S_i]}}{\dmax(S_i)} \geq \frac{\delta(S_i)}{ \dmax(S_i) \cdot 2} \geq \frac{1}{2\cdot \degfracS}.\]
Rearranging the terms above proves the statement.
\end{proof}

\subsection{Spectral Clustering}\label{appendix:background_spectralclustering}

Another component used in our analysis is spectral clustering. To analyse the theoretical performance of spectral clustering, one can examine the scenario in which there is a large gap between $\lambda_{k+1}$ and $\rho(k)$. By the higher-order Cheeger inequality~\eqref{eq:Higher Cheeger}, we know that   a low value of $\rho(k)$ ensures that the vertex set $V$ of $G$ can be partitioned into $k$ subsets~(clusters), each of which has conductance upper bounded by $\rho(k)$; on the other hand, a large value of $\lambda_{k+1}$ implies that  any $(k+1)$-th partition of $V$ would introduce some $A\subset V$ with conductance $\Phi_G(A)\geq \rho(k+1)\geq \lambda_{k+1}/2$.  
 Based on this,  Peng et al.~\cite{peng_partitioning_2017} introduced the parameter 
\begin{equation}\label{def:Upsilon}
\Upsilon(k) \triangleq \frac{\lambda_{k+1}}{\rho(k)}
\end{equation}
and showed that a large value of $\Upsilon(k)$ is sufficient to guarantee a good performance of   spectral clustering. The result of \cite{peng_partitioning_2017} has been improved by a sequence of works, and the following result, which can be  shown easily by combining the proof technique of \cite{peng_partitioning_2017} and \cite{MS22}, will be used in our analysis.

\begin{lemma}\label{lem:MS22+}
There is an absolute  constant  $\CGap\in\mathbb{R}_{>0}$, such that the following holds: 
Let $G$ be a graph with $k$ optimal clusters $\{S_i\}_{i=1}^k$, and $\Upsilon(k) \geq \CGap \cdot k$. Let $\{P_i\}_{i=1}^k$ be the output of spectral clustering and, without loss of generality, the optimal correspondence of $P_i$ is $S_i$ for any $1\leq i \leq k$. Then, it holds  for any $1\leq i\leq k$ that 
\[
  \vol(P_i\triangle S_i)  
  \leq \frac{k \cdot \CGap}{3\Upsilon(k)} \cdot \vol(S_i),
\]
where $A\triangle B$ for any sets $A$ and $B$ is defined by $A\triangle B\triangleq (A\setminus B)\cup (B\setminus A)$.
 Moreover, these $P_1,\ldots, P_k$ can be computed in nearly-linear time.  
\end{lemma}

\subsection{Strong Decomposition Lemma}\label{appendix:background_strong_decomp_lemma}

 The objective of this subsection is to prove the following decomposition lemma. The lemma shows that, under a certain eigen-gap condition, an input graph $G$ can be partitioned into clusters with bounded inner and outer conductance, and certain constraints on cut values. We remark that, while obtaining the partition promised by the lemma below requires high time complexity, we only need the existence of such partition in our analysis. 

\begin{lemma}[Improved Strong Decomposition Lemma]\label{lem:Improved Decomposition}
 Let $G = (V, E, w)$ be a graph such that $\lambda_{k+1} > 0$  and $\lambda_k < \lp \frac{1}{270 \cdot c_0 \cdot (k+1)^6} \rp^2$, where $c_0$ is the constant in Lemma~\ref{lem:Upperbound eig induced graphs}. Then, there is a polynomial-time algorithm that finds an $\ell$-partition of $V$ into sets $\{C_i\}_{i=1}^{\ell}$, for some $\ell \leq k$, such that for every $1 \leq i \leq \ell$ and every vertex $u \in C_i$ the following properties hold:
\begin{enumerate}[label=(A\arabic*)]
    \item  $\Phi(C_i) = O(k^6 \sqrt{\lambda_k})$; 
    \item  $\Phi_{G[C_i]} = \Omega (\lambda_{k+1}^2/k^4)$;
    \item  $w(u, V \setminus C_i) \leq 6(k+1) \cdot \vol_{G[C_i]}(u)$.
\end{enumerate}
\end{lemma}

  We remark that the first two properties of the partitioning promised by Lemma~\ref{lem:Improved Decomposition} are the same as the  ones from  \cite{manghiuc_sun_hierarchical}. However, the third property of our lemma is stronger than theirs, as Property~(A3) now holds for \textit{all} vertices in $u \in C_i$, instead of  only the critical nodes $N \in \mathcal{S}_i$. We emphasise that this  improved decomposition is crucial for our final analysis. In particular, we only use Lemma~\ref{lem:Improved Decomposition} to show 
  the \emph{existence} of a strong decomposition with  properties (A1), (A2), and (A3), and we use the strengthened property (A3) in the analysis of the approximation factor of our algorithm.

Before starting the analysis, by convention  we set $\Phi_G(V) \triangleq 0$, and  $\Phi_G(\emptyset) \triangleq 1$. We also set $w(\emptyset, S) \triangleq  0$ for any non-empty subset $S \subset V$. The following two results will be used in our analysis.

\begin{lemma}[Cheeger Inequality, \cite{Alon/86}]\label{lem:Cheeger's ineq}
It holds for any graph $G$ that 
    $
        \frac{\lambda_2}{2} \leq \Phi_G \leq \sqrt{2\lambda_2}$.
    Furthermore, there is a  nearly-linear time algorithm (i.e., the \SpecPart algorithm) that finds a set $S$ such that $\vol(S) \leq \vol(V)/2$, and $\Phi_G(S) \leq  \sqrt{ 2\lambda_2}$.   
\end{lemma}

 \begin{lemma}[Lemma~1.13, \cite{GT14}]\label{lem:Upperbound eig induced graphs}
    There is an absolute constant $c_0 > 1$ such that for any $k\geq 2$ and any $r$-way partition $C_1,\ldots, C_r$ of $V$, where $r \leq k -1  $, we have that 
    \[
        \min_{1 \leq i \leq r} \lambda_2\left(\mathcal{L}_{G[C_i]}\right) \leq 2 c_0 \cdot k^6 \cdot \lambda_{k}.
    \]
\end{lemma} 

Now we describe the underlying algorithm and show a sequence of claims, which are used to prove Lemma~\ref{lem:Improved Decomposition}. At the very high level, our  algorithm for computing a stronger decomposition of a   well-clustered graph can be viewed as an adjustment to Algorithm~3 in \cite{manghiuc_sun_hierarchical}, which itself is based on Algorithm~3 in \cite{GT14} The main idea 
 can be summarised as follows: the algorithm starts with the trivial $1$-partition of $G$, i.e.,  $C_1 = V$; in every iteration, the algorithm applies  the \SpecPart algorithm for every graph in   $\{G[C_i]\}_{i=1}^{r}$, and tries to find a sparse cut $(S, C_i \setminus S)$ for some  $S \subset C_i$.\footnote{We  denote by $r$ the number of clusters in the current run of the algorithm, and   denote by $\ell$ the final number of clusters output by the algorithm.}   
\begin{itemize}
    \item If such a cut is found, the algorithm uses this cut to either introduce a new partition set $C_{r+1}$ of low conductance, or   refine the current  partition  $\{C_i\}_{i=1}^r$;
    
    \item  If no such cut is found, the algorithm checks if it is possible to perform a local refinement of the partition sets $\{C_i\}_{i=1}^r$ in order to reduce the overall weight of the crossing edges, i.e. $\sum_{i \neq j} w(C_i, C_j)$. If such a refinement is not possible, the algorithm terminates and outputs the current partition;  otherwise, the partition sets are locally refined and the process is repeated. 
\end{itemize}
The output of the algorithm is guaranteed to  satisfy Properties $(A1)$ and $(A2)$  of Lemma~\ref{lem:Improved Decomposition}. 

Our improved  analysis will show that  Property~$(A3)$ holds as well, and this will be proven  with  the two additional Properties  $(A4)$ and $(A5)$ stated later. 
We begin our analysis by setting the notation, most of which follows from \cite{GT14}. 
We write $\{C_i\}_{i=1}^r$ as  a partition of $V$ for some integer $r\geq 1$, and  every  partition set  $C_i$ contains some \emph{core set} denoted by $\mathrm{core}(C_i)\subseteq C_i$.  For an arbitrary subset $S \subset C_i$,  we define $$S^+ \triangleq S\cap \mathrm{core}(C_i),$$ and $$S^- \triangleq S\setminus S^+.$$ 
We further define  $$\overline{S^+}\triangleq \mathrm{core}(C_i)\setminus S,$$ and $\overline{S^-}\triangleq C_i\setminus (S \cup \mathrm{core}(C_i))$. Note that $\{S^{+}, \overline{S^{+}}\}$ forms a partition of $\mathrm{core}(C_i)$, and $\{S^{-}, \overline{S^{-}}\}$ forms a partition of $C_i \setminus \mathrm{core}(C_i)$. For the ease of presentation, we always write $u^+$ and $u^-$ if the set $S$ consists of a single vertex $u$. 
For any sets $S, T \subseteq V$ which are not necessarily disjoint,  we write 
    \[
        w(S \rightarrow T) \triangleq w(S, T \setminus S),
    \]
    For any subsets $S \subseteq C \subseteq V$, we follow \cite{GT14} and define the \emph{relative conductance} as
    \[
        \varphi(S, C) \triangleq \frac{w(S \rightarrow C)}{\frac{\vol(C \setminus S)}{\vol(C)} \cdot w(S \rightarrow V\setminus C)},
    \]
    whenever the right-hand side is defined and otherwise we set $\varphi(S, C) = 1$.
  To explain the meaning of $\varphi(S, C)$, suppose that $C\subset V$ is the vertex set such that $\Phi_G(C)$ is low and $\Phi_{G[C]}$ is high, i.e., $C$ is a cluster. Then, we know that most of the  subsets $S\subset C$ with $\vol(S)\leq \vol(C)/2$ satisfy the following properties:
\begin{itemize}
    \item Since $\Phi_{G[C]}(S)$ is high, a large fraction of the edges adjacent to vertices in $S$ would leave $S$;
    \item Since $\Phi_G(C)$ is low, a small fraction of edges adjacent to $S$ would leave $C$.
\end{itemize}
Combining the above observations, one could conclude that $w(S \rightarrow C) \gtrsim w(S \rightarrow V \setminus C)$ if $C$ is a good cluster, which means that $\varphi(S, C)$ is lower bounded by a constant. Moreover, Oveis Gharan and Trevisan~\cite{GT14} showed a converse of this fact:  if $\varphi(S, C)$ is large for all $S \subset C$, then $C$ has high inner conductance. These facts suggest that the relative conductance provides a good quantitative measure of the quality of a cluster.

\begin{algorithm}[H]
    \caption{Algorithm for partitioning $G$ into $\ell \leq k$ clusters}
    \label{alg:Strong Decomposition}
    \SetAlgoLined
    \KwIn{$G = (V, E, w)$, $k > 1$ such that $\lambda_{k+1} > 0$}
    \KwOut{A $\left( \phiIn^2/4, \phiOut \right)$ $\ell$-partition $\{C_i\}_{i=1}^{\ell}$ of $G$ satisfying $(A1) - (A3)$, for some $\ell \leq k$}
    Let $r = 1$, $\core(C_1) = C_1 = V$\;
    Let $\phiIn = \frac{\lambda_{k+1}}{140(k+1)^2}$,  and 
    $\phiOut = 90 c_0 \cdot (k+1)^6 \sqrt{\lambda_k}$\;
    Let $\rho^* = \min \left\{ \frac{\lambda_{k+1}}{10}, 30 c_0 \cdot (k+1)^5 \cdot \sqrt{\lambda_k}\right\}$\;
    \While{At least one of the following conditions holds \label{alg:while line}
    \begin{enumerate}
        \item $\exists 1 \leq i \leq r$ s.t. $w(C_i \setminus \core (C_i) \rightarrow C_i) < w(C_i \setminus \core(C_i) \rightarrow C_j)$ for some $j \neq i$; 
        \item \SpecPart finds $S \subseteq C_i$ with~\footnotemark{} $\vol(S^+) \leq \vol(\mathrm{core}(C_i))/2$, such that $\max \left\{\Phi_{G[C_i]}(S),  \Phi_{G[C_i]}(C_i \setminus S)\right\} < \phiIn$;
    \end{enumerate}\label{alg:line:while condition}}
    {
        Order the sets $C_1,\ldots, C_r$ such that   $\lambda_2(\mathcal{L}_{G[C_1]}) \leq \ldots \leq \lambda_2(\mathcal{L}_{G[C_r]})$\;\label{alg:line:beginning while loop}
        Let $1 \leq i \leq r$ the smallest index for which item~$2$ of the while-condition is satisfied, and let $S \subset C_i$ be the corresponding set\;\label{alg:line:let i and S}
        
        \uIf{$\max \left\{ \Phi_G(S^+), \Phi_G\left(\barr{S^+}\right)\right\} \leq  \left( 1 + \frac{1}{k+1} \right)^{r+1} \cdot \rho^*$\label{alg:line:condition if-1}}
    	{
    	    Let $C_i = C_i \setminus \overline{{S^+}}$, $\core(C_i) = S^+$, $C_{r+1} = \core(C_{r+1}) = \barr{S^+}$,   $r = r + 1$ and \textbf{go to} Line~\ref{alg:line:while condition}\;\label{alg:line:update if-1}
    	}
    	\uIf{$ \min \left\{ \varphi(S^+, \core(C_i)), \varphi\left(\barr{S^+}, \core(C_i)\right) \right\} \leq \frac{1}{3(k+1)}$\label{alg:line:condition if-2}}
    	{
    	    Update $\core(C_i)$ to either $S^+$ or $\barr{S^+}$ with the lower conductance, and \textbf{go to} Line~\ref{alg:line:while condition}\;\label{alg:line:update if-2}
    	}
    	\uIf{$ \Phi_G(S^-) \leq \left( 1 + \frac{1}{k+1} \right)^{r+1} \cdot \rho^*$\label{alg:line:condition if-3}}
    	{
    	    Let $C_i = C_i \setminus S^-$, $C_{r +1} = \core(C_{r+1}) = S^-$, set $r = r +1$ and \textbf{go to} Line~\ref{alg:line:while condition}\;\label{alg:line:update if-3}
    	}
    	\uIf{$ w(C_i \setminus \core(C_i) \rightarrow C_i) < w(C_i \setminus \core(C_i) \rightarrow C_j)$ for some $j \neq i$\label{alg:line:condition if-4}}
    	{
    	    Let $C_i = \core(C_i)$, merge $(C_i \setminus \core(C_i))$ with $\mathrm{argmax}_{C_j} \{ w(C_i \setminus \core(C_i) \rightarrow C_j)\}$, and \textbf{go to} Line~\ref{alg:line:while condition}\;\label{alg:line:update if-4}
    	}
    	\uIf{$ w(S^- \rightarrow C_i) < w(S^- \rightarrow C_j)$ for some $j \neq i$\label{alg:line:condition if-5}}
    	{
    	    Let $C_i = C_i \setminus S^-$, merge $S^-$ with $\mathrm{argmax}_{C_j} \{ w(S^- \rightarrow C_j)\}$, and \textbf{go to} Line~\ref{alg:line:while condition}\;\label{alg:line:update if-5}
    	}
    }
    For every partition set $C_i$\;\label{alg:line:compute trees}
    
    \uIf{ $\exists u \in C_i$ such that  $\vol(u^+) \leq \vol(\core(C_i)) /2$ and $\varphi(u^+, \core(C_i)) \leq \frac{1}{3(k+1)}$\label{alg:line:condition if-7}}
    {
    	Update $\core(C_i)$ to either $u^+$ or $\barr{u^+}$ with the lower conductance, and \textbf{go to} Line~\ref{alg:line:while condition}\;\label{alg:line:update if-7}
    }
    \uIf{ $\exists u \in C_i$ such that $\vol(u^-) \leq \vol(C_i)/2$ and $ w(u^- \rightarrow C_i) < w(u^- \rightarrow C_j)$ for some $j \neq i$\label{alg:line:condition if-8}}
    {
    	Let $C_i = C_i \setminus u^-$, merge $u^-$ with $\mathrm{argmax}_{C_j} \{ w(u^- \rightarrow C_j)\}$, and \textbf{go to} Line~\ref{alg:line:while condition}\;\label{alg:line:update if-8}
    }
    \KwRet{the partition $\{C_i\}_{i=1}^r$.}\label{alg:line:return}
\end{algorithm}
\footnotetext{If the set $S \subset C_i$ returned by \SpecPart has $\vol(S^+) > \vol(G)/2$, swap $S$ with $C_i\setminus S$.}

 Now we explain the high-level idea of the proposed algorithm, and refer the reader to Algorithm~\ref{alg:Strong Decomposition} for the formal description. Our algorithm starts with the partitioning algorithm~(Algorithm~3 in \cite{GT14}), and obtains an   intermediate partition $\{C_i\}_{i=1}^r$~(Lines~\ref{alg:while line}--\ref{alg:line:update if-5}).   For every $C_i~(1\leq i\leq r)$, the algorithm further checks if the following conditions are satisfied: 
\begin{enumerate}[label=$(A\arabic*)$]
    \setcounter{enumi}{3}
    \item   For every vertex $u \in C_i$ with $\vol(u^+) \leq \vol(\core(C_i))/2$, it holds that $$\varphi(u^+, \core(C_i)) \geq \frac{1}{3(k+1)};$$  
    \item  For every vertex $u \in C_i$ with  $\vol(u^-) \leq \vol(C_i)/2$, it holds that $$w(u^- \rightarrow C_i) \geq w(u^- \rightarrow V \setminus C_i) \cdot \frac{1}{k+1}.$$
\end{enumerate}
 If (A4) is violated by some vertex $u \in C_i$  for some $i$, then the algorithm 
uses  $u$  to refine the core set $\mathrm{core}(C_i)$~(Line~\ref{alg:line:update if-7}). If $(A5)$ is not satisfied, then the algorithm further refines the partition~(Line~\ref{alg:line:update if-8}).  The algorithm repeats this local refinement process until no such update is found anymore. 
In the following analysis, we set
\[
    \rho^* \triangleq \min \left\{ \frac{\lambda_{k+1}}{10}, 30 c_0 \cdot (k+1)^5 \cdot \sqrt{\lambda_k}\right\},
\]
where $c_0$ is the constant   specified in Lemma~\ref{lem:Upperbound eig induced graphs}, and
\begin{equation}\label{eq:PhiIn PhiOut Choices}
    \phiIn \triangleq \frac{\lambda_{k+1}}{140(k+1)^2}, \qquad
    \phiOut \triangleq 90 c_0 \cdot (k+1)^6 \sqrt{\lambda_k}.
\end{equation}
Notice that, by assuming $\lambda_k < \lp \frac{1}{270 \cdot c_0 \cdot (k+1)^6} \rp^2$ in Lemma~\ref{lem:Improved Decomposition}, it holds that $\phiOut < 1/3$. This fact will be used several times in our analysis.

Following the proof structure in \cite{GT14},  we will  prove Lemma~\ref{lem:Improved Decomposition} via a sequence of claims.  Notice that, during the entire execution  of the algorithm, the sets $\{C_i\}_{i=1}^{r}$ always form a partition of $V$,  and each $\core(C_i)$  is a subset of $C_i$. Firstly, we show that, at any point during the execution of the algorithm, the core sets $\core(C_i)~(1\leq i\leq r)$ always have low conductance.

\begin{claim}\label{claim:Strong decomp claim 1}
    Throughout the algorithm, we always have that 
    \[
        \max_{1 \leq i \leq r} \Phi_G(\core(C_i)) \leq \rho^* \cdot \lp 1 + \frac{1}{k+1} \rp ^{r}.
    \]
\end{claim}

The following result will be used in our proof:
\begin{lemma}[Lemma~2.2,  \cite{GT14}]\label{lem:Refinement B_i}
Let $G=(V,E, w)$ be a graph, and let $S, W$   be two subsets such that $S \subset W \subseteq V$.  
 Suppose that the following two conditions are satisfied for some $\varepsilon > 0$:
\begin{enumerate}
    \item $\varphi(S, W) \leq  \varepsilon /3$ and 
    \item $\max \left\{ \Phi_G(S), \Phi_G(W \setminus S)\right\} \geq (1 + \varepsilon)\cdot \Phi_G(W)$.
\end{enumerate}
Then it holds that 
\[
    \min \left\{ \Phi_G(S), \Phi_G(W \setminus S)\right\} \leq \Phi_G(W).
\]
\end{lemma}

\begin{proof}[Proof of Claim~\ref{claim:Strong decomp claim 1}]

Let $r$ be the current number of clusters generated by the algorithm, and we prove by induction that the claim holds during the entire execution of the algorithm. First of all, for  the base case of  $r = 1$, we have that $\core(C_1) = C_1 = V$, which gives us that $\Phi_G(\core(C_1))=0$; hence,  the statement holds trivially.

Secondly, for the inductive step, we assume that the statement holds for some fixed configuration of the core sets $\{\mathrm{core}(C_i)\}_{i=1}^r$ and we prove that the statement holds after the algorithm updates the current configuration. Notice that   $\{\core(C_i)\}_{i=1}^r$ are updated through Lines~\ref{alg:line:update if-1}, \ref{alg:line:update if-2}, \ref{alg:line:update if-3},   and \ref{alg:line:update if-7} of the algorithm, so it suffices to show that the claim holds after executing  these lines. We continue the proof with case distinction.
\begin{itemize}
    \item  When executing Lines~\ref{alg:line:update if-1}, \ref{alg:line:update if-3}, the algorithm introduces some new set $\core(C_{r+1})$ such that
    \[
    \Phi_G(\core(C_{r+1})) \leq \rho^*\cdot \left(1 +\frac{1}{k+1} \right)^{r+1}.
    \]
Combining this with the inductive hypothesis, which assumes the inequality holds for $\core(C_{i})~(1\leq i\leq r)$, we have that 
\[
\max_{1\leq i\leq r+1}\Phi_G(\core(C_{i})) \leq \rho^*\cdot \left(1 +\frac{1}{k+1} \right)^{r+1}.
\]
\item  The case for executing Lines~\ref{alg:line:update if-2} and \ref{alg:line:update if-7} is similar, and we  prove this by applying Lemma~\ref{lem:Refinement B_i}. We first focus on Line~\ref{alg:line:update if-2} here, dealing with Line~\ref{alg:line:update if-7} next. When executing Line~\ref{alg:line:update if-2}, we know that the \texttt{if}-condition in Line~\ref{alg:line:condition if-1} does not hold, so we have that 
    \[
        \max \left\{ \Phi_G(S^+), \Phi_G\left(\barr{S^+}\right) \right\} > \lp 1 + \frac{1}{k+1}\rp^{r +1} \cdot \rho^* \geq \lp 1 + \frac{1}{k+1}\rp \cdot \Phi_G(\core(C_i)),
    \]
    where the last inequality follows by the inductive hypothesis. Moreover, when executing Line~\ref{alg:line:update if-2}, we also know that the \texttt{if}-condition in Line~\ref{alg:line:condition if-2}  holds, i.e.,
    \[
        \min \left\{ \varphi\left(S^+, \core(C_{i})\right), \varphi\left(\barr{S^+}, \core(C_{i})\right)\right\} \leq \frac{1}{3(k+1)}.
    \]
    Therefore, by applying Lemma~\ref{lem:Refinement B_i} with  $S^+ \subset \mathrm{core}(C_i)$ and $\varepsilon = 1/(k+1)$ and using the inductive hypothesis, we conclude that
    \[
        \min \left\{ \Phi_G(S^+), \Phi_G\left(\barr{S^+}\right)\right\} \leq \Phi_G(\core(C_{i})) \leq\rho^*\cdot  \lp 1 + \frac{1}{k+1} \rp^{r}.
    \]
    \item  Now we look at Line~\ref{alg:line:update if-7}. When executing Line~\ref{alg:line:update if-7}, we know that $u^+ = u$; since otherwise   $u^+ = \emptyset$ and $\varphi(\emptyset, \core(C_i)) = 1$. Therefore, it holds that
    \[
        1 = \max \left\{ \Phi_G(u^+), \Phi_G\left(\barr{u^+}\right) \right\} > \lp 1 + \frac{1}{k+1}\rp^{r +1} \cdot \rho^* \geq \lp 1 + \frac{1}{k+1}\rp \cdot \Phi_G(\core(C_i)),
    \]
    where the first inequality holds because $\Phi_G(u^+) = 1$, and the last inequality follows by the inductive hypothesis. Moreover, when executing Line~\ref{alg:line:update if-7}, we also know that the \texttt{if}-condition in Line~\ref{alg:line:condition if-7}  holds, i.e.,
    \[
         \varphi\left(u^+, \core(C_{i})\right) \leq \frac{1}{3(k+1)}.
    \]
    Therefore, by applying Lemma~\ref{lem:Refinement B_i} with  $u^+ \subset \mathrm{core}(C_i)$ and $\varepsilon = 1/(k+1)$ and using the inductive hypothesis, we conclude that
    \[
        \min \left\{ \Phi_G(u^+), \Phi_G\left(\barr{u^+}\right)\right\} = \Phi_G\left(\barr{u^+}\right) \leq \Phi_G(\core(C_{i})) \leq\rho^*\cdot  \lp 1 + \frac{1}{k+1} \rp^{r},
    \]
     where the first equality holds because     $\Phi_G(u^+) = 1$  holds for $u^+ = u$.
\end{itemize}
Combining the two cases above, we know that the claim always holds during the entire execution of the algorithm. This completes the proof.
\end{proof}

Next, we  show that the number of partition sets cannot exceed $k$. This proof  is identical to Claim~$3.2$ in \cite{GT14}, and we include the proof  here for completeness.  

\begin{claim}\label{claim:Strong decomp claim 2}
    The total number of clusters returned by the algorithm satisfies that $\ell\leq k$.
\end{claim}

\begin{proof}

Suppose for contradiction that the number of clusters becomes $r = k + 1$ at some point during the execution of the algorithm. Then, since $\core(C_1),\ldots, \core(C_{k+1})$ are disjoint, by the definition of $\rho(k+1)$ and Claim~\ref{claim:Strong decomp claim 1} we have that 
\[
        \rho(k+1) 
        \leq \max_{1 \leq i \leq k+1} \Phi_G(\mathrm{core}(C_i))
        \leq \lp 1 + \frac{1}{k+1}\rp^{k+1} \cdot \rho^*
        \leq \mathrm{e} \cdot \rho^*
        \leq \mathrm{e} \cdot \frac{\lambda_{k+1}}{10}
        < \frac{\lambda_{k+1}}{2},
    \]
    which contradicts the higher-order Cheeger inequality~\eqref{eq:Higher Cheeger}.   Therefore, the total number of clusters at any time satisfies $r < k+1$, and the total number of returned clusters   satisfies $\ell \leq k$.  
\end{proof}

Now we are ready to show   that the output  $\{C_i\}_{i=1}^{\ell}$ of Algorithm~\ref{alg:Strong Decomposition} and its core sets $\{\core(C_i) \}_{i=1}^{\ell}$ satisfy Properties~$(A4)$ and $(A5)$, which will be used in proving Lemma~\ref{lem:Improved Decomposition}. 

\begin{claim}\label{claim:Strond decomp claim 3}
Let $\{C_i\}_{i=1}^{\ell}$ be the output of Algorithm~\ref{alg:Strong Decomposition} with corresponding core sets $\{\core(C_i)\}_{i=1}^{\ell}$.
Then, the following hold for any $1\leq i\leq \ell$:
\begin{enumerate}
    \item $\Phi_G(\core(C_i))\leq \phiOut/(k+1)$;
    \item $\Phi_G(C_i) \leq \phiOut$;
    \item $\Phi_{G[C_i]}\geq \phiIn^2/4 $;
    \item  \ For every   $u \in C_i$ with $\vol(u^+)\leq \vol(\core(C_i))/2$, we have that
    \[
        \varphi(u^+, \core(C_i)) \geq \frac{1}{3(k+1)};
    \]
    \item 
     For every   $u\in C_i$ with
     $\vol(u^-) \leq \vol(C_i)/2$, we have  that $$w(u^- \rightarrow C_i) \geq w(u^- \rightarrow V \setminus C_i) \cdot \frac{1}{k+1}.$$
\end{enumerate}
\end{claim}

\begin{proof}
First of all, by Claim~\ref{claim:Strong decomp claim 1} we have for any $1\leq i\leq \ell$ that 
 \[
    \Phi_G(\core(C_i)) 
    \leq \rho^* \cdot \lp 1 + \frac{1}{k+1}\rp^{\ell}
    \leq \mathrm{e} \cdot \rho^*
    \leq 30 \cdot \mathrm{e} \cdot c_0 \cdot (k+1)^5 \sqrt{\lambda_k}
    \leq \frac{\phiOut}{k+1},
\]
where the second inequality holds by the fact  that $\ell\leq k+1$, the third one holds by the choice of $\rho^*$, and the last one holds by the choice of $\phiOut$. This proves Item~(1).

To prove Item~(2), we notice that the first condition of the \texttt{while}-loop~(Line~\ref{alg:while line}) doesn't hold when the algorithm terminates, hence   we have for any $1\leq i\neq j\leq \ell$ that \[w(C_i \setminus\core(C_i)\rightarrow C_i)\geq w(C_i \setminus \core(C_i) \rightarrow C_j).
\]
By applying the averaging argument, we have that 
\begin{align}
    w(C_i \setminus \core(C_i) \rightarrow \core(C_i)) & = w(C_i \setminus \core(C_i) \rightarrow C_i) \nonumber\\
    &\geq \frac{w(C_i \setminus \core(C_i) \rightarrow V)}{\ell}\nonumber \\
    &     \geq \frac{w(C_i \setminus \core(C_i) \rightarrow V \setminus C_i)}{k}. \label{eq:claim Strong Decomp Claim 3}
\end{align} 
We apply the same analysis used in \cite{GT14}, and have that
\begin{align*}
    \Phi_G(C_i) 
    &= \frac{w(C_i \rightarrow V)}{\vol(C_i)} \\
    &\leq \frac{w \lp \mathrm{core}(C_i) \rightarrow V \rp + w \lp C_i \setminus \mathrm{core}(C_i)\rightarrow V \setminus C_i \rp - w(C_i \setminus \mathrm{core}(C_i) \rightarrow \mathrm{core}(C_i)) }{\vol(\mathrm{core}(C_i))}\\
    &\leq \Phi_G(\mathrm{core}(C_i)) + \frac{(k-1) \cdot w(C_i \setminus \mathrm{core}(C_i) \rightarrow \mathrm{core}(C_i))}{\vol(\mathrm{core}(C_i))}\\
    &\leq k \cdot \Phi_G(\mathrm{core}(C_i))\\
    &\leq \phiOut,
\end{align*}
where the second inequality uses equation~\eqref{eq:claim Strong Decomp Claim 3}. This proves Item~(2).

Next, we analyse Item~(3). Again, we know that   the second condition within the \texttt{while}-loop~(Line~\ref{alg:while line}) does not hold when the algorithm terminates. By the performance of  the \SpecPart algorithm~(i.e., Lemma~\ref{lem:Cheeger's ineq}), it holds for any $1\leq i\leq \ell$ that 
$\Phi_{G[C_i]}\geq  \phiIn^2/4$. With this, we prove that   Item~(3) holds.
    
 Similarly, when the algorithm terminates, we know that for any node $u \in C_i$ the \texttt{if}-condition in Line~\ref{alg:line:condition if-7} does not hold.
Hence, for any $1\leq i\leq \ell$ and any $u \in C_i$ with $\vol(u^+)\leq \vol(\core(C_i))/2$, we have that
 \[
    \varphi(u^+, \core(C_i)) \geq \frac{1}{3(k+1)}.
\] This shows that Item~(4) holds as well.
    
 Finally, since there is no $u \in C_i$ satisfying the \texttt{if}-condition in Line~\ref{alg:line:condition if-8} of the algorithm, it holds for any $1\leq i\neq j\leq \ell$ and every vertex $u \in C_i$ that $w(u^{-} \rightarrow C_i) \geq w(u^{-} \rightarrow C_j)$. Therefore, by the same averaging argument we have that \[
    w(u^- \rightarrow C_i) 
    \geq \frac{w(u^- \rightarrow V)}{\ell}
    \geq \frac{w(u^- \rightarrow V\setminus C_i)}{k+1},
\]which shows that Item~(5) holds.
\end{proof}

It remains to prove that the algorithm does terminate. To prove this, we first show that, in each iteration of the \texttt{while}-loop (Lines~\ref{alg:line:beginning while loop}--\ref{alg:line:update if-5}), at least one of the \texttt{if}-conditions will be satisfied, and some sets are updated accordingly. This fact, stated as  Claim~\ref{claim:Strong decomp claim 4}, is important, since otherwise the algorithm might end up
in   an infinite loop. The following result will be used in our proof.

\begin{lemma}[Lemma~2.6, \cite{GT14}] \label{lem:Inner Conductance C_i}
    Let $\mathrm{core}(C_i) \subseteq C_i \subseteq V$, and  $S \subset C_i$ be such that $\vol(S^+)~\leq~\vol(\mathrm{core}(C_i))/2$.  Suppose that the following hold for some parameters $\rho$ and $0 < \varepsilon < 1$:
    \begin{enumerate}
        \item $\rho \leq \Phi_G(S^-)$ and $\rho \leq \max \{ \Phi_G(S^+), \Phi_G(\barr{S^+})\}$;
        \item  If $S^- \neq \emptyset$, then $w(S^- \rightarrow C_i) \geq w(S^- \rightarrow V)/k$;
        \item If $S^+ \neq \emptyset$, then $\varphi(S^+, \mathrm{core}(C_i)) \geq \varepsilon/3$ and $\varphi\left(\barr{S^+}, \mathrm{core}(C_i)\right) \geq \varepsilon/3$.
    \end{enumerate}
    Then,  it holds that $$\Phi_{G[C_i]}(S) \geq \varepsilon \cdot \frac{\rho}{14k}.$$
\end{lemma}

\begin{claim}\label{claim:Strong decomp claim 4}
    If at least one condition of the \texttt{while}-loop is satisfied, then at least one of the \texttt{if}-conditions~(Lines~\ref{alg:line:condition if-1},\ref{alg:line:condition if-2},\ref{alg:line:condition if-3},\ref{alg:line:condition if-4} or \ref{alg:line:condition if-5}) is satisfied.  
\end{claim}

\begin{proof}
First of all, notice that if the first condition of the \texttt{while}-loop  is satisfied, then the \texttt{if}-condition in Line~\ref{alg:line:condition if-4} will be satisfied and the claim holds. Hence,   we assume   that only the second condition of the \texttt{while}-loop is satisfied,  and we   prove the claim by contradiction. That is, we   show that, if none of the \texttt{if}-conditions holds, then the set $S$ returned by  the \SpecPart algorithm would satisfy that $\Phi_{G[C_i]}(S)\geq \phiIn$.
 The proof is structured in the following two steps:
\begin{enumerate}
    \item We  first prove that  \[\Phi_{G[C_i]}(S) \geq \frac{\max \{ \rho^*, \rho(r+1)\}}{14(k+1)^2};\]
    \item Using Item~$(1)$ we prove that $\Phi_{G[C_i]}(S)\geq \phiIn$ and reach our desired contradiction.
\end{enumerate} 

 \paragraph{Step~1:}
    We prove this fact by applying Lemma~\ref{lem:Inner Conductance C_i} with parameters
    \[
        \rho \triangleq \max \{\rho^*, \rho(r+1)\}
        \quad \text{and}
        \quad \varepsilon \triangleq \frac{1}{k+1}.
    \]
    Let us show that the conditions of Lemma~\ref{lem:Inner Conductance C_i} are satisfied, beginning with the first one. If $S^{-} = \emptyset$, then we trivially have that $1 = \Phi_G(S^{-}) \geq \rho$; so we assume that $S^{-} \neq \emptyset$. As the \texttt{if}-condition in Line~\ref{alg:line:condition if-3} is not  satisfied, we have that 
    \[
        \Phi_G(S^-)\geq \left(1 + \frac{1}{k+1} \right)^{r+1}\cdot \rho^*,
    \]
    and combining this with Claim~\ref{claim:Strong decomp claim 1} gives us that
    \[
        \Phi_G(S^-) = \max 
        \left\{ \Phi_G(\core(C_1)), \dots, \Phi_G(\core(C_r)), \Phi_G(S^-)\right\} \geq \rho(r+1).
    \]
    Therefore, we have that 
    \begin{equation}\label{eq:Lower Bound Cond S-}
        \Phi_G(S^-) \geq \max\{\rho^*, \rho(r+1)\} = \rho.
    \end{equation}
    Similarly, if $S^+ = \emptyset$, then we trivially have that $1 = \max \{\Phi_G(S^+), \Phi_G(\overline{S^+})\} \geq \rho$; so we assume that $S^+ \neq \emptyset$. Moreover, since
    we have chosen $S^+$ such that $\vol(S^+) \leq \vol(\mathrm{core}(C_i))/2$, we know that $\overline{S^+} = \mathrm{core}(C_i) \setminus S^+ \neq \emptyset$. As the \texttt{if}-condition in Line~\ref{alg:line:condition if-1} is not satisfied, we have that
    \[
    \max\left\{\Phi_G(S^+), \Phi_G\left(\overline{S^+}\right)\right\}\geq \left(1 + \frac{1}{k+1}\right)^{r+1} \cdot\rho^*.
    \]
    Combining this with Claim~\ref{claim:Strong decomp claim 1} gives us that
    \begin{align*}
        \lefteqn{\max\left\{\Phi_G(S^+), \Phi_G\left(\overline{S^+}\right)\right\}}\\
        &=  \max \{ \Phi_G(\core(C_1)), \dots, \Phi_G(\core(C_{i-1})), \Phi_G(S^+),\Phi_G(\barr{S^+}), \\
        & \qquad  \Phi_G(\core(C_{i+1})), \dots, \Phi_G(\core(C_{r}))\}\\
        &\geq   \rho(r+1).
    \end{align*}
    
    Therefore we have that 
    \begin{equation}\label{eq:Lower bound S+ and S+bar}
     \max\left\{\Phi_G(S^+), \Phi_G\left(\overline{S^+}\right)\right\} \geq \max\{\rho^*, \rho(r+1)\} = \rho.
    \end{equation}
    Combining  \eqref{eq:Lower Bound Cond S-} and \eqref{eq:Lower bound S+ and S+bar},  we see that the first condition of Lemma~\ref{lem:Inner Conductance C_i} is satisfied.
    Since the \texttt{if}-condition in Line~\ref{alg:line:condition if-5} is not satisfied, it follows by an averaging argument that 
    \[
         w(S^- \rightarrow C_i) \geq  \frac{w(S^- \rightarrow V)}{k}, 
    \]
    which shows that the second condition of Lemma~\ref{lem:Inner Conductance C_i} is satisfied.
    Finally, since the \texttt{if}-condition in Line~\ref{alg:line:condition if-2} is not satisfied, we know that
    \[
        \min \left\{ \varphi(S^+, \core(C_i)), \varphi\left(\barr{S^+}, \core(C_i)\right) \right\} \geq \frac{1}{3(k+1)},
    \]
    which shows that the third condition of Lemma~\ref{lem:Inner Conductance C_i} is satisfied as well. 
Hence, by Lemma~\ref{lem:Inner Conductance C_i} we conclude that
\begin{equation}\label{eq:claim:Strong decomp claim 4 eq1}
    \Phi_{G[C_i]}(S) 
    \geq \frac{\varepsilon \cdot \rho}{14(k+1)}
    = \frac{\max\{ \rho^*, \rho(r+1)\} }{14(k+1)^2},
\end{equation}
which completes the proof of the first step.

\paragraph{Step~2:}We prove this step with a case distinction as follows:

    \emph{Case~1: $r = k$.} By \eqref{eq:claim:Strong decomp claim 4 eq1} and \eqref{eq:Higher Cheeger}, we have that 
    \[
        \Phi_{G[C_i]}(S) 
        \geq \frac{\rho(r+1)}{14(k+1)^2}
        = \frac{\rho(k+1)}{14(k+1)^2}
        \geq \frac{\lambda_{k+1}}{28(k+1)^2}
        \geq \phiIn,
    \]
    which leads to the desired contradiction.  
    
    \emph{Case~2: $r < k$.} Recall that the  partition  sets $\{C_i\}_{i=1}^{r}$ are labelled such that $\lambda_2(\mathcal{L}_{G[C_1]}) \leq \ldots \leq \lambda_2(\mathcal{L}_{G[C_r]})$, and the algorithm has chosen the lowest index $i$ for which the set $S\subset C_i$ returned by the \SpecPart algorithm satisfies the second condition of the \texttt{while}-loop.
     Our proof is based on  a further case distinction depending  on the value of  $i$.

    \emph{Case~2a:} $i=1$~(i.e., the algorithm selects $S \subseteq C_1$).
    We combine  the performance of the \SpecPart  algorithm~(Lemma~\ref{lem:Cheeger's ineq}) with  Lemma~\ref{lem:Upperbound eig induced graphs}, and obtain that
    \begin{equation}\label{eq:claim:Strong decomp claim 4 eq2}
        \Phi_{G[C_1]}(S) \leq  \sqrt{2\lambda_2(\mathcal{L}_{G[C_1]})} = \min_{1 \leq j \leq r} \sqrt{2 \lambda_2( \mathcal{L}_{G[C_j]})} \leq \sqrt{4c_0 \cdot k^6 \cdot \lambda_k}.
    \end{equation}
    Combining   \eqref{eq:claim:Strong decomp claim 4 eq1} and \eqref{eq:claim:Strong decomp claim 4 eq2} we have that 
    \[
        \rho^* \leq 28 c_0 \cdot (k+1)^5\sqrt{\lambda_k}.
    \]
    Thus, by the definition of $\rho^*$ we   have that 
    \[
        \rho^* = \frac{\lambda_{k+1}}{10}.
    \]
    We combine this with  \eqref{eq:claim:Strong decomp claim 4 eq1}, and have that 
    \[
        \Phi_{G[C_i]}(S) \geq \frac{\lambda_{k+1}}{140 (k+1)^2} = \phiIn,
    \]
    which gives our desired contradiction.
    
    \emph{Case~2b: $i>1$}~(i.e., the algorithm selects $S \subset C_i$  for some $i \geq 2$).
    Let $S_1 \subset C_1$ be the set obtained by applying the \SpecPart algorithm to the graph $G[C_1]$. Since the algorithm did not select $S_1 \subset C_1$, we know that $\Phi_{G[C_1]}(S_1) \geq \phiIn$. Combining the performance of the \SpecPart algorithm~(Lemma~\ref{lem:Cheeger's ineq}) with Lemma~\ref{lem:Upperbound eig induced graphs}, we have that  
    \[
        \phiIn \leq \Phi_{G[C_1]}(S_1) \leq \min_{1 \leq j \leq r} \sqrt{2 \lambda_{2}(\mathcal{L}_{G[C_j]})} \leq 2c_0 \cdot k^3 \cdot \sqrt{\lambda_k}.
    \]
       This gives us  that
    \[
         \frac{\lambda_{k+1}}{10} = 14(k+1)^2 \cdot \phiIn < 30c_0 \cdot (k+1)^5 \cdot \sqrt{\lambda_k},
    \]
    and it holds by the definition of $\rho^*$ that
    \[
        \rho^* = \frac{\lambda_{k+1}}{10}.
    \]
    Therefore,  by \eqref{eq:claim:Strong decomp claim 4 eq1} we have that
    \[
        \Phi_{G[C_i]}(S) \geq \frac{\lambda_{k+1}}{140 (k+1)^2} = \phiIn.
    \]
    Combining the two cases above gives us the desired contradiction. With this, we complete the proof of the claim. 
\end{proof}

Next, we will show that the total number of iterations that the algorithm runs, i.e.,  the number of times the instruction ``\texttt{go to Line~\ref{alg:while line}}'' is executed, is finite.

\begin{claim}\label{claim:Strong decomp claim 5}
For any graph $G=(V,E,w)$ with the minimum weight $w_{\min}$ as the input, Algorithm~\ref{alg:Strong Decomposition} terminates after executing the while-loop $O \lp  k \cdot n \cdot \vol(G)/w_{\min}\rp$ times.
\end{claim}

\begin{proof} 
Notice that the algorithm goes back to check the loop conditions~(Line~\ref{alg:while line})
right after any of   Lines~\ref{alg:line:update if-1}, \ref{alg:line:update if-2}, \ref{alg:line:update if-3}, \ref{alg:line:update if-4}, \ref{alg:line:update if-5},  \ref{alg:line:update if-7}~and~\ref{alg:line:update if-8} is executed, and each of these commands changes the current structure of our partition $\{C_i\}_{i=1}^{r}$ with core sets $\core(C_i) \subseteq C_i$. We classify these updates into the following three types:
    \begin{enumerate}
        \item The updates that introduce  a new partition set $C_{r+1}$. These correspond to Lines~\ref{alg:line:update if-1}, and \ref{alg:line:update if-3};
        \item The updates  that contract the core sets $\core(C_i)$ to a strictly smaller subset $T \subset \core(C_i)$.  These correspond to Lines~\ref{alg:line:update if-2} and \ref{alg:line:update if-7}; 
        \item The updates that refine the partition sets $\{C_i\}_{i=1}^{r}$ by moving  a subset $ T \subseteq C_i \setminus \core(C_i)$ from the partition set $C_i$ to a different partition set $C_j$, for some $C_i \neq C_j$. These correspond to Lines~\ref{alg:line:update if-4}, \ref{alg:line:update if-5}, and \ref{alg:line:update if-8}. 
    \end{enumerate}
    We prove that these updates can occur only a finite number of times. The first type of updates can occur at most $k$ times, since we know by Claim~\ref{claim:Strong decomp claim 2}  that the algorithm outputs $\ell\leq k$ clusters.  Secondly, for a fixed value of $\ell$, the second type of updates occurs  at most $n$ times, since each update decreases the size of some $\core(C_i)$ by at least one. Finally, for a fixed $\ell$ and a fixed configuration of core sets $\core(C_i) \subseteq C_i$, the third type of updates occurs at most  $O(\vol(G)/w_{\min})$  times. This is due to the fact that, whenever every such update is  executed, the total weight between different partition sets, i.e., $\sum_{i\neq j } w(C_i, C_j)$, decreases by at least $w_{\min}$.   Combining everything together proves the lemma.
\end{proof}

Finally, we combine everything together and prove Lemma~\ref{lem:Improved Decomposition}.

\begin{proof}[Proof of Lemma~\ref{lem:Improved Decomposition}]
We first show that Properties~$(A1), (A2)$ and $(A3)$ hold, and in the end we analyse the runtime of the algorithm. Combining Items~(2) and (3)  of Claim~\ref{claim:Strond decomp claim 3} with the choices of $\phiIn, \phiOut$ in \eqref{eq:PhiIn PhiOut Choices}, we obtain for all $1 \leq i \leq \ell$ that $\Phi_G(C_i) \leq \phiOut = O \lp k^6 \cdot \sqrt{\lambda_k}\rp$ and $\Phi_{G[C_i]} \geq  \phiIn^2/ 4 = \Omega \lp\lambda_{k+1}^2/k^4 \rp$. Hence, Properties $(A1)$ and $(A2)$ hold for every $C_i$.

 To analyse Property~(A3), we fix an arbitrary node $u\in C_i$ that belongs to the partition set $C_i$ with core set $\core(C_i)$. 
By definition, we have that 
\[
    w(u, V\setminus C_i) = w(u^+, V\setminus C_i) + w(u^-, V\setminus C_i).
\]
 We study $w(u^+, V\setminus C_i)$ and $ w(u^-, V\setminus C_i)$ separately.

\paragraph{Bounding  the value of  $w(u^+, V \setminus C_i)$:} We analyse $w(u^+, V \setminus C_i)$ by the following case distinction.

 \underline{Case~1: $\vol(u^+) \leq \vol(\core(C_i))/2$.}
By Item~(4)  of Claim~\ref{claim:Strond decomp claim 3} we know that\[
    \varphi(u^+, \core(C_i)) \geq \frac{1}{3(k+1)},
\]
which is equivalent to 
\[
    3(k+1) \cdot \frac{\vol(\core(C_i))}{\vol(\barr{u^+})} \cdot w(u^+ \rightarrow \core(C_i)) \geq  w(u^+ \rightarrow V \setminus \core(C_i)).
\]
This implies that 
\[
    6(k+1) \cdot w(u^+ \rightarrow \core(C_i)) \geq w(u^+ \rightarrow V \setminus \core(C_i)),
\]
and we have   that
    \begin{align*}
        w(u^+, V \setminus C_i) 
        & 
        \leq w(u^+ \rightarrow V \setminus \core(C_i)) \\
         &\leq 6(k+1) \cdot w(u^+ \rightarrow \core(C_i)) \\
        &\leq 6(k+1) \cdot \vol_{G[C_i]}(u^+).
    \end{align*}
    
     \underline{Case~2: $\vol(u^+) >  \vol(\core(C_i))/2$.}
    We have that 
    \begin{align}
        w(u^+, V \setminus C_i) 
        &\leq w(\core(C_i), V \setminus C_i) 
        \leq w(\core(C_i), V \setminus \core(C_i)) \nonumber
       \\
       & = \vol(\core(C_i)) \cdot \Phi_G(\core(C_i))\leq \vol(\core(C_i)) \cdot \frac{\phiOut}{k+1} \nonumber \\
       & 
        < \frac{2\phiOut}{k+1} \cdot \vol(u^+) \label{eq:claim:cost crossing edges eq1}, 
    \end{align} where the third  inequality follows by Item~(1)  of Claim~\ref{claim:Strond decomp claim 3}. Therefore,  we have that
    \begin{equation}\label{eq:claim:cost crossing edges eq2}
        \vol_{G[C_i]}(u^+) = \vol(u^+) - w(u^+, V \setminus C_i) > \vol(u^+) \lp 1 - \frac{2\phiOut}{k+1} \rp,
    \end{equation} where the last inequality follows by \eqref{eq:claim:cost crossing edges eq1}. We further combine   \eqref{eq:claim:cost crossing edges eq1} with  \eqref{eq:claim:cost crossing edges eq2}, and obtain  that
    \[
        w(u^+, V \setminus C_i) 
        \leq \frac{2\phiOut}{k+1} \cdot \frac{1}{1 - \frac{2\phiOut}{k+1}} \cdot \vol_{G[C_i]}(u^+)
        \leq \frac{2\phiOut}{k} \cdot \vol_{G[C_i]}(u^+),
    \]where the last inequality holds by our assumption that $\phiOut < 1/3$.
    
  Therefore,  combining the two cases above gives us that 
    \begin{equation}\label{eq:claim:cost crossing edges eq3}
        w(u^+, V \setminus C_i) \leq   6(k+1)   \cdot \vol_{G[C_i]}\left(u^+\right).
    \end{equation}
    
     \paragraph{Bounding the value of  $w(u^-, V \setminus C_i)$:}
    We analyse   $w(u^-, V \setminus C_i)$ based on the following two cases.
    
     \underline{Case~1: $\vol(u^-) \leq \vol(C_i)/ 2$.} By Item~(5) of Claim~\ref{claim:Strond decomp claim 3}, we know that 
    \[
        w(u^- \rightarrow C_i) \geq w(u^-\rightarrow V \setminus C_i) \cdot \frac{1}{(k+1)},
    \]
    which gives us that   
    \[
        w(u^-, V \setminus C_i) \leq (k+1) \cdot w(u^- \rightarrow C_i) \leq (k+1)
        \cdot \vol_{G[C_i]}(u^-).
    \]
    
   \underline{Case 2: $\vol(u^-) >  \vol(C_i)/2$.} In this case, we have that
    \begin{equation}\label{eq:claim:cost crossing edges eq4}
        w(u^-, V \setminus C_i) 
        \leq w(C_i, V \setminus C_i) 
         = \Phi_G(C_i) \cdot \vol(C_i) 
        \leq \phiOut \cdot \vol(C_i)
        \leq 2\phiOut \cdot \vol(u^-),
    \end{equation}
    where the  second  inequality follows by  Item~(2) of Claim~\ref{claim:Strond decomp claim 3}. This implies that 
    \begin{equation}\label{eq:claim:cost crossing edges eq5}
        \vol_{G[C_i]}(u^-) 
        = \vol(u^-) - w(u^-, V \setminus C_i) 
        \geq (1 - 2\phiOut) \cdot \vol(u^-),
    \end{equation}
    where the last inequality follows by  \eqref{eq:claim:cost crossing edges eq4}. Finally,  combining \eqref{eq:claim:cost crossing edges eq4} and \eqref{eq:claim:cost crossing edges eq5} gives us that
    \[
        w(u^-, V \setminus C_i) \leq \frac{2\phiOut}{1 - 2\phiOut} \cdot \vol_{G[C_i]}(u^-) \leq 2 \cdot \vol_{G[C_i]}(u^-),
    \] where the last inequality follows by our assumption that $\phiOut < 1/3$. 
    Therefore,   combining the two cases together gives us that  
    \begin{equation}\label{eq:claim:cost crossing edges eq6}
        w(u^-, V \setminus C_i) \leq (k+1) \cdot \vol_{G[C_i]}(u^-).
    \end{equation}
     Our claimed property $(A3)$ follows by summing the inequalities in   \eqref{eq:claim:cost crossing edges eq3} and \eqref{eq:claim:cost crossing edges eq6} and the fact that  \[\vol_{G[C_i]}(u) = \vol_{G[C_i]}(u^+) + \vol_{G[C_i]}(u^-).\]
    
Finally, we analyse the runtime of the algorithm.
  By   Claims~\ref{claim:Strong decomp claim 4} and \ref{claim:Strong decomp claim 5}, we know that the algorithm does terminate, and the total number of iterations of the main \texttt{while}-loop executed by the algorithm is upper bounded by  $O(k\cdot n \cdot\vol(G) / w_{\min})$. Notice that   this quantity is upper bounded by $O(\poly(n))$ given our assumption that $w_{\max}/w_{\min} = O(\poly(n))$.  This completes the proof. 
\end{proof}

\section{Omitted Details of Proof of Theorem~\ref{thm:main1}}\label{sec:appendixouralg}
This section presents omitted details of the proof of Theorem~\ref{thm:main1}. Our key approach of breaking the running time barrier of \cite{manghiuc_sun_hierarchical} is that, instead of approximating the optimal tree $\tree^*$, we    approximate  the tree constructed by the algorithm of \cite{manghiuc_sun_hierarchical}. 

At a very high level, we first slightly adjust the tree construction by \cite{manghiuc_sun_hierarchical}, and show that the cost of this adjusted variant is the same as their original tree construction.  We call the adjusted tree $\tree_{\textsf{MS}}$. Then, we perform a sequence of changes to $\tree_{\textsf{MS}}$, such that the total induced cost of these changes is not too high. After this sequence of changes, we end up with a tree $\tree_{\textsf{MS}}''$, whose cost can be approximated  in nearly-linear time with  the output of Algorithm~\ref{alg:spectral_clustering_degree_bucketing}.
  In particular, we only require the \emph{existence} of the tree $\tree_{\textsf{MS}}$, and not its explicit construction.

The section is organised as follows:  In Section~\ref{sec:appendix_adjusted_ms} we introduce a variant of the algorithm and tree construction from \cite{manghiuc_sun_hierarchical}. In Section~\ref{sec:keylemourtechniques} we prove Lemma~\ref{lemma:upper bound inner nodegree} and Lemma~\ref{lemma:upper bound WOPT nodegree}. Finally,  we prove Theorem~\ref{thm:main1} in Section~\ref{sec:proof_main_theorem}.

\subsection{Caterpillar Tree on Extended Critical Nodes}\label{sec:appendix_adjusted_ms}
 Now we present an adjusted version of the result on well-clustered graphs from \cite{manghiuc_sun_hierarchical}. Their original algorithm first uses their variant of Lemma~\ref{lem:Improved Decomposition} to explicitly construct a partition $\{C_1, \ldots, C_\ell\}$, and then decomposes every $C_i$ into critical nodes $\mathcal{S}_i$. These critical nodes are used to construct their final tree.

In our adjusted form, instead of using the critical nodes $\mathcal{S}_i$ as the building blocks, we use the ``extended'' critical nodes as building blocks, which are described below.

Let $C_1, \ldots C_\ell$ be the partition returned by Lemma~\ref{lem:Improved Decomposition}. We assume for every $C_i$ that $(A_0, \ldots, A_{r})$ is the dense branch of $\T_i = \T_{\deg}(G[C_i])$ for some $r\in\mathbb{Z}_+$, and we extend the dense branch to 
 $(A_0,\ldots, A_{i_{\max}-1})$ with the   property that, for every $i\in[r,i_{\max}-1]$,  $A_i$ is the child of $A_{i-1}$ with the higher volume, and $A_{i_{\max}-1}$ having children $A_{i_{\max}}$ and $B_{i_{\max}}$  such that $|A_{i_{\max}}| \leq 2$ and let 
$\mathcal{S}_i\triangleq \{B_1,\ldots, B_{i_{\max}}, A_{i_{\max}}\}$ be the set of all \emph{extended} critical nodes.

 Based on this  new notion of extended critical nodes, Algorithm~\ref{alg:MergeCrit} gives a formal description of our adjusted tree called   $\treeCC$, which consists of  
  three phases: \texttt{Partition}, \texttt{Prune} and \texttt{Merge}. In the \texttt{Partition} phase~(Lines \ref{alg:PruneMerge:line:begin partition}--\ref{alg:PruneMerge:line:end partition}), the algorithm employs Lemma~\ref{lem:Improved Decomposition} to partition  $V(G)$ into sets $\{C_i\}_{i=1}^{\ell}$, and applies Theorem~\ref{thm:degree} to obtain the corresponding trees $\{\T_i\}_{i=1}^{\ell}$. The \texttt{Prune} phase~(Lines \ref{alg:PruneMerge:line:begin prune}--\ref{alg:PruneMerge:line:end prune}) consists of a repeated pruning process: for every such tree $\T_i$, the algorithm prunes the subtree $\T_i[N]$, where $N \in \mathcal{S}_i$ is the extended critical node closest to the root in $\T_i$, and adds $\T_i[N]$ to $\mathbb{T}$, which is the collection of all the pruned subtrees~(Line \ref{alg:PruneMerge:line:end prune}). The process is repeated
until $\T_i$ is completely pruned.
 Finally, in the \texttt{Merge} phase~(Lines \ref{alg:PruneMerge:line:begin merge}--\ref{alg:PruneMerge:line:end merge})  the algorithm combines the trees in $\mathbb{T}$ in a ``caterpillar style'' according to an increasing order of their size, where we define the size of a tree $\tree$ as $|\tree| \triangleq |\leaves(\tree)|$. We call the final constructed tree $\treeCC$, and the performance of this algorithm is summarised in Lemma~\ref{thm:main_k_clusters}. 
 We emphasise \emph{again} that in our algorithm we do not explicitly run Algorithm~\ref{alg:MergeCrit} to construct our tree. Instead, we use it to show the \emph{existence} of   tree   $\treeCC$, and our final analysis is to  approximate the cost of $\treeCC$.

\begin{algorithm}
    \DontPrintSemicolon
    \KwData{A graph $G = (V, E, w)$, a parameter $k \in \mathbb{Z}_{+}$ such that $\lambda_{k+1} > 0$;}
    \KwResult{An \textsf{HC} tree $\treeCC$ of $G$;}
    
    Apply the partitioning algorithm (Lemma~\ref{lem:Improved Decomposition}) on input $(G, k)$ to obtain $\{C_i\}_{i=1}^{\ell}$ for some $\ell \leq k$;\label{alg:PruneMerge:line:begin partition}
    
    Let $\T_i = \mathrm{\texttt{HCwithDegrees}}(G[C_i])$;\label{alg:PruneMerge:line:end partition}
    
    Initialise $\mathbb{T} = \emptyset$;
    
    \For{All clusters $C_i$}
    {
        \label{alg:PruneMerge:line:begin prune}
        Let $\mathcal{S}_i$ be the set of \textit{extended} critical nodes of $\T_i$;
        
        \For{$N \in \mathcal{S}_i$}
        {
            Update $\mathbb{T} \leftarrow \mathbb{T} \cup \T_i[N]$.\label{alg:PruneMerge:line:end prune}
        }
    }
    
    Let $t = |\mathbb{T}|$, and $\mathbb{T} = \{\widetilde{\T_1}, \dots, \widetilde{\T_t}\}$ be such that $|\widetilde{\T}_i| \leq |\widetilde{\T}_{i+1}|$ for all $1 \leq i < t$;
    
    Initialise $\treeCC = \widetilde{\T}_1$;\label{alg:PruneMerge:line:begin merge}
    
    \For{$i=2,\ldots, t$}
    {
        Let $\treeCC$ be the tree with $\treeCC$ and $\widetilde{\T}_i$ as its two children;\label{alg:PruneMerge:line:end merge}
    }
    
    \Return $\treeCC$

\caption{\texttt{Merge Extended Critical Nodes}\label{alg:MergeCrit}}
\end{algorithm}

 \begin{lemma}\label{thm:main_k_clusters}
Let $G=(V, E, w)$ be a  graph, and $k>1$  such that $\lambda_{k+1}>0$  and $\lambda_k < \lp \frac{1}{270 \cdot c_0 \cdot (k+1)^6} \rp^2$, where $c_0$ is the constant in Lemma~\ref{lem:Upperbound eig induced graphs}. Then, the tree $\treeCC$ returned by Algorithm~\ref{alg:MergeCrit} satisfies $\COST_G(\treeCC) = O\lp k^{22}/\lambda_{k+1}^{10}\rp \cdot \OPT_G$. 

\end{lemma}

 The remaining part of this subsection is to analyse Algorithm~\ref{alg:MergeCrit}, and prove Lemma~\ref{thm:main_k_clusters}. First of all, notice that, since $\{C_i\}_{i=1}^{\ell}$ is the output of our partitioning algorithm, it holds for every $u\in C_i$ for any $1\leq i\leq \ell$ that  
\begin{enumerate}[label=$(A\arabic*)$]
    \item  $\Phi(C_i) = O(k^6 \sqrt{\lambda_k})$;
    \item  $\Phi_{G[C_i]} = \Omega (\lambda_{k+1}^2/k^4)$;
    \item  $w(u, V \setminus C_i) \leq 6(k+1) \cdot \vol_{G[C_i]}(u)$.
\end{enumerate}
Generalising Property~(A3), it is easy to see that it holds for every extended critical node $N\in\mathcal{S}_i~(1\leq i\leq  \ell)$ that 
\begin{enumerate}[label=$(A\arabic*^*)$, leftmargin=25pt]
\setcounter{enumi}{2}   
    \item  $w(N, V \setminus C_i) \leq 6(k+1) \cdot \vol_{G[C_i]}(N)$.
\end{enumerate}
 Moreover, with the parameter  $ \phiIn = \Theta\lp\lambda_{k+1}/k^2\rp$, we have that $\Phi_{G[C_i]} = \Omega(\phiIn^2)$ holds for any $1\leq i\leq \ell$. Let $\T_i = \T_{\deg}(G[C_i])$ with the corresponding  set  of extended critical nodes $\mathcal{S}_i$ for all $1 \leq i \leq \ell$, and   $\mathcal{S} = \bigcup_{i=1}^{\ell} \mathcal{S}_i$ be the set of all extended critical nodes. Now we group the edges of $G$ into two categories: let $E_1$ be the set of edges in the induced subgraphs $G[C_i]$ for all $1\leq i\leq \ell$, i.e.,
\[
    E_ 1 \triangleq  \bigcup_{i=1}^{\ell} E\left[G[C_i] \right],
\]
and   $E_2$ be the remaining crossing edges. Therefore, we   write the cost of our tree $\treeCC$ as
 \begin{equation}\label{eq:Split cost into E_1 and E_2}
    \COST_G(\treeCC) = \sum_{e \in E_1} \cost_{\treeCC}(e) + \sum_{e \in E_2} \cost_{\treeCC}(e),
\end{equation}
 and   analyse each sum individually in Lemmas~\ref{lem:cost edges inside each cluster} and \ref{lem:cost crossing edges}.

 Our first result bounds the size of the parent of an extended critical node $N \in \mathcal{S}_i$ in $\treeCC$, with respect to the size of its parent in the tree $\T_i$. This result will be extensively used when analysing the cost of the edges adjacent to $N$.

 \begin{lemma}\label{lem:Parent size pruned nodes}
It holds for every $1 \leq i \leq \ell$ and every extended critical node $N \in \mathcal{S}_i $ that 
    \[
        \left|\parent_{\treeCC}(N)\right| \leq 6k \cdot \left|\parent_{\T_i}(N)\right|.
    \]  
\end{lemma}
 \begin{proof}
     Suppose the extended dense branch of $\T_i$ is $(A_0, \dots, A_{i_\mathrm{max}-1})$ for some $i_\mathrm{max} \in \mathbb{Z}_{\geq 0}$, with $B_j$ being the sibling of $A_j$ and $A_{i_\mathrm{max}-1}$ having children $A_{i_\mathrm{max}}, B_{i_\mathrm{max}}$. Recall that the set of extended critical nodes is $\mathcal{S}_i = \{ B_1, \dots, B_{i_\mathrm{max}}, A_{i_\mathrm{max}}\}$. By construction of the degree-based tree, we know from Lemma~\ref{lem:critical node properties} that it holds for all $1 \leq j \leq i_\mathrm{max}-1$ that  $|A_j| = 2 \cdot |A_{j+1}|$, which implies that  $|B_{j + 1}| = |A_{j+1}|$ and  $|B_{j}| = 2 \cdot |B_{j+1}|$ for all $j \geq 2$. Thus, we   conclude that for every interval $\lp 2^{s-1}, 2^s\rsp$, for some $s \in \mathbb{Z}_{\geq 0}$, there are at most $3$ critical nodes\footnote{We remark that, in the worst case, all three nodes $B_1, B_{i_\mathrm{max}}$ and $A_{i_\mathrm{max}}$ could have size in $(2^{s-1}, 2^s]$. } $N \in \mathcal{S}_i$ of size $|N| \in \lp 2^{s-1}, 2^s \rsp$. 
    Now let us fix $N \in \mathcal{S}_i$. By construction, we have that
  \begin{equation} \label{eq:claim parent size eq1-new}
        \left|\parent_{\T_i}(N)\right| \geq 2 \cdot |N|.
    \end{equation}
    On the other hand, by the construction of $\treeCC$ we have that
    \begin{align}
        \left|\parent_{\treeCC}(N)\right| 
        &= \sum_{j=1}^{\ell} \sum_{\substack{M \in \mathcal{S}_j \\ |M| \leq |N|}} |M|
        \leq \sum_{j=1}^{\ell} \sum_{s=0}^{\ceil{\log{|N|}}} \sum_{\substack{M \in \mathcal{S}_j \\ 2^{s-1} < |M| \leq 2^s}} |M| \nonumber\\
        &\leq \sum_{j=1}^{\ell} \sum_{s = 0}^{\ceil{\log{|N|}}} 3 \cdot 2^s
        \leq \sum_{j=1}^k 3 \cdot 2^{\ceil{\log{|N|}} + 1} \label{eq:claim parent size eq2-new}
        \leq 12k \cdot |N|.
    \end{align}
    By combining \eqref{eq:claim parent size eq1-new} and \eqref{eq:claim parent size eq2-new}, we have  that
     \[
        \left|\parent_{\treeCC}(N)\right| \leq 12k \cdot |N| \leq 6k \cdot \left|\parent_{\T_i}(N)\right|,
    \]which proves the statement.
\end{proof}

Now we prove the  two main technical lemmas of this subsection.

 \begin{lemma}\label{lem:cost edges inside each cluster}
    It holds that $\sum_{e \in E_1} \cost_{\treeCC}(e) = O\left( k/\phiIn^8 \right) \cdot \OPT_G$.
\end{lemma}

\begin{proof}
     Notice that
 \begin{equation*}\label{eq:lem:cost edges inside each cluster eq1}
        \sum_{e \in E_1} \cost_{\treeCC}(e) = \sum_{i=1}^{\ell} \sum_{e \in E[G[C_i]]} \cost_{\treeCC}(e). 
    \end{equation*}
    
   We   prove that, for every $1\leq i\leq\ell$ and $e\in E[G[C_i]]$, the cost of $e$ in $\treeCC$ and the one in $\T_i$ differ by at most a factor of $O(k)$. Combining this with Theorem~\ref{thm:degree}  will prove the lemma.
    
  To prove this $O(k)$-factor bound, we fix any $1\leq i\leq \ell$ and let $\mathcal{S}_i$ be the set of extended critical nodes of $C_i$. As the nodes of $\mathcal{S}_i$ form a partition of the vertices of $G[C_i]$, any edge $e\in E[G[C_i]]$  satisfies exactly one of the following conditions: (i) $e$ is inside a critical node; (ii) $e$ is adjacent to a critical node. Formally, it holds that 
    \begin{equation*}
         \sum_{e \in E[G[C_i]]} \cost_{\treeCC}(e) =  
        \sum_{\substack{N \in \mathcal{S}_i \\ e \in E(N,N)}} \cost_{\treeCC}(e) 
        + \sum_{\substack{N \in \mathcal{S}_i \\ M \in \mathcal{S}_i \setminus \{N\} \\ \left|\parent_{\T_i}(M) \right| \leq \left|\parent_{\T_i}(N) \right|}} \sum_{e \in E(N, M)} \cost_{\treeCC}(e) 
    \end{equation*}
We first examine the case in which the cost of $e$ in both trees are the same, i.e., Case~(i). Since   we do not change the structure of the tree inside any critical node, it holds that
    \begin{align*}
        \sum_{\substack{N \in \mathcal{S}_i \\ e \in E(N,N)}} \cost_{\treeCC}(e) &= \sum_{\substack{N \in \mathcal{S}_i \\ e \in E(N,N)}} \cost_{\T_i}(e).
    \end{align*}
 For Case~(ii), the cost of any such edge increases by at most a factor of  $O(k)$  due to Lemma~\ref{lem:Parent size pruned nodes} and the construction of $\treeCC$. 
  Formally, let $N \in \mathcal{S}_i$ be an arbitrary extended critical node, and  $M \in \mathcal{S}_i \setminus\{N\}$ be an extended critical node such that $\left|\parent_{\T_i}(M) \right| \leq \left|\parent_{\T_i}(N) \right|$. Firstly, notice that if $\parent_{\T_i}(N)$ is the root node of $\T_i$, then for any edge $e \in E(N, M)$ it holds that $\cost_{\T_i}(e) = w_e \cdot |\T_i|.$ On the other hand, by the construction of $\treeCC$, we know that $\cost_{\treeCC}(e) \leq 6k \cdot w_e \cdot |\T_i|$, so we conclude that $\cost_{\treeCC}(e) \leq 6k \cdot \cost_{\T_i}(e)$. Secondly, if $\parent_{\T_i}(N)$ is not the root node of $\T_i$ and since $\left|\parent_{\T_i}(M) \right| \leq \left|\parent_{\T_i}(N) \right|$, we know that $|M| \leq |N|$. Therefore it holds for any edge $e \in E(N, M)$ that
    \[
        \cost_{\treeCC}(e) 
        = w_e \cdot \left | \parent_{\treeCC}(N)\right |
        \leq 6k \cdot w_e \cdot \left | \parent_{\T_i}(N)\right |
        = 6k \cdot \cost_{\T_i}(e),
    \]
    where the inequality follows by Lemma~\ref{lem:Parent size pruned nodes}.
    Combining the above facts, we have  that 
 \begin{equation}\label{eq:lem:cost edges inside each cluster eq3-new}
        \sum_{e \in E_1} \cost_{\treeCC}(e)
        \leq \sum_{i=1}^{\ell} 6k \cdot \sum_{e \in E[G[C_i]]} \cost_{\T_i}(e)
        \leq \sum_{i=1}^{\ell} 6k \cdot \COST_{G[C_i]}(\T_i).
    \end{equation}
     On the other side, let $\mathcal{T}^*$  be any optimal \textsf{HC} tree of $G$ with cost $\OPT_G$, and it holds that
\begin{equation}\label{eq:lem:cost edges inside each cluster eq4-new}
    \OPT_G  \geq  \sum_{i=1}^{\ell} \OPT_{G[C_i]} = \sum_{i=1}^{\ell} \Omega \lp \Phi_{G[C_i]}^4\rp \cdot \COST(\T_i) 
    = \sum_{i=1}^{\ell} \Omega \lp \phiIn^8 \rp \cdot \COST(\T_i),
\end{equation}
     where the last equality follows by Property~(A2)  of Lemma~\ref{lem:Improved Decomposition} and Theorem~\ref{thm:degree} applied to every $G[C_i]$. Finally, by combining  \eqref{eq:lem:cost edges inside each cluster eq3-new} and \eqref{eq:lem:cost edges inside each cluster eq4-new} we have that
 \[
        \sum_{e \in E_1} \cost_{\treeCC}(e) = O(k/\phiIn^8) \cdot \OPT_G,
    \]
 which proves the lemma.
\end{proof}

 \begin{lemma}\label{lem:cost crossing edges}
It holds that \[\sum_{e \in E_2} \cost_{\treeCC}(e) = O\left( k^2/\phiIn^{10} \right) \cdot \OPT_G.\]
\end{lemma}
\begin{proof}
   For the edges $e \in E_2$, we   bound the cost with the help of Lemma~\ref{lem:Parent size pruned nodes} similar as before. Specifically, we have that
   \begin{align}
        \sum_{e \in E_2} \cost_{\treeCC}(e)  & \leq \sum_{i=1}^{\ell} \sum_{\substack{N \in \mathcal{S}_i \\ M \in \mathcal{S} \setminus \mathcal{S}_i \\ |M| \leq |N|}} \sum_{\substack{e \in E(N, M)}} \cost_{\treeCC}(e) \nonumber\\
        & \leq \sum_{i=1}^{\ell} \sum_{N \in \mathcal{S}_i}  \left| \parent_{\treeCC}(N)\right| \cdot w(N, V\setminus C_i) \nonumber\\
        & \leq \sum_{i=1}^{\ell} \sum_{N \in \mathcal{S}_i}  36 k(k+1) \cdot\left| \parent_{\T_i}(N)\right| \cdot \vol_{G[C_i]}(N) \label{eq:lem:cost crossing edges eq2}\\
        & \leq 36k(k+1) \cdot \sum_{i=1}^{\ell} \frac{4}{\Phi_{G[C_i]}} \cdot \COST_{G[C_i]}(\T_i) \label{eq:lem:cost crossing edges eq3}\\
        &= O(k^2) \cdot \sum_{i=1}^{\ell} \frac{1}{\Phi_{G[C_i]}^5} \cdot \OPT_{G[C_i]} \label{eq:lem:cost crossing edges eq4}\\
        & = O(k^2/\phiIn^{10}) \cdot \OPT_G, \nonumber
    \end{align}
    where \eqref{eq:lem:cost crossing edges eq2} follows from Property $(A3^*)$ of Lemma~\ref{lem:Improved Decomposition} and Lemma~\ref{lem:Parent size pruned nodes}, \eqref{eq:lem:cost crossing edges eq3} follows by Lemma~\ref{lem:Cost light nodes}, and \eqref{eq:lem:cost crossing edges eq4} follows by Theorem~\ref{thm:degree} applied to every induced subgraph $G[C_i]$.
\end{proof}

 Finally, we are ready to prove  Lemma~\ref{thm:main_k_clusters}.

\begin{proof}[Proof of Lemma~\ref{thm:main_k_clusters}]
     Let $C_1, \ldots C_\ell$ be the partitioned returned by the algorithm from Lemma \ref{lem:Improved Decomposition}. Let $\mathcal{S} \triangleq \bigcup_{i=1}^k \mathcal{S}_i$ be the set of all critical nodes, and   $\tree_{\mathrm{CC}}$ be the ``caterpillar'' style tree on the critical nodes according to an increasing order of their sizes. We have that 
     \begin{align*}
         \COST_G(\treeCC) 
        &= \sum_{e \in E_1} \cost_{\treeCC}(e) + \sum_{e \in E_2} \cost_{\treeCC}(e)\\
        &= O(k/\phiIn^8) \cdot \OPT_G + O(k^2/\phiIn^{10}) \cdot \OPT_G\\& 
        = O(k^2/\phiIn^{10}) \cdot \OPT_G \\
        & = O(k^{22}/\lambda_{k+1}^{10}) \cdot \OPT_G,
    \end{align*} where the second equality follows by Lemmas~\ref{lem:cost edges inside each cluster} and \ref{lem:cost crossing edges}, and the last equality follows by the definition of $\phiIn$. 
\end{proof}

\subsection{Proof of Lemmas~\ref{lemma:upper bound inner nodegree} and \ref{lemma:upper bound WOPT nodegree}}\label{sec:keylemourtechniques}

 Now we  prove   Lemma~\ref{lemma:upper bound inner nodegree} and Lemma~\ref{lemma:upper bound WOPT nodegree}, which allow us to approximate the cost of $\treeCC$. To sketch the main proof idea for Lemma~\ref{lemma:upper bound WOPT nodegree}, we construct a tree $\treeCC''$ and upper bound its weighted Dasgupta cost on the bucketing $\C$, i.e., $\WCOST_{G / \buckets}\lp \treeCC'' \rp$, with respect to $\OPT_G$. To construct $\treeCC''$, we start with the tree $\treeCC$ on $G$, whose upper bound with respect to $\OPT_G$ is known  because of Lemma~\ref{thm:main_k_clusters}. We then transform $\treeCC$ into $\treeCC''$ in two steps, such that there is only a small induced cost for both transformations. This allows us to upper bound $\COST_G(\treeCC'')$ with respect to $\COST_G(\treeCC)$. Crucially, we can approximate the tree $\treeCC$ with Algorithm~\ref{alg:spectral_clustering_degree_bucketing} \emph{without}   explicitly constructing tree $\treeCC$. The proof for Lemma~\ref{lemma:upper bound inner nodegree} follows from the proof for Lemma~\ref{lemma:upper bound WOPT nodegree}.

 In order to analyse the constructed trees $\treeCC'$ and $\treeCC''$, we analyse the properties of the buckets in the context of the entire graph $G$, and not just on every set $P_i$ returned by spectral clustering.
We denote by $\buckets$ the set containing all the buckets obtained throughout all sets $P_i$, i.e., 
 \[
    \buckets \triangleq \bigcup_{i=1}^k \buckets_i;
\]
 notice that   $\buckets$ is a partition of $V$. We use
 \[
    \kappa \triangleq |\buckets|
\]
 to denote the total number of buckets. For convenience, we label the buckets $\buckets = \{B_1, \dots, B_{\kappa}\}$   arbitrarily, and in general we always refer to a bucket $B_j \in \buckets$ with the subscript $j$. 

 Next,  we also analyse the properties of the extended critical nodes in the context of the entire graph $G$, not just from every set $\mathcal{S}_i$. Therefore, let 
 \[
\mathcal{S} = \bigcup_{i=1}^{\ell} \mathcal{S}_i
\] 
 be the set of all extended critical nodes, and we label them as $\mathcal{S}=\{N_1, \ldots N_{k_1}\}$ in the order they appear in if one would travel upwards starting from the bottom of the caterpillar tree $\treeCC$. In general, we always refer to an extended critical node $N_i \in \mathcal{S}$ with the subscript $i$. We use $C(N_i)$  to denote the original cluster that $N_i$ belongs to, i.e, the cluster $C(N_i)$ is the $C_s $ such that $N_i \in \mathcal{S}_s$. We also use $\T_{\mathrm{deg}}(N_i)$ to denote the tree $\T_{\mathrm{deg}}$ that $N_i$ originally corresponds to. To avoid confusion, we remind the reader that $\T_i$ is the degree-based tree constructed on the induced graph $G[C_i]$, i.e., $\T_i = \T_\mathrm{deg}(G[C_i])$. 

Next, we define 
\[
X_j \triangleq \{ N_i \in \mathcal{S}| N_i \cap B_j \neq \emptyset \}
\] as the set of critical nodes that have non-empty intersection  with   bucket $B_j \in \buckets$. Note   that the subscript of $X_j$ matches the one of the corresponding bucket $B_j$. We then define the \emph{anchor} (critical) node $N_{i'} \in \mathcal{S}$ of the set $X_j$ as the largest critical node inside $X_j$. i.e., we have that 
     \[
    N_{i'} \triangleq \argmax_{N_{i''} \in X_j} |N_{i''}|,
    \] where ties are broken by choosing $N_{i'}$ to be the critical node highest up in $\treeCC$. These anchor nodes corresponding to the sets $X_j$ -- and therefore the buckets $B_j \in \buckets$ -- are the key in our proof of the approximation factor of our algorithm; they determine the locations in the tree $\treeCC$ where we group together the buckets $B_j \in \buckets$. It is also important to note that for every $X_j$ its anchor node is determined \emph{uniquely}. On the other hand, any $N_i \in \mathcal{S}$ might be the anchor node for multiple $X_j$'s. 

Now we're ready to describe the two transformations.

  \emph{Step~1: Splitting the Critical Nodes.} 
 In this first step, we transform the tree $\treeCC$ such that all the sets $N_i \cap B_j$ for $1 \leq i \leq k_1$ and $1 \leq j \leq \kappa$ are separated and restructured to form a caterpillar tree. This is done by splitting every $N_i$ into $N_i \cap B_j$ for all $1 \leq j \leq \kappa$ and appending them next to each other arbitrarily in the caterpillar tree in the original location of $N_i$ in $\treeCC$.

 More formally, we   consider the partition of every $N_i \in \mathcal{S}$ into $\{N_i \cap B_j\}_{j=1}^\kappa$. Then, for every internal node $\parent(N_i)$ in the tree $\treeCC$, we create new internal nodes $\parent(N_i \cap B_{j'})$ for every $N_i \cap B_{j'} \in \{N_i \cap B_j\}_{j=1}^\kappa$, and   let one of the children of each of these be a new internal node corresponding to $N_i \cap B_{j'}$, and let the other child be another node $\parent(N_i \cap B_{j''})$ for $j' \neq j''$. Moreover, we ensure that the ordering (starting from the bottom) of the new nodes $\parent(N_i \cap B_j)$ along the back of the caterpillar tree will be such that the first all the ones corresponding to $N_1$ will appear, then $N_2$ etc., up until $N_{k_1}$.

 We call the resulting tree $\treeCC'$, whose cost is bounded in the lemma below.
 
 \begin{lemma}\label{lem:critical node split}
    It holds that
    \[
    \COST_G\lp\treeCC'\rp \leq \COST_G\lp\treeCC\rp + O(k/\phiIn^{10}) \cdot \OPT_G.
    \]
\end{lemma}
\begin{proof}
     First of all, notice that for any edge crossing critical nodes, i.e., $e\in E(N_i, N_{i'})$ for $i\neq i'$,  we have that
    \[
    \cost_{\treeCC'}(e) \leq \cost_{\treeCC}(e).
    \]
     This is because splitting up the critical nodes cannot increase the size of the lowest common ancestor of any crossing edge between critical nodes, due to the caterpillar tree construction. So we don't need to consider the increase in the cost of crossing edges between critical nodes.

 Next, we study the edges inside critical nodes $e\in E(N_i, N_i)$ for any $i$. These edges' cost could be increased compared with the one in $\treeCC$, since they can span across nodes $N_i \cap B_j$. However, the increase in size of their lowest common ancestor is by construction \emph{at most} $|\parent_{\treeCC}(N_i)|$, and therefore the total cost of this transformation can be upper bounded as
    \begin{align}
    \sum_{i=1}^{k_1} \sum_{\substack{e \in E(N_i, N_i)}} \cost_{\treeCC'}(e)     \leq &\sum_{i=1}^{k_1} \sum_{\substack{e \in E(N_i, N_i)}} w(e) \cdot \left| \parent_{\treeCC}(N_i)\right| \nonumber\\
         \leq &\sum_{i=1}^{k_1} 6 k \cdot\left| \parent_{\T_{\mathrm{deg}}(N_i)}(N_i)\right| \cdot \vol_{G[C(N_i)]}(N_i) \label{eq:lem:simple cost crossing edges eq2}\\
        \leq &6k \cdot \sum_{s=1}^{\ell} \frac{4}{\Phi_{G[C_s]}} \cdot \COST_{G[C_s]}(\T_s) \label{eq:lem:simple cost crossing edges eq3}\\
        = &O(k) \cdot \sum_{s=1}^{\ell} \frac{1}{\Phi_{G[C_s]}^5} \cdot \OPT_{G[C_s]} \label{eq:lem:simple cost crossing edges eq4}\\
        = &O(k/\phiIn^{10}) \cdot \OPT_G, \nonumber
    \end{align}
   where  \eqref{eq:lem:simple cost crossing edges eq2} follows from Lemma~\ref{lem:Parent size pruned nodes}, \eqref{eq:lem:simple cost crossing edges eq3} follows from Lemma~\ref{lem:Cost light nodes}, and \eqref{eq:lem:simple cost crossing edges eq4} follows from Theorem~\ref{thm:degree} applied to every induced subgraph $G[C_i]$. This proves the statement.  

\end{proof}

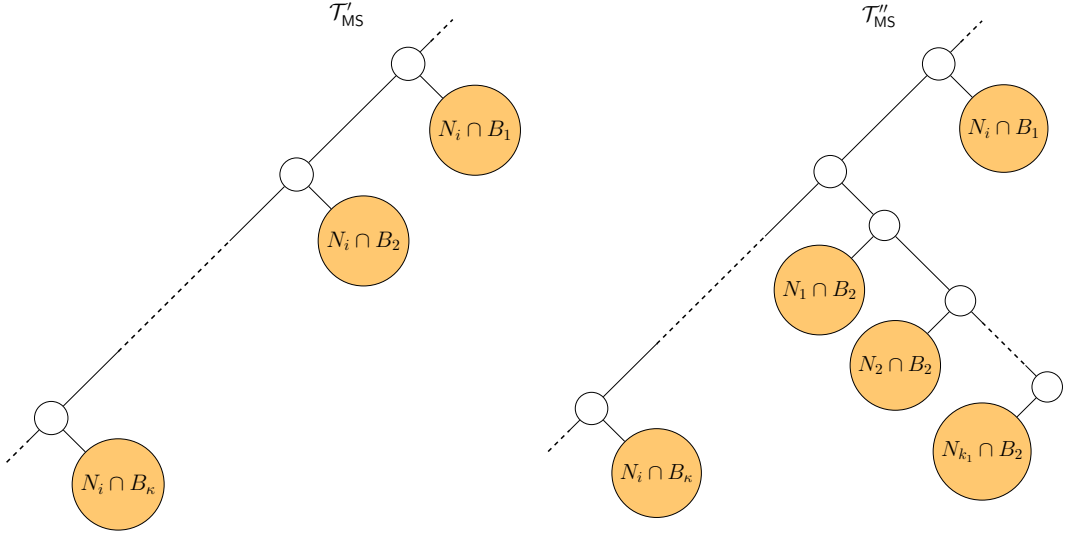
\begin{figure}[htbp]
    \centering
    \resizebox{7cm}{!}{%
\begin{tikzpicture}
	\begin{pgfonlayer}{nodelayer}
		\node [style=minsize] (0) at (0, 0) {};
		\node [style=minsize] (1) at (-5, -5) {};
		\node [style=none] (8) at (-2.75, 2.25) {\Large{$\mathcal{T}_\mathsf{MS}'$}};
		\node [style=none] (9) at (1, 1) {};
		\node [style=none] (10) at (2, 2) {};
		\node [style={caterpillar_nodes}] (11) at (3, -3) {\Large{$N_i \cap B_1$}};
		\node [style={caterpillar_nodes}] (12) at (-2, -8) {\Large{$N_i \cap B_2$}};
		\node [style=none] (13) at (-13, -13) {};
		\node [style=minsize] (14) at (-16, -16) {};
		\node [style=none] (16) at (-13, -13) {};
		\node [style=none] (17) at (-8, -8) {};
		\node [style={caterpillar_nodes}] (19) at (-13, -19) {\Large{$N_i \cap B_{\kappa}$}};
		\node [style=none] (20) at (-17, -17) {};
		\node [style=none] (21) at (-18, -18) {};
	\end{pgfonlayer}
	\begin{pgfonlayer}{edgelayer}
		\draw (0) to (1);
		\draw (0) to (9.center);
		\draw [style=undirected dashed] (9.center) to (10.center);
		\draw (0) to (11);
		\draw (1) to (12);
		\draw (13.center) to (14);
		\draw (13.center) to (16.center);
		\draw [style=undirected dashed] (16.center) to (17.center);
		\draw [in=135, out=-45] (14) to (19);
		\draw (17.center) to (1);
		\draw (20.center) to (14);
		\draw [style=undirected dashed] (21.center) to (20.center);
	\end{pgfonlayer}
\end{tikzpicture}
    }
    \resizebox{7cm}{!}{%
\begin{tikzpicture}
	\begin{pgfonlayer}{nodelayer}
		\node [style=minsize] (0) at (0, 0) {};
		\node [style=minsize] (1) at (-5, -5) {};
		\node [style=none] (8) at (-2.75, 2.25) {\Large{$\mathcal{T}_\mathsf{MS}''$}};
		\node [style=none] (9) at (1, 1) {};
		\node [style=none] (10) at (2, 2) {};
		\node [style={caterpillar_nodes}] (11) at (3, -3) {\Large{$N_i \cap B_1$}};
		\node [style=smallCircle] (12) at (-2.5, -7.5) {};
		\node [style=none] (13) at (-13, -13) {};
		\node [style=minsize] (14) at (-16, -16) {};
		\node [style=none] (16) at (-13, -13) {};
		\node [style=none] (17) at (-8, -8) {};
		\node [style={caterpillar_nodes}] (19) at (-13, -19) {\Large{$N_i \cap B_{\kappa}$}};
		\node [style={caterpillar_nodes}] (20) at (-2, -14) {\Large{$N_2 \cap B_2$}};
		\node [style={caterpillar_nodes}] (21) at (-5.5, -10.5) {\Large{$N_1 \cap B_2$}};
		\node [style={caterpillar_nodes}] (22) at (2, -18) {\Large{$N_{k_1} \cap B_2$}};
		\node [style=smallCircle] (23) at (1, -11) {};
		\node [style=smallCircle] (24) at (5, -15) {};
		\node [style=none] (25) at (2, -12) {};
		\node [style=none] (26) at (4, -14) {};
		\node [style=none] (27) at (-17, -17) {};
		\node [style=none] (28) at (-18, -18) {};
	\end{pgfonlayer}
	\begin{pgfonlayer}{edgelayer}
		\draw (0) to (1);
		\draw (0) to (9.center);
		\draw [style=undirected dashed] (9.center) to (10.center);
		\draw (0) to (11);
		\draw (1) to (12);
		\draw (13.center) to (14);
		\draw (13.center) to (16.center);
		\draw [style=undirected dashed] (16.center) to (17.center);
		\draw [in=135, out=-45] (14) to (19);
		\draw (17.center) to (1);
		\draw (12) to (23);
		\draw (23) to (25.center);
		\draw (26.center) to (24);
		\draw [style=undirected dashed] (25.center) to (26.center);
		\draw (21) to (12);
		\draw (20) to (23);
		\draw (22) to (24);
		\draw (27.center) to (14);
		\draw [style=undirected dashed] (28.center) to (27.center);
	\end{pgfonlayer}
\end{tikzpicture}
    }
    \caption{Visualisation of a part of the transformation described in step 2, where the buckets are grouped. On the left we visualise a section of the tree $\treeCC'$, corresponding to a single critical node $N_i$. For illustration, we assume here that $N_i$ is the anchor node for $X_2$. This means that all $N_{i'} \cap B_2$ for $N_{i'} \in X_2 \setminus N_j$ are moved up into the same subtree,  as illustrated on the right in the new tree $\treeCC''$.}\label{fig:step2_transform}
\end{figure}

\emph{Step 2: Grouping Degree Buckets.}
 In this second step, we   move around the sets $N_i \cap B_j$, such that  the buckets $B_1, \ldots ,B_{\kappa}$  are separated after this step.

We do this by traversing the tree $\treeCC'$ from the bottom to the top, and visiting every node $\parent(N_i \cap B_j)$ up along the ``back'' of the caterpillar construction. This means we first visit all the nodes $\parent(N_i \cap B_j)$ corresponding to $N_1$, then $N_2$ etc. up to $N_{k_1}$.  
When visiting $\parent(N_i \cap B_j)$ we check whether $N_i$ is an anchor node of $X_j$. If not, we continue travelling up. \emph{If it is}, we transform the tree $\treeCC'$ by moving up all the internal nodes $N_{i'} \cap B_j$ for $N_{i'} \in X_j \setminus N_i$ and  placing them in the same subtree as the internal node corresponding to $N_i \cap B_j$. We do so by creating a new root node $R_{B_j}$, and we loop through each $N_{i'} \cap B_j$ for $N_{i'} \in X_j$ in an arbitrary order. In the first iteration, we set the children of $R_{B_j}$ to be $N_{i'} \cap B_j$, and a newly created node denoted by $z$. We iteratively do this, setting the two children of $z$ in the same way as for the root node, such that we construct a caterpillar tree on the nodes $N_{i'} \cap B_j$ for $N_{i'} \in X_j$. Finally, after this construction, we replace the child node $N_i \cap B_j$ of $\parent(N_i \cap B_j)$ with the root node $R_{B_j}$ of our newly constructed caterpillar tree. 

 Note that every $\parent(N_{i'} \cap B_j)$ for $N_{i'} \in X_j \setminus N_i$ no longer has a corresponding $N_{i'} \cap B_j$ as a child node, since it has been placed in the new caterpillar tree. Therefore, as a final step, we remove $\parent(N_{i'} \cap B_j)$ from the tree $\treeCC'$, and appropriately update the child and parent node directly below and above $\parent(N_{i'} \cap B_j)$ to ensure that the tree stays connected. 

 After travelling up the entire tree, this transformation ensures that all the vertices in $B_j$ are placed in the same subtree, since $\bigcup_{N_{i'} \in X_j} N_{i'} \cap B_j = B_j$.  Crucially, because $N_i$ is an anchor node, we know by the construction of the tree $\treeCC$ that every $N_{i'} \in X_j \setminus N_i$ is placed lower in the tree $\treeCC'$, and hence we do not have to consider the induced cost of moving any critical node higher up in the tree down to $N_i$. To illustrate, We visualise (part of) the transformation in Figure~\ref{fig:step2_transform}. We call the resulting tree $\treeCC''$, and we bound its cost with the following lemma. 

\begin{lemma}\label{lem:critical node movement}
    It holds that
    \[
    \COST_G\lp\treeCC''\rp \leq \COST_G\lp\treeCC'\rp + O(\beta \cdot k^3/\phiIn^{10}) \cdot \OPT_G.
    \]
\end{lemma}
\begin{proof}
We   bound the cost of the transformation by looking at the total induced cost incurred at every $N_i \in \mathcal{S}$ that is an anchor node. First of all, notice that for all nodes $N_i \cap B_j$ in $\treeCC'$ we have that
\begin{equation}\label{eq:upperbound_size_LCA_splitbuckets}
|\parent_{\treeCC'}(N_i \cap B_j)| \leq |\parent_{\treeCC}(N_i)| \leq 6k \cdot |\parent_{\T_{\mathrm{deg}}(N_i)}(N_i)|,
\end{equation}
where  the first inequality holds because of the construction of $\treeCC'$ and the second inequality follows by Lemma~\ref{lem:Parent size pruned nodes}.  Recall that   $\T_{\mathrm{deg}}(N_i)$ is  the tree $\T_{\mathrm{deg}}$ that $N_i$ originally corresponds to.

 Now let us fix some $N_i$ that is also an anchor node, and let $\mathcal{X}_i$ be the set such that $X_j \in \mathcal{X}_i$ shares $N_i$ as their anchor node, i.e., $N_i = \argmax_{N_{i'} \in X_j} |N_{i'}| $ for $X_j \in \mathcal{X}_i$. The only edges whose cost can increase when moving up $N_{i'} \cap B_j$ for some $N_{i'} \in X_j$ will be the edges with one endpoint in $N_{i'} \cap B_j$, since moving them up potentially increases their lowest common ancestor to be $|\parent_{\treeCC'}(N_i \cap B_{j})|$. Note that because we carefully choose to perform the transformations at the \emph{anchor} node, we don't introduce any new vertices from higher up in the tree which could potentially change the size of the lowest common ancestor. 

Therefore, we can bound the cost of moving up all $N_{i'} \cap B_j$ for $N_{i'} \in X_j$ and $X_j \in \mathcal{X}_i$ as  
\begin{align}
    & \sum_{X_j \in \mathcal{X}_i} \left( |\parent_{\treeCC'}(N_i \cap B_{j})|    \sum_{N_{i'} \in X_j} \Vol{N_{i'} \cap B_{j}}  \right)\nonumber \\
    \leq & 6k \cdot | \parent_{\T_{\mathrm{deg}}(N_i)}(N_i) | \cdot\sum_{X_{j} \in \mathcal{X}_i} \sum_{N_{i'} \in X_j} \Vol{N_{i'} \cap B_{j}} \label{eq:merge_buckets_nodegree:step05} \\
    \leq & 6k \cdot | \parent_{\T_{\mathrm{deg}}(N_i)}(N_i) | \cdot\sum_{X_{j} \in \mathcal{X}_i} \sum_{N_{i'} \in X_j} |N_{i'} \cap B_{j}| \cdot \dmax_G(N_{i'} \cap B_{j}) \nonumber \\
    \leq & 6k \cdot | \parent_{\T_{\mathrm{deg}}(N_i)}(N_i) | \sum_{X_{j} \in \mathcal{X}_i} \sum_{N_{i'} \in X_j} |N_{i'} \cap B_{j}| \cdot \beta \cdot \dmax_G(N_i \cap B_{j}) \label{eq:merge_buckets_nodegree:step075}\\
    \leq & 6k \cdot | \parent_{\T_{\mathrm{deg}}(N_i)}(N_i) | \cdot \sum_{X_{j} \in \mathcal{X}_i} \sum_{N_{i'} \in X_j} |N_{i'} \cap B_{j}| \cdot \beta \cdot \dmax_G(N_i) \nonumber\\
    \leq & 6k\cdot \beta \cdot | \parent_{\T_{\mathrm{deg}}(N_i)}(N_i) | \cdot 14k \cdot \dmax_{G[C(N_i)]}(N_i) \cdot \sum_{X_j \in  \mathcal{X}_i} \sum_{N_{i'} \in X_j} |N_{i'}| \label{eq:merge_buckets_nodegree:step2}\\
    \leq & O(k^3 \cdot \beta) \cdot | \parent_{\T_{\mathrm{deg}}(N_i)}(N_i) | \cdot \dmax_{G[C(N_i)]}(N_i) \cdot | \parent_{\T_{\mathrm{deg}}(N_i)}(N_i) |\label{eq:merge_buckets_nodegree:step3}\\
    \leq & O(k^3 \cdot \beta) \cdot | \parent_{\T_{\mathrm{deg}}(N_i)}(N_i) | \cdot \dmin_{G[C(N_i)]}(\sibling_{\T_{\mathrm{deg}}(N_i)}(N_i)) \cdot | \parent_{\T_{\mathrm{deg}}(N_i)}(N_i) |\label{eq:merge_buckets_nodegree:step35}\\
    \leq & O(k^3 \cdot \beta) \cdot | \parent_{\T_{\mathrm{deg}}(N_i)}(\sibling_{\T_{\mathrm{deg}}(N_i)}(N_i)) | \cdot \mathrm{vol}_{G[C(N_i)]}\lp\sibling_{\T_{\mathrm{deg}}(N_i)}(N_i)\rp\label{eq:merge_buckets_nodegree:step4}.
\end{align}
Here,  \eqref{eq:merge_buckets_nodegree:step05} follows by \eqref{eq:upperbound_size_LCA_splitbuckets}, \eqref{eq:merge_buckets_nodegree:step075} follows by the fact that all degrees $d_u$ of $ u \in N_{i'} \cap B_j$ for $N_{i'} \in X_j$ lie within a factor of $\beta$ of each other.
Since   we denote by $C(N_i)$ the cluster  to which $N_i$ originally corresponds,  \eqref{eq:merge_buckets_nodegree:step2} holds because of Property~(A3) of Lemma~\ref{lem:Improved Decomposition}, 
\eqref{eq:merge_buckets_nodegree:step3} holds because we only move the  critical nodes $N_{i'}$ that are below  $N_i$ in the caterpillar construction, so the sum of their sizes is at most $|\parent_{\treeCC}(N_i)|$, which we upper bound with~\eqref{eq:upperbound_size_LCA_splitbuckets}. Moreover, \eqref{eq:merge_buckets_nodegree:step35} follows by  the fact that the minimum degree in $G[C(N_i)]$ of a sibling node of $N_i$ (Definition~\ref{def:sibling}) is at least the maximum degree of $N_i$ in $G[C(N_i)]$, 
\eqref{eq:merge_buckets_nodegree:step4} holds because 
\[
|\parent_{\T_{\mathrm{deg}}(N_i)}(N_i)| \leq 2 \cdot | \parent_{\T_{\mathrm{deg}}(N_i)}(\sibling_{\T_{\mathrm{deg}}(N_i)}(N_i)) |
\]
and 
\[
\dmin_{G[C(N_i)]}(\sibling_{\T_{\mathrm{deg}}(N_i)}(N_i)) \cdot | \parent_{\T_{\mathrm{deg}}(N_i)}(N_i)| \leq 4 \cdot \mathrm{vol}_{G[C(N_i)]}\lp\sibling_{\T_{\mathrm{deg}}(N_i)}(N_i) \rp. \]

 The bound above holds as long as $N_i$ has a sibling node, and it remains to examine the case in which $N_i$  doesn't have a sibling node.  Recall that the set of extended critical nodes of a cluster $C_s$ returned by Lemma~\ref{thm:main_k_clusters} is $\mathcal{S}_s\triangleq \{B_1,\ldots, B_{s_{\max}}, A_{s_{\max}}\}$. The only critical node \emph{without} a sibling node is $A_{s_{\max}}$, and  we know by construction  that $|A_{s_{\max}}| \leq 2$. Because $\treeCC$ is a caterpillar tree based on the sizes of the nodes, $A_{s_{\max}}$ is   placed as one of the critical nodes furthest down the tree. Given that there are at most three critical nodes of size at most two in every $C_s$, $A_{s_{\max}}$ must be one of the last $3 \ell$ extended critical nodes in the caterpillar tree. 

Suppose one of these $A_{s_{\max}}$ from some $\mathcal{S}_{s}$ is an anchor node. The only nodes in the tree $\treeCC'$ that can be moved up to this node are other nodes of size at most two, of which there are at most $3 \ell$ (three for every partition $C_s$). Let $N_q \triangleq \argmax_{\substack{N_{i'} \in \mathcal{S} \\ |N_{i'}| \leq 2}} \Vol{N_{i'}}$ be the critical node of size at most two with the largest volume. Then,  all the transformations (at most $3k$) occurring in the lower part of the tree $\treeCC$ can be upper bounded by

\begin{align}
    \lefteqn{3k \lp  3k \cdot 2 \cdot \sum_{\substack{i' \\ |N_{i'}| \leq 2} } \vol_G(N_{i'})\rp} \nonumber \\
    = &  O(k^3)   \cdot \vol_G(N_q)  \nonumber\\
    = &O(k^4) \cdot \mathrm{vol}_{G[C(N_q)]}\lp N_q\rp \nonumber\\
    =& O(\beta \cdot k^3) \cdot | \parent_{\T_\mathrm{deg}(N_q)}(N_q) | \cdot \mathrm{vol}_{G[C(N_q)]}\lp N_q\rp, \label{small_crit_bound_step3}
\end{align}
where the first equality   holds as $N_q$ is the largest volume extended critical node of size at most two, the second one   holds by Property~(A3) of Lemma~\ref{lem:Improved Decomposition}, and the last one holds because $| \parent_{\T_\mathrm{deg}(N_q)}(N_q) | \leq 4$ and the fact that $k \leq 2^{k(\gamma+1)} = \beta$ for $k \geq 1$ and $\gamma > 0$.

Finally, to upper bound the total cost of the transformation, we assume for the upper bound that every extended critical node is an anchor node, and we split the sum over the induced cost of all the anchor nodes that have critical siblings~\eqref{eq:merge_buckets_nodegree:step4}, and the transformation in the bottom of the tree corresponding to critical nodes of size at most 2, which don't have critical siblings~\eqref{small_crit_bound_step3}. Hence, the total cost of the transformation is at most
\begin{align}
& O(\beta \cdot k^3) \cdot \sum_{\substack{i=1 \\ |N_i| > 2}}^{k_1} | \parent_{\T_{\mathrm{deg}}(N_i)}(\sibling_{\T_{\mathrm{deg}}(N_i)}(N_i)) | \cdot \mathrm{vol}_{G[C(N_i)]}\lp\sibling_{\T_{\mathrm{deg}}(N_i)}(N_i) \rp + \nonumber\\
&\qquad O(\beta \cdot k^3) \cdot | \parent_{\T_\mathrm{deg}(N_q)}(N_q) | \cdot \mathrm{vol}_{G[C(N_q)]}\lp N_q\rp \nonumber\\
    \leq & O(\beta \cdot k^3) \cdot\sum_{i=1}^{k_1} \left| \parent_{\T_{\mathrm{deg}}(N_i)}(N_i)\right| \cdot \vol_{G[C(N_i)]}(N_i) \label{eq:lem:final simple cost crossing edges eq2} \\
        \leq &O(\beta \cdot k^3) \cdot \sum_{s=1}^{\ell} \frac{4}{\Phi_{G[C_s]}} \cdot \COST_{G[C_s]}(\T_s) \label{eq:lem:final simple cost crossing edges eq3}\\
        = &O(\beta \cdot k^3) \cdot \sum_{s=1}^{\ell} \frac{1}{\Phi_{G[C_s]}^5} \cdot \OPT_{G[C_s]} \label{eq:lem:final simple cost crossing edges eq4}\\
        = &O(\beta \cdot k^3/\phiIn^{10}) \cdot \OPT_G. \nonumber
\end{align}
Here, \eqref{eq:lem:final simple cost crossing edges eq2} holds because
 the first sum of \eqref{eq:lem:final simple cost crossing edges eq2} is over $| \parent_{\T_{\mathrm{deg}}(N_i)}(N_i)| \cdot \vol_{G[C(N_i)]}(N_i)$ for all $1 \leq i \leq k_1$ \emph{except} for the extended critical nodes that are not the sibling node of any critical node. Hence, we   upper bound the sum by including these terms, and relabeling the indices.
Moreover, 
\eqref{eq:lem:final simple cost crossing edges eq3} follows by Lemma~\ref{lem:Cost light nodes}, and \eqref{eq:lem:final simple cost crossing edges eq4} follows by Theorem~\ref{thm:degree} applied to every induced subgraph $G[C_s]$. This proves the statement.
\end{proof}

\begin{proof}[Proof of Lemma~\ref{lemma:upper bound WOPT nodegree}]
 Since the tree $\treeCC''$ separates all the individual buckets $B \in \buckets$ by construction,  we can upper bound the cost of $\WOPT_{G / \buckets}$ with respect to the cost of $\treeCC''$ and
have that 
\[
\WOPT_{G / \C} \leq \COST_G\lp\treeCC\rp + O(k/\phiIn^{10}) \cdot \OPT_G + O(\beta \cdot k^3/\phiIn^{10}) \cdot \OPT_G = O\lp  \beta \cdot k^{23}/\lambda_{k+1}^{10}\rp \cdot \OPT_G,
\]
where the first inequality holds by Lemmas~\ref{lem:critical node split} and \ref{lem:critical node movement}, and   the second one  holds by Lemma~\ref{thm:main_k_clusters} and the fact $\phi_{\mathrm{in}}= \Theta(\lambda_{k+1}/k^2)$. 
\end{proof}

\begin{proof}[Proof of Lemma~\ref{lemma:upper bound inner nodegree}]
We bound the cost of the internal edges in $G[B_j]$ for $1\leq j \leq \kappa$ using the trivial upper bound for every $B_j$, and have that 
\begin{align}
     \sum_{i=1}^k \sum_{B \in \buckets_{i}} \COST_{G[B]} \left( \tree_{B} \right) & \leq \sum_{j=1}^\kappa |B_{j}| \cdot \Vol{B_{j}}  
    = \sum_{i'=1}^{k_1} \sum_{X_j \in \mathcal{X}_{i'}} |B_j| \cdot\sum_{N_{i''} \in X_j} \Vol{N_{i''} \cap B_{j}} \nonumber \\
    \leq &\sum_{i'=1}^{k_1} \sum_{X_j \in \mathcal{X}_{i'}} |\parent_{\treeCC'}(N_{i'} \cap B_{j})| \cdot\sum_{N_{i''} \in X_j} \Vol{N_{i''} \cap B_{j}} \nonumber
\end{align}
where   the last inequality holds by construction of the tree $\treeCC'$ and the property of anchor nodes. By repeating the same calculation as the one in proving  Lemma~\ref{lemma:upper bound WOPT nodegree}, we have that 
\begin{align}
     \sum_{i=1}^k \sum_{B \in \buckets_{i}} \COST_{G[B]} \left( \tree_{B} \right) 
    \leq &\sum_{i'=1}^{k_1} \sum_{X_j \in \mathcal{X}_{i'}} |\parent_{\treeCC'}(N_{i'} \cap B_{j})| \cdot\sum_{N_{i''} \in X_j} \Vol{N_{i''} \cap B_{j}} \nonumber\\
    = &O\lp  \beta \cdot k^{23}/\lambda_{k+1}^{10}\rp \cdot \OPT_G,\nonumber
\end{align} 
which proves the statement.
\end{proof}

\subsection{Proof of Theorem~\ref{thm:main1}}\label{sec:proof_main_theorem}
Now we   prove the approximation and running time guarantee of Algorithm~\ref{alg:spectral_clustering_degree_bucketing}. In the following Lemma we analyse the approximation and running time guarantee of \textsf{WRSC}. 

\begin{lemma}\label{lem:running_time_wrsc}
    Let $w_\mathrm{max} / w_\mathrm{min} = O(n^\gamma)$ for a constant $\gamma > 1$, and   $\buckets$   the set of all buckets constructed inside every $P_1, \ldots P_k$. Then, given a   contracted graph $H= G / \buckets$ as input the \textsf{WRSC} algorithm runs in   $\widetilde{O}(n)$ time and returns an \HC\ tree $\treeDT$, such that
    \[
    \WCOST(\treeDT) = O\lp 1 \rp \cdot \WOPT_{G / \buckets}.
    \]
\end{lemma}
\begin{proof}
We denote by $\buckets$ the set containing all the buckets obtained throughout all sets $P_i$, i.e., 
\[
    \buckets \triangleq \bigcup_{i=1}^k \buckets_i,
\]
and it holds that  $
    |\buckets| = \kappa \leq  k \cdot \log_\beta \lp \dmax_G / \dmin_G \rp$. 
Recall that our choice of $\gamma$ and $\beta$ satisfies
\[
\frac{w_{\max}}{w_{\min}} = O(n^{\gamma})
\]
and $\beta = 2^{k (\gamma+1)}$. Hence,  it holds that 
\[
\frac{\Delta_G}{\delta_G}  = O\left(n^{\gamma+1}\right),
\]
and the total number of buckets in $\mathcal{B}$ is upper bounded by 
\[
k\cdot \max\{1, \log_{\beta} n^{\gamma+1} \}\leq k + \frac{k\cdot (\gamma+1)}{\log\beta} \cdot \log n = k + \log n.
 \]
 Hence, there are at most $k + \log n$ vertices in $G/\mathcal{B}$.
     To compute the optimal weighted sparsest cut in $H = G/\mathcal{B}$, we  compute $\mathsf{sparsity}_H(S)$ for every subset $S \subset V(H)$, the total number of which is upper bounded by  $  2^{k + \log n} = O\lp 2^k \cdot n\rp$.      Since both $w(S, V(H)\setminus S)$ and $\vol_H(S)$ can be computed in $\widetilde{O}(1)$ time, the sparsest cut of $G/\mathcal{B}$ can be computed in $\widetilde{O}(n)$ time given that $k$ is a constant.    Combining this with the   master theorem proves the time complexity of the  \textsf{WRSC} algorithm.

The approximation guarantee for the algorithm follows from Lemma~\ref{lem:wcost_rsc_approx} and the fact that we compute the optimal weighted sparsest cut; as such we have $\alpha=1$
\end{proof}

Next, we  study the time complexity of Algorithm~\ref{alg:spectral_clustering_degree_bucketing}.  
We first show that the cost $\COST_G \lp \treeDT \rp$ can be computed in nearly-linear time, assuming that the internal nodes of our constructed tree $\treeDT$ have several useful attributes.

\begin{lemma}\label{lem:computing_cost_nearlylinear_dt}
   Let $\tree_G$ be an $\HC$ tree of a graph $G=(V,E,w)$, with depth $D$. Assume that every internal node $N$ of $\tree_G$ has 
   (i) a pointer $\parentsf$ that points to its parent node,
   (ii) an attribute $\sizesf$, which specifies the number of leaves in the subtree $\tree_G \lp G[N] \rp$, and 
   (iii) an attribute $\depthsf$ that  specifies the depth of $\tree_G \lp G[N] \rp$. 
   Then,  $\COST_G(\tree_G)$ can be computed in   $O(m \cdot D)$ time. 
\end{lemma}

\begin{proof}
 Since $G$ has $m$ edges, it suffices to show that we can compute the value  $\cost_G(e)$ for every edge $e = \{u, v\} \in E$ in $O(D)$ time. Let us fix an arbitrary edge $e=\{u,v\}$, and let  $N_u$ and $N_v$ be the leaf nodes of   $\tree_G$  corresponding to $u$ and $v$ respectively. Without loss of generality, we assume that $N_u(\depthsf) \geq N_v(\depthsf)$. 
We successively travel from  $N_u$ up the tree via the $\parentsf$ pointers until we reach a node $N'_u$, such that $N'_u(\depthsf) = N_v(\depthsf)$. 
After that, we find the lowest common ancestor of $u$ and $v$ by simultaneously travelling up one node at a time, starting from both $N'_u$ and $N_v$, 
until we reach the same internal node $N_r$.
Once this condition is satisfied,  
we  can readily compute $\cost_G(\{u,v\}) = w_e \cdot N_r(\sizesf)$. The  overall running time for computing $\cost_G(\{u,v\})$ is $O(D)$, and therefore the statement holds.
\end{proof}

By Lemma~\ref{lem:computing_cost_nearlylinear_dt} we know that the time complexity of constructing $\tree$ is linear in the depth of $\tree$. The following lemma shows that the depth of $\tree$ is $O(\log n)$, and the entire $\tree$ can be constructed in nearly-linear time.

\begin{lemma}\label{lem:cost_constructing_treeDT}Algorithm~\ref{alg:spectral_clustering_degree_bucketing} runs in   $\tilde{O} \lp m\rp$ time, and constructs the tree $\tree$ of depth  $O\lp\log n\rp$. 
     Moreover,  every internal node $N \in \treeDT$ has a pointer $\parentsf$, and 
     every node $T\in\treeDT$ has     
     attributes $\sizesf$ and $\depthsf$. 
\end{lemma}

\begin{proof}
We follow the execution of Algorithm~\ref{alg:spectral_clustering_degree_bucketing} to analyse its time complexity, and the depth of its constructed $\tree$. Line~\ref{alg:specWRSC:line:step1} of Algorithm~\ref{alg:spectral_clustering_degree_bucketing} applies spectral clustering to obtain the clusters $\{P_i\}_{i=1}^k$, and the running time of this step is $\widetilde{O}(m)$~\cite{peng_partitioning_2017}. Since ordering the vertices of every $P_i$ with respect to their degrees  takes $O\lp|P_i|\log|P_i|\rp$ time, the total time complexity of sorting the vertices for  $\{P_i\}$, finding $\{u^{(i)}\}$, as well as constructing $\{\mathcal{B}_i\}$~(Lines~\ref{alg:specWRSC:line:step2_start}--\ref{alg:specWRSC:line:step2_end}) is  $O \lp \sum_{i=1}^k |P_i|\log |P_i|\rp = O(n \log n)$.

Now let  $B_j \in \buckets$ be an arbitrary bucket,
and   we study 
 the construction of  $\tree_{B_j}$.
We assume that $B_j$ has vertices $u_1, \dots, u_{n_j}$ sorted in increasing order of their degrees, where $n_j = |B_j|$.  
We construct the balanced tree $\tree_{B_j}$ recursively as follows:

\begin{enumerate}
    \item Initialise a root node $N_0$, such that $N_0(\sizesf) = n_j$, and $N_0(\parentsf) = \emptyset$;
    \item Create nodes $N_1$ and $N_2$ corresponding to the vertices   $u_1, \dots, u_{\lceil n_j/2 \rceil}$ an  $u_{\lceil n_j/2 \rceil + 1}, \dots u_{n_j}$ respectively.
    Set $N_1(\parentsf) = N_0$ and $N_2(\parentsf) = N_0$;
    \item Recursively construct the subtrees rooted at $N_1$ and $N_2$
    on the corresponding sets of vertices, until we reach the leaves.
\end{enumerate}
By construction, $\tree_{B_j}$ is a balanced tree and has  $O(n_j)$ internal nodes. Hence, the depth of $\tree_{B_j}$ is  $O(\log n_j) = O(\log n)$. Moreover, it's straightforward to see that this step~(Lines~\ref{alg:specWRSC:line:step3_start}--\ref{alg:specWRSC:line:step3_end}) takes $O(n\log n)$ time for constructing  each  $\tree_{B_j}$.

For constructing the contracted graph $H$ from $\mathcal{B}$, notice that every edge of $G$ only contributes to exactly one of the edge weights in $H$; as such constructing $H$~(Line~\ref{alg:specWRSC:line:construct_contracted}) takes $O(m)$ time.
Finally, to complete the construction of $\treeDT$,  we run   the \textsf{WRSC} algorithm on $H$~(Line~\ref{alg:specWRSC:line:WRSC}), and   assume that we know how the vertices are split 
  at each iteration of the recursive weighted sparsest cut algorithm.  We construct the tree $\treeDT$ recursively as follows:
\begin{enumerate}
    \item Initialise a root node $R_0$, such that $R_0(\sizesf) = n$, and $R_0(\parentsf) = \emptyset$;
    \item Create two nodes $R_1$ and $R_2$ corresponding to the sets $A$ and $B$, where $(A, B)$ is the first split of the \textsf{WRSC} algorithm.
    Set $R_1(\parentsf) = R_0$ and $R_2(\parentsf) = R_0$;
    \item Recursively construct the subtrees rooted at $R_1$ and $R_2$
    on the corresponding sets of vertices and split $(A,B)$ returned by the \textsf{WRSC} algorithm, until we reach the leaves corresponding to $B_1, \ldots B_\kappa$. 
    \item For each $\tree_{B_j} \in \T$, we update the parent of the root node $N_0(\parentsf)$ with $N_0(\parentsf) = R_{B_j}(\parentsf)$, where  $R_{B_j}$ is the leaf node corresponding to $B_j$ in the tree construction returned above by the \textsf{WRSC} algorithm. We also delete the leaf node $R_{B_j}$, thereby effectively replacing $R_{B_j}$ by $\tree_{B_j}$.
\end{enumerate}
By Lemma~\ref{lem:running_time_wrsc} we know that the \textsf{WRSC} algorithm (Line~\ref{alg:specWRSC:line:WRSC}) runs in nearly-linear $\widetilde{O}(m)$ time. Taking everything together, this proves the time complexity of Algorithm~\ref{alg:spectral_clustering_degree_bucketing}. 
The $O(\log n)$ depth of $\tree$  follows by the fact that (i) the depth of every tree $\tree_{B_j}$ is $O(\log n)$, and  (ii) the  distance between  the root of $\treeDT$ and the root of any such tree $\T_{B_j}$  in $\treeDT$ is also $O(\log{n})$.

Finally, to set the $\depthsf$ attribute of every internal node, we perform a top-down traversal from the root of $\treeDT$ to the leaves, and update the $\depthsf$ for every internal node as one more than the depth of their parent.
\end{proof}

We are now ready to prove the main result, which we will state here in full.
\begin{theorem}[Formal statement of Theorem~\ref{thm:main1}]\label{thm:main1_full}
Let $G=(V, E, w)$ be a  graph with $w_\mathrm{max} / w_\mathrm{min} = O(n^\gamma)$ for a constant $\gamma>1$, and $k>1$  such that $\lambda_{k+1}>0$  and $\lambda_k < \lp \frac{1}{270 \cdot c_0 \cdot (k+1)^6} \rp^2$, where $c_0$ is the constant in Lemma~\ref{lem:Upperbound eig induced graphs}. Algorithm~\ref{alg:spectral_clustering_degree_bucketing}  runs in  time $\tilde{O}(m)$ and both constructs an \textsf{HC} tree $\tree$ of $G$ and computes $\COST_G(\tree)$ satisfying $\COST_G(\tree) = O\lp 2^{k(\gamma+1)} \cdot k^{23}/\lambda_{k+1}^{10}\rp \cdot \OPT_G$. In particular, when  $\lambda_{k+1}=\Omega(1)$ and $k=O(1)$, the algorithm's constructed tree $\tree$ satisfies that $\COST_G(\tree) = O(1) \cdot \OPT_G$.  
\end{theorem}

\begin{proof}[Proof of Theorem~\ref{thm:main1_full}]

Let $\treeDT$ be the tree returned by Algorithm~\ref{alg:spectral_clustering_degree_bucketing}.
By combining~\eqref{eq:sum_cost_degree_buckets_nodegree},~\eqref{eq:upperbound_crossing_edges}, Lemma~\ref{lemma:upper bound inner nodegree}, and the approximation guarantee in Lemma~\ref{lem:running_time_wrsc} we have that

\[
    \COST_G(\treeDT) = O\lp  2^{k(\gamma+1)} \cdot k^{23}/\lambda_{k+1}^{10}\rp \cdot \OPT_G.
\]
The time complexity of Algorithm~\ref{alg:spectral_clustering_degree_bucketing} follows 
by Lemma~\ref{lem:cost_constructing_treeDT}.
Moreover, as the depth of $\tree$ is $O(\log n)$, and every internal node $N \in \tree$ has a pointer $\parentsf$ and attributes $\sizesf$ and $\depthsf$, by Lemma~\ref{lem:computing_cost_nearlylinear_dt} we can compute $\COST_G(\tree)$ in nearly-linear time $\tilde{O}(m)$.

We complete the proof by dealing with our assumption that the algorithm knows the number of clusters $k$. If the number of clusters $k$ is unknown, we perform a standard technique from the literature ~\cite{CAKMT17, manghiuc_sun_hierarchical} and run independent copies of Algorithm~\ref{alg:spectral_clustering_degree_bucketing} with all possible values $k'$ ranging from $1$ to $k$.
By adding an extra $O(k) = O(1)$ factor in the overall time complexity, we ensure that one of the runs has the correct number of clusters $k' = k$.  
\end{proof}

\section{Omitted Details of Proof of Theorem~\ref{thm:main2}}\label{sec:appendix_second_result}
This section presents omitted details of the proof of Theorem~\ref{thm:main2}, and is structured as follows: we present a technical overview of our result in  Section~\ref{sec:technique_overview}; we describe our designed algorithm in Section~\ref{sec:algo_description}, and prove Theorem~\ref{thm:main2} in Section~\ref{sec:algo_analysis}.
 We introduce the parameter   $\degfracS$ such that 
    \[
        \degfracS \geq \max_{1\leq i\leq k} \frac{\Delta(S_i)}{\delta(S_i)};
    \]
this will be used in our later discussion.

\subsection{Overview of Main Techniques}\label{sec:technique_overview} 
We give an overview of our main techniques and   discuss why spectral clustering suffices to construct \HC\ trees with good approximation guarantees for well-clustered graphs. One general challenge for analysing the approximation ratio of any \HC\  tree $\tree$ of $G$ is due to the limited understanding of the ``structure'' of an optimal \HC\ tree. There are two main approaches to address this: (1) the first approach is to ensure that most ``cuts'' induced at the top of the tree $\tree$ are sparse; as such, a recursive application of the sparsest cut algorithm is employed; (2) the second approach is to reason about the structure of an optimal \HC\ tree assuming that the underlying graphs have a highly symmetric hierarchical structure of clusters~\cite{CAKMT17}, or a clear  cluster structure~\cite{manghiuc_sun_hierarchical}. In particular, the algorithm  in \cite{manghiuc_sun_hierarchical} requires that every returned vertex set has high inner conductance.  Therefore, taking these into account,  the starting point of our approach is to consider the role of high inner-conductance  in bounding the cost of an \HC\ tree.

Suppose that the optimal clusters $S_1,\ldots, S_k$ corresponding to $\rho(k)$ are given, and it holds $\Phi_{G[S_i]} \geq \phiin$ for $1 \leq i \leq k$. These $k$ clusters could have different sizes, volumes, and degree distributions. 
We show that one can construct a tree $\tree_S$, by first constructing trees $\tree_{G[S_i]}$ on every induced subgraph $G[S_i]$, and simply concatenating these $\tree_{G[S_i]}$ in a ``caterpillar'' style according to the sizes of $|S_i|~(1\leq i\leq k)$. Moreover, the cost of our constructed tree $\tree_S$ can be upper bounded with respect to $\OPT_G$. See Figure~\ref{fig:caterpillar_S} for an illustration of a caterpillar tree, and Lemma~\ref{lem:optimal_tree_on_k_expanders} for a formal statement; the proof of Lemma~\ref{lem:optimal_tree_on_k_expanders} can be found at the end of the subsection.

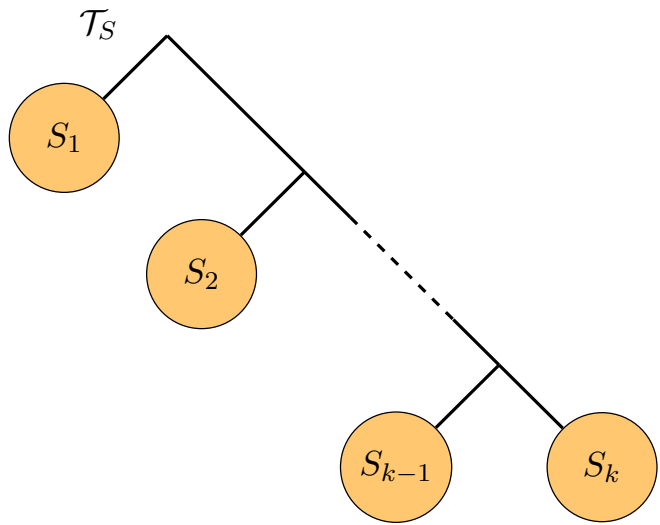
\begin{figure}[h]
    \centering
    \resizebox{9cm}{!}{%
\begin{tikzpicture}
	\begin{pgfonlayer}{nodelayer}
		\node [style={caterpillar_node_smaller}] (0) at (-3.25, -1.25) {$S_1$};
		\node [style={caterpillar_node_smaller}] (1) at (-0.25, -4.25) {$S_2$};
		\node [style={caterpillar_node_smaller}] (2) at (4, -8.5) {$S_{k-1}$};
		\node [style={caterpillar_node_smaller}] (3) at (8.5, -8.5) {$S_{k}$};
		\node [style=none] (4) at (-1, 1) {};
		\node [style=none] (5) at (6.25, -6.25) {};
		\node [style=none] (6) at (2, -2) {};
		\node [style=none] (8) at (-2.5, 1.25) {$\mathcal{T}_S$};
		\node [style=none] (9) at (3, -3) {};
		\node [style=none] (10) at (5.25, -5.25) {};
		\node [style=none] (11) at (5.25, -5.25) {};
	\end{pgfonlayer}
	\begin{pgfonlayer}{edgelayer}
		\draw [style=undirected] (5.center) to (2);
		\draw [style=undirected] (5.center) to (3);
		\draw [style=undirected, in=225, out=45] (0) to (4.center);
		\draw [style=undirected] (1) to (6.center);
		\draw [style=undirected] (4.center) to (6.center);
		\draw [style=undirected] (11.center) to (5.center);
		\draw [style=undirected] (6.center) to (9.center);
		\draw [style=undirected dashed] (9.center) to (11.center);
	\end{pgfonlayer}
\end{tikzpicture}
    }
    \caption{An illustration of our constructed ``caterpillar'' tree from optimal clusters $S_1, \ldots, S_k$ that
satisfy $|S_1| \geq \ldots \geq |S_k|$.}\label{fig:caterpillar_S}
\end{figure}

\begin{lemma}\label{lem:optimal_tree_on_k_expanders}
Let $G=(V,E,w)$ be an input graph of $n$ vertices, and $G$ has optimal clusters $S_1,\ldots, S_k$ corresponding to $\rho(k) \leq 1/k$ and it holds that $\Phi_{G[S_i]} \geq \phiin$ for $1 \leq i \leq k$. Given  $S_1,\ldots, S_k$ as part of the input, there is a nearly-linear time algorithm that constructs an \textsf{HC} tree $\tree_S$ such that \[
\COST_G(\tree_S) \leq \frac{36 \cdot \degfracS}{\phiin}   \cdot \OPT_G.
\]
\end{lemma}

Since $\phiin = \Omega(1)$ for many well-clustered graphs, Lemma~\ref{lem:optimal_tree_on_k_expanders} states that a simple caterpillar construction based on the $k$ optimal clusters achieves an approximation of $O(\degfracS)$, which is $O(1)$ if the degrees of \emph{every} cluster are almost-balanced. Moreover, thanks to this caterpillar structure over the $S_i$'s and an arbitrary tree construction for every $G[S_i]$,  we can construct $\tree_S$ and compute its cost in nearly-linear time.

 However, the above discussion assumes that the optimal clusters are known, which are co-$\mathsf{NP}$ hard to find~\cite{blum1981complexity}. As the starting point of our algorithm, we therefore employ the spectral clustering algorithm, expecting that the output  clusters $P_1,\ldots, P_k$ could directly replace $S_1,\ldots, S_k$. Unfortunately, it is unclear whether this approach would work for the following two reasons:
\begin{enumerate}
    \item while the output $P_1,\ldots, P_k$ of spectral clustering can be applied to approximate the optimal $S_1,\ldots, S_k$, it is unknown whether every $P_i$ approximates their corresponding optimal $S_i$ with respect to $\Phi_{G[S_i]}$; moreover, it remains open whether
    the inner conductance $\Phi_{G[S_i]}$ is approximated by $\Phi_{G[P_i]}$. Without the inner conductance guarantee of $\Phi_{G[P_i]}$ for all $1\leq i\leq k$, we cannot apply Lemma~\ref{lem:optimal_tree_on_k_expanders}.
    \item while Lemma~\ref{lem:MS22+} shows that the symmetric difference between every $P_i$ and their optimal corresponding $S_i$ can be bounded, this approximation guarantee can still increase the ratio between the maximum and minimum degrees in $P_i$. That is, with spectral clustering alone, we cannot upper bound $\eta_P$ (which is an upper bound for $\max_i (\Delta(P_i)/\delta(P_i))$) from $\eta_S$. Moreover, Lemma~\ref{lem:MS22+} does not provide any guarantees on the difference in size between $S_i$ and $P_i$, which is crucial if we want to construct arbitrary $\tree_{P_i}$'s. 
\end{enumerate}

We  overcome these two  bottlenecks through the following high-level ideas:
\begin{enumerate}
    \item we show that bounded inner conductance of every output $P_i$ isn't needed, and that the approximate $P_i$'s returned from spectral clustering suffice to construct an \HC\ tree with guaranteed approximation. This is a fundamental difference compared with \cite{manghiuc_sun_hierarchical}, which requires the  inner conductance being constant for every partition set.
    
    \item we introduce a novel bucketing procedure that decomposes $P_i$'s into further subsets, by grouping together vertices of similar degree. These subsets will form the basis of our tree construction. 
\end{enumerate}
Thanks to the bucketing step, we are able to construct an  \HC\ tree in nearly-linear line.

\begin{proof}[Proof of  Lemma~\ref{lem:optimal_tree_on_k_expanders}] 
Given an input graph $G$ and optimal clusters $S_1,\ldots, S_k$ such that $|S_1|\geq\ldots\geq |S_k|$, we construct the tree $\tree_S$ through the following two steps: first of all, we construct an arbitrary \textsf{HC} tree $\tree_{S_i}$ for every $G[S_i]$, $1\leq i \leq k$; secondly, we construct $\tree_S$ by concatenating the trees $\tree_{S_1},\ldots, \tree_{S_k}$ in a caterpillar style, i.e., $\tree_{S_1}$ is at the top of the tree, and $\tree_{S_2},\ldots, \tree_{S_k}$ are appended inductively. Since we can construct every $\tree_{G[S_i]}$ in an arbitrary manner and the tree $\tree_S$ can be constructed from $\{\tree_{G[S_i]}\}_{1\leq i\leq k}$ through sorting the clusters by their sizes, the overall algorithm runs in nearly-linear time.

Next, we analyse $\COST_G(\tree_S)$. By definition, it holds that
\begin{equation}\label{eq:cost_twoterms}
    \COST_G(\tree_S) = \sum_{i=1}^k \sum_{e \in E[G[S_i]]} \cost_G(e)  + \sum_{e \in \tilde{E}} \cost_{G}(e),
\end{equation}
where $\tilde{E}$ is the set of edges crossing different clusters. For the first summation of \eqref{eq:cost_twoterms}, we have that
\begin{equation} 
    \sum_{i=1}^k \COST_{G[S_i]}(\tree_{S_i}) 
    \leq \frac{18 \cdot \degfracS}{\phiin} \cdot \sum_{i=1}^k \OPT_{G[S_i]}  
    \leq \frac{18 \cdot \degfracS}{\phiin} \cdot \OPT_G.\label{eq:inner_edges}
\end{equation}
Here, the    first inequality follows the fact that the cost of any \textsf{HC} tree of $G=(V,E, w)$ is at most $|V|\cdot \vol(G)$ (Fact~\ref{fact:trivial_upper_bound}) and subsequently applying Lemma~\ref{lem:lower_bound_sum_SiVolSi}; the second inequality holds because $\OPT_{G[S_i]}$ is at most $\OPT_G$ for all $1\leq i\leq k$.

For the second summation of \eqref{eq:cost_twoterms}, it holds that
\[
\sum_{e\in\tilde{E}}\cost_{G}(e) = \sum_{i=1}^{k-1} \sum_{j=i+1 }^k\sum_{\substack{ u\in S_i, v\in S_j\\ e = \{u,v\} } } \cost_G(e).
\]
We look at all the edges adjacent to vertices in $S_1$. By construction, it holds that $|S_1|\geq n/k$, and $\abs{\leaves \lp \tree_S[u \vee v]\rp} = n$ for any edge $\{u,v\}$ satisfying $u\in S_1, v\in V\setminus S_1$. Since $w(S_1, V\setminus S_1)\leq \vol(S_1)\cdot \rho(k)$, it holds that 
\[
\sum_{\substack{u\in S_1, v\not\in S_1 \\ e=\{u,v\}}} \cost_G(e) = n\cdot w(S_1, V\setminus S_1) \leq n\cdot \vol(S_1)\cdot \rho(k) \leq k\cdot |S_1| \cdot \rho(k)\cdot \vol(S_1).
\]

We continue with looking at   the edges between $S_2$ and $V\setminus (S_1\cup S_2)$. As $S_2$ is the largest cluster among $S_2,\ldots, S_k$, it holds that $|S_2|\geq \frac{n-|S_1|}{k-1}$, which implies that $n- |S_1|< k\cdot |S_2|$. Notice that any edge $e=\{u,v\}$ with $u\in S_2$ and $v\in V\setminus (S_1\cup S_2)$ satisfies that $\abs{\leaves \lp \tree_S[u \vee v]\rp} = n - |S_1|$, and as such we have that 
\[
\sum_{ \substack{ u\in S_2, v\not\in S_1\cup S_2 \\ e=\{u,v\}}} \cost_G(e)  \leq (n - |S_1|)\cdot w(S_2, V\setminus (S_1\cup S_2)) \leq k\cdot |S_2|\cdot \rho(k)\cdot \vol(S_2).
\]
Generalising the analysis above, we have that 
\begin{align}
 \sum_{e\in\tilde{E}}\cost_{G}(e) & = \sum_{i=1}^{k-1} \sum_{j=i+1 }^k\sum_{\substack{ u\in S_i, v\in S_j\\ e = \{u,v\} } } \cost_G(e) \nonumber \\
 & \leq k \cdot \rho(k) \cdot \sum_{i=1}^k |S_i|\cdot \vol(S_i) \nonumber\\
 & \leq \frac{18 \cdot \degfracS}{\phiin}\cdot \sum_{i=1}^k\OPT_{G[S_i]}, \label{eq:step2}
\end{align}
where the last inequality follows from Lemma~\ref{lem:lower_bound_sum_SiVolSi}.
Combining \eqref{eq:cost_twoterms}, \eqref{eq:inner_edges} with \eqref{eq:step2}, we have that 
\[
\COST_G(\tree_S) \leq \frac{36 \cdot \degfracS}{\phiin}\cdot\OPT_G,
\]
which proves the statement.
\end{proof}

\subsection{Algorithm Description}\label{sec:algo_description}

    Now informally describe our algorithm and refer the reader to Algorithm~\ref{alg:spectral_degree_clustering} for the formal presentation.
    The algorithm takes as input a graph $G= (V, E, w)$ with $k$ optimal clusters $\{S_i\}_{i=1}^k$. 
    For simplicity, we assume that the algorithm knows the target number of clusters $k$ and an upper bound $\degfracS$ such that 
    \[
        \degfracS \geq \max_{1\leq i\leq k} \frac{\Delta(S_i)}{\delta(S_i)},
    \]
    and we carefully deal with this assumption at the end of the section.
    As the first step, the algorithm runs the $\SpecClust$ subroutine and obtains   the approximate clusters $\{P_i\}_{i=1}^k$.
    Following our previous discussion, it is unclear whether the algorithm can or cannot use the clusters $P_i$ directly in order to construct an $O(1)$-approximate tree.
    To overcome this difficulty, we further partition every $P_i$ into degree-based buckets, but rather than setting $\beta=2^{k(\gamma+1)}$ as we did for Algorithm~\ref{alg:spectral_clustering_degree_bucketing}, we set $\beta=\degfracS$. Recall then for any set $P_i$ and every vertex $u \in P_i$, the \emph{bucket} of $P_i$ starting at $u$ is defined
        \[
            \bkt{u} = \left\{v \in P_i : d_u \leq d_v < \degfracS \cdot d_u \right\}.
        \]

    Again, it is important to note that the bucket $B(u)$ contains \emph{every} vertex $v \in P_i$ whose degree satisfies $d_u\leq d_v < \degfracS \cdot d_u$.
    We refer the reader to Figure~\ref{fig:bucketing}(a) for an illustration of the bucketing procedure.

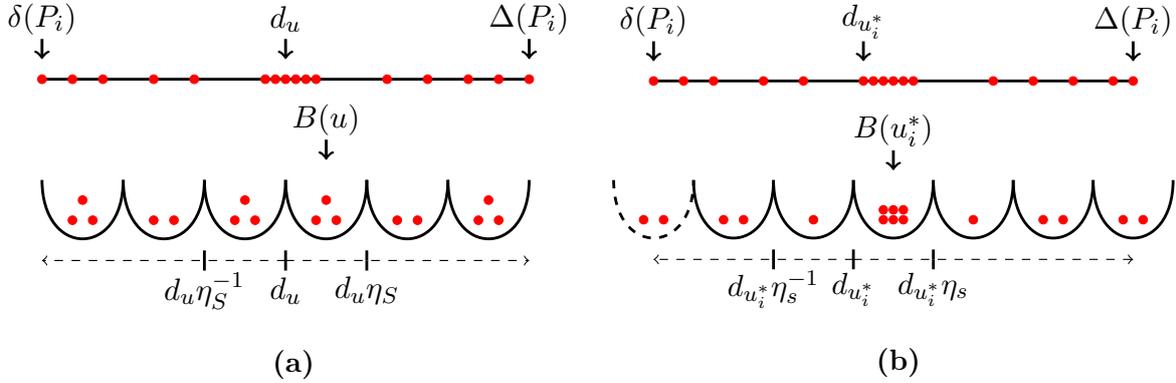
\begin{figure}[h]
    \begin{minipage}{0.5\textwidth}
    \centering
    \resizebox{7.8cm}{!}{%
    \begin{tikzpicture}
	\begin{pgfonlayer}{nodelayer}
		\node [style=none] (2) at (-7, 0.5) {};
		\node [style=none] (3) at (5, 0.5) {};
		\node [style=none] (20) at (-1, -2) {};
		\node [style=none] (21) at (-3, -2) {};
		\node [style=none] (22) at (-5, -2) {};
		\node [style=none] (23) at (1, -2) {};
		\node [style=none] (24) at (3, -2) {};
		\node [style=none] (25) at (5, -2) {};
		\node [style=none] (27) at (-7, -2) {};
		\node [style=none] (46) at (-1, -4) {};
		\node [style=none] (47) at (1, -4) {};
		\node [style=none] (48) at (3, -4) {};
		\node [style=none] (50) at (-1, -4) {};
		\node [style=none] (53) at (1, -3.75) {};
		\node [style=none] (54) at (1, -4.25) {};
		\node [style=none] (55) at (-1, -3.75) {};
		\node [style=none] (56) at (-1, -4.25) {};
		\node [style=none] (57) at (-1, -4.75) {$d_u$};
		\node [style=none] (64) at (-1, -4.25) {};
		\node [style=none] (65) at (1, -4.75) {$d_u \degfracS$};
		\node [style=none] (67) at (-7, 1.25) {};
		\node [style=smallred] (68) at (-1, 0.5) {};
		\node [style=smallred] (69) at (-1.25, 0.5) {};
		\node [style=smallred] (70) at (-1.5, 0.5) {};
		\node [style=smallred] (71) at (-0.75, 0.5) {};
		\node [style=smallred] (74) at (2.5, 0.5) {};
		\node [style=smallred] (75) at (3.5, 0.5) {};
		\node [style=smallred] (76) at (4.25, 0.5) {};
		\node [style=smallred] (78) at (-0.5, 0.5) {};
		\node [style=smallred] (79) at (-3.25, 0.5) {};
		\node [style=smallred] (80) at (-4.25, 0.5) {};
		\node [style=smallred] (81) at (-5.5, 0.5) {};
		\node [style=smallred] (82) at (-6.25, 0.5) {};
		\node [style=smallred] (83) at (-6.25, -3) {};
		\node [style=smallred] (84) at (-5.75, -3) {};
		\node [style=smallred] (85) at (-4.25, -3) {};
		\node [style=smallred] (86) at (-3.75, -3) {};
		\node [style=smallred] (87) at (-2.25, -3) {};
		\node [style=smallred] (88) at (-1.75, -3) {};
		\node [style=smallred] (89) at (-2, -2.5) {};
		\node [style=smallred] (90) at (-0.25, -3) {};
		\node [style=smallred] (91) at (0.25, -3) {};
		\node [style=smallred] (92) at (0, -2.5) {};
		\node [style=smallred] (93) at (2.25, -3) {};
		\node [style=smallred] (94) at (3.75, -3) {};
		\node [style=smallred] (95) at (4.25, -3) {};
		\node [style=smallred] (97) at (1.5, 0.5) {};
		\node [style=smallred] (98) at (1.75, -3) {};
		\node [style=none] (99) at (-1, 1.5) {};
		\node [style=none] (100) at (-1, 1.5) {};
		\node [style=none] (101) at (5, -4) {};
		\node [style=none] (102) at (-7, -4) {};
		\node [style=none] (103) at (-3, -3.75) {};
		\node [style=none] (104) at (-3, -4.25) {};
		\node [style=none] (105) at (-3, -4.75) {$d_u \degfracS^{-1}$};
		\node [style=none] (106) at (-1, 1) {};
		\node [style=none] (107) at (5, 1.5) {};
		\node [style=none] (108) at (5, 1.5) {};
		\node [style=none] (109) at (5, 1) {};
		\node [style=none] (110) at (-7, 1.5) {};
		\node [style=none] (111) at (-7, 1.5) {};
		\node [style=none] (112) at (-7, 1) {};
		\node [style=none] (113) at (-7, 2) {$\dmin(P_i)$};
		\node [style=none] (114) at (-1, 2) {$d_u$};
		\node [style=none] (115) at (5, 2) {$\dmax(P_i)$};
		\node [style=none] (116) at (0, -0.5) {$B(u)$};
		\node [style=none] (117) at (0, -1) {};
		\node [style=none] (118) at (0, -1.5) {};
		\node [style=none] (119) at (0, -1.5) {};
		\node [style=none] (120) at (0, -1) {};
		\node [style=smallred] (121) at (5, 0.5) {};
		\node [style=smallred] (122) at (-7, 0.5) {};
		\node [style=smallred] (123) at (-6, -2.5) {};
		\node [style=smallred] (124) at (4, -2.5) {};
		\node [style=smallred] (125) at (-0.25, 0.5) {};
	\end{pgfonlayer}
	\begin{pgfonlayer}{edgelayer}
		\draw [style=undirected] (2.center) to (3.center);
		\draw [style=undirected, bend right=90, looseness=2.50] (27.center) to (22.center);
		\draw [style=undirected, bend right=90, looseness=2.50] (22.center) to (21.center);
		\draw [style=undirected, bend right=90, looseness=2.50] (21.center) to (20.center);
		\draw [style=undirected, bend right=90, looseness=2.50] (23.center) to (24.center);
		\draw [style=undirected, bend right=90, looseness=2.50] (24.center) to (25.center);
		\draw [style=undirected] (47.center) to (53.center);
		\draw [style=undirected, in=90, out=-90] (47.center) to (54.center);
		\draw [style=undirected] (50.center) to (55.center);
		\draw [style=undirected] (50.center) to (56.center);
		\draw [style=undirected, in=270, out=-90, looseness=2.50] (20.center) to (23.center);
		\draw [style=directed dash] (50.center) to (102.center);
		\draw [style=directed dash] (50.center) to (101.center);
		\draw [style=undirected] (103.center) to (104.center);
		\draw [style=directed] (100.center) to (106.center);
		\draw [style=directed] (108.center) to (109.center);
		\draw [style=directed] (111.center) to (112.center);
		\draw [style=directed] (120.center) to (119.center);
	\end{pgfonlayer}
\end{tikzpicture}
    }\\
    \vspace{0.5cm}
    \textbf{(a)}
    \end{minipage}%
    \begin{minipage}{0.5\textwidth}
    \centering
    \resizebox{7.8cm}{!}{%
    \begin{tikzpicture}
	\begin{pgfonlayer}{nodelayer}
		\node [style=none] (2) at (-7, 0.5) {};
		\node [style=none] (3) at (5, 0.5) {};
		\node [style=none] (20) at (0, -2) {};
		\node [style=none] (21) at (-2, -2) {};
		\node [style=none] (22) at (-4, -2) {};
		\node [style=none] (23) at (2, -2) {};
		\node [style=none] (24) at (4, -2) {};
		\node [style=none] (25) at (6, -2) {};
		\node [style=none] (27) at (-6, -2) {};
		\node [style=none] (46) at (-2, -4) {};
		\node [style=none] (47) at (0, -4) {};
		\node [style=none] (48) at (2, -4) {};
		\node [style=none] (50) at (-2, -4) {};
		\node [style=none] (53) at (0, -3.75) {};
		\node [style=none] (54) at (0, -4.25) {};
		\node [style=none] (55) at (-2, -3.75) {};
		\node [style=none] (56) at (-2, -4.25) {};
		\node [style=none] (57) at (-2, -4.75) {$d_{u_i^*}$};
		\node [style=none] (64) at (-2, -4.25) {};
		\node [style=none] (65) at (0, -4.75) {$d_{u_i^*} \eta_s$};
		\node [style=none] (67) at (-7, 1.25) {};
		\node [style=smallred] (68) at (-1, 0.5) {};
		\node [style=smallred] (69) at (-1.25, 0.5) {};
		\node [style=smallred] (70) at (-1.5, 0.5) {};
		\node [style=smallred] (71) at (-0.75, 0.5) {};
		\node [style=smallred] (74) at (2.5, 0.5) {};
		\node [style=smallred] (75) at (3.5, 0.5) {};
		\node [style=smallred] (76) at (4.5, 0.5) {};
		\node [style=smallred] (78) at (-0.5, 0.5) {};
		\node [style=smallred] (79) at (-3.25, 0.5) {};
		\node [style=smallred] (80) at (-4.25, 0.5) {};
		\node [style=smallred] (81) at (-5.5, 0.5) {};
		\node [style=smallred] (82) at (-6.25, 0.5) {};
		\node [style=smallred] (97) at (1.5, 0.5) {};
		\node [style=none] (99) at (-1.75, 1.5) {};
		\node [style=none] (100) at (-1.75, 1.5) {};
		\node [style=none] (101) at (5, -4) {};
		\node [style=none] (102) at (-7, -4) {};
		\node [style=none] (103) at (-4, -3.75) {};
		\node [style=none] (104) at (-4, -4.25) {};
		\node [style=none] (105) at (-4, -4.75) {$d_{u_i^*} \eta^{-1}_s$};
		\node [style=none] (106) at (-1.75, 1) {};
		\node [style=none] (107) at (5, 1.5) {};
		\node [style=none] (108) at (5, 1.5) {};
		\node [style=none] (109) at (5, 1) {};
		\node [style=none] (110) at (-7, 1.5) {};
		\node [style=none] (111) at (-7, 1.5) {};
		\node [style=none] (112) at (-7, 1) {};
		\node [style=none] (113) at (-7, 2) {$\dmin(P_i)$};
		\node [style=none] (114) at (-1.75, 2) {$d_{u_i^*}$};
		\node [style=none] (115) at (5, 2) {$\dmax(P_i)$};
		\node [style=none] (116) at (-8, -2) {};
		\node [style=smallred] (117) at (-6.75, -3) {};
		\node [style=smallred] (118) at (-5.25, -3) {};
		\node [style=smallred] (119) at (-4.75, -3) {};
		\node [style=smallred] (120) at (-3, -3) {};
		\node [style=smallred] (121) at (-1.25, -3) {};
		\node [style=smallred] (122) at (-1, -3) {};
		\node [style=smallred] (123) at (-0.75, -3) {};
		\node [style=smallred] (124) at (-1, -2.75) {};
		\node [style=smallred] (125) at (-1.25, -2.75) {};
		\node [style=smallred] (126) at (-0.75, -2.75) {};
		\node [style=smallred] (127) at (1, -3) {};
		\node [style=smallred] (128) at (2.75, -3) {};
		\node [style=smallred] (129) at (3.25, -3) {};
		\node [style=smallred] (130) at (4.75, -3) {};
		\node [style=none] (131) at (-1, -0.75) {$B(u_i^*)$};
		\node [style=none] (132) at (-1, -1.25) {};
		\node [style=none] (133) at (-1, -1.75) {};
		\node [style=smallred] (134) at (-7, 0.5) {};
		\node [style=smallred] (135) at (5, 0.5) {};
		\node [style=smallred] (136) at (5.25, -3) {};
		\node [style=smallred] (137) at (-7.25, -3) {};
		\node [style=smallred] (138) at (-1.75, 0.5) {};
	\end{pgfonlayer}
	\begin{pgfonlayer}{edgelayer}
		\draw [style=undirected] (2.center) to (3.center);
		\draw [style=undirected, bend right=90, looseness=2.50] (27.center) to (22.center);
		\draw [style=undirected, bend right=90, looseness=2.50] (22.center) to (21.center);
		\draw [style=undirected, bend right=90, looseness=2.50] (21.center) to (20.center);
		\draw [style=undirected, bend right=90, looseness=2.50] (23.center) to (24.center);
		\draw [style=undirected, bend right=90, looseness=2.50] (24.center) to (25.center);
		\draw [style=undirected] (47.center) to (53.center);
		\draw [style=undirected] (47.center) to (54.center);
		\draw [style=undirected, in=270, out=90] (50.center) to (55.center);
		\draw [style=undirected] (50.center) to (56.center);
		\draw [style=undirected, in=270, out=-90, looseness=2.50] (20.center) to (23.center);
		\draw [style=directed dash] (50.center) to (102.center);
		\draw [style=directed dash] (50.center) to (101.center);
		\draw [style=undirected] (103.center) to (104.center);
		\draw [style=directed] (100.center) to (106.center);
		\draw [style=directed] (108.center) to (109.center);
		\draw [style=directed] (111.center) to (112.center);
		\draw [style=undirected dashed, bend right=90, looseness=2.50] (116.center) to (27.center);
		\draw [style=directed] (132.center) to (133.center);
	\end{pgfonlayer}
\end{tikzpicture}
    }\\
    \vspace{0.5cm}
    \textbf{(b)}
    \end{minipage}%
\caption{This figure illustrates two different bucketings of the vertices inside some $P_i$. Figure \textbf{(a)} is a bucketing induced by some vertex $u$. Figure \textbf{(b)} is the bucketing induced by $u_i^*$ with $\Vol{B(u^*)}$ maximised. } \label{fig:bucketing}
\end{figure}

    Similar to before, vertices $u, v \in P_i$ having different degrees $d_u \neq d_v$ generally induce different bucketings.  
    In order to determine the appropriate bucketing, Algorithm~\ref{alg:spectral_degree_clustering} chooses as representative a vertex $u_i^* \in P_i$, whose induced bucket $\bkt{u_i^*}$ has the highest volume (Line~\ref{alg:line:choosing ui*} of Algorithm~\ref{alg:spectral_degree_clustering}).
    To motivate this choice, notice that every cluster $P_i$ has a large overlap with its corresponding optimal cluster $S_i$ and the degrees of the vertices inside $S_i$ are within a factor of $\degfracS$ of each other.
    If an arbitrary vertex $u \in P_i$ is chosen as representative,
    the bucketing $\bkt{u}$ of $P_i$ might equally divide the vertices in $P_i \cap S_i$ into two consecutive buckets, which is undesirable.
    However, we will prove that our specific choice of $u_i^*$ guarantees that the bucketing $\bkting{u_i^*}$ contains one bucket $B_j \in \bkting{u_i^*}$ largely overlapping with $S_i$ (see Figure~\ref{fig:bucketing}(b) for an illustration).

    Once the bucketing procedure of $P_i$ is complete the algorithm proceeds and constructs an arbitrary \emph{balanced} binary tree $\tree_{B_j}$, for every bucket $B_j$, and adds the tree $\tree_{B_j}$ to a global set of trees $\mathbb{T}$.
    This process is repeated for every cluster $P_i$.
    As a final step, the algorithm merges the trees inside $\mathbb{T}$ in a caterpillar fashion in decreasing order of the sizes, placing the larger sized trees closer to the root (Lines~\ref{alg:line:begin caterpillar}-\ref{alg:line:end caterpillar} of Algorithm~\ref{alg:spectral_degree_clustering}).
    The resulting tree $\TSCD$ is   the output of the algorithm.

\subsection{Algorithm Analysis}\label{sec:algo_analysis}

This subsection analyses the performance of   Algorithm~\ref{alg:spectral_degree_clustering}, and proves Theorem~\ref{thm:main2}. It is organised as follows: in the first part we prove several properties of our carefully constructed bucketing procedure.
In the second part, we prove that the cost of our constructed tree $\TSCD$ achieves the claimed approximation guarantee.
Thirdly, we give a detailed description of how we can efficiently construct the tree $\TSCD$ and compute its cost in nearly-linear time.
The proof of Theorem~\ref{thm:main2} follows by combining the three main results.

We first prove that   our bucketing procedure induces at most $2k$ buckets inside every cluster $P_i$, and the total number of buckets throughout all clusters $P_i$ is at most $2k^2$.

\begin{lemma}\label{lemma:Bucketing props}
Let $P_i$ be an arbitrary cluster, and $u^*_i = \argmax_{u \in P_i} \Vol{\bkt{u}}$ be a vertex inducing the bucketing $\bkting{u_i^*}$. The following statements hold:
\begin{enumerate}
    \item $\bkt{u_i^*} \in \bkting{u_i^*}$;
    \item For every optimal cluster $S_j$, the vertices in $S_j \cap P_i$ belong to at most two buckets, i.e., 
    \[
        \left\lvert \left\{ B_t \in \bkting{u_i^*} : S_j \cap B_t \neq \emptyset \right\} \right\rvert \leq 2;
    \]
    \item The bucketing induced by $u_i^*$ contains at most $2k$ buckets
    $\left\lvert \bkting{u_i^*} \right\rvert \leq 2k$.
\end{enumerate}
\end{lemma}
\begin{proof}

The first property follows trivially from the definition. The second property follows from our assumption that $\Delta(S_i) \leq \degfracS \cdot \delta(S_i)$ for every optimal cluster $S_i$. 
Therefore, the vertices inside $S_i$ belong  to at most $2$ buckets. 
The third property follows from the second one and summing over all clusters.
\end{proof}

Our second result presents a very important property of our bucketing procedure.
We show that, in every cluster $P_i$, exactly one bucket $B_j \in \bkting{u_i^*}$ has  large overlap with the optimal cluster $S_i$, and all the other buckets have low volume.
We emphasise that this result is only possible due to our calibration technique of choosing the bucketing induced by a vertex $u_i^* \in P_i$, whose bucket $\bkt{u_i^*}$ has the highest volume. 

\begin{lemma}\label{lem:Bucket calibration}
    Let $P_i$ be an arbitrary cluster, and $u^{*}_i = \argmax_{u \in P_i} \Vol{\bkt{u}}$ be a vertex inducing the bucketing $\bkting{u_i^*}$.
    The following two statements hold:
    \begin{enumerate}
        \item There exists a unique \emph{heavy} bucket $B_j \in \bkting{u^*_i}$ such that $\Vol{B_j \cap S_i} \geq \lp 1 - \frac{2k \cdot \CGap}{3\Upsilon(k)}\rp \cdot \Vol{S_i}$;
        \item For every other bucket  $B_t \in \bkting{u^*_i} \setminus \{B_j\}$ it holds that $\Vol{B_t} \leq \frac{k \cdot \CGap}{\Upsilon(k)} \cdot \Vol{S_i}$.
    \end{enumerate}
\end{lemma}
\begin{proof}
For the first statement, we show that the bucket $B(u^*_i)$ satisfies the claim.
Let $u_{\min}$ be a vertex in $P_i \cap S_i$ of the lowest degree.
Since the degrees of all vertices in $S_i$ (and hence in $P_i \cap S_i$) are within a factor of $\degfracS$ from each other, we know that all vertices in $P_i \cap S_i$ must be in the bucket $\bkt{u_{\min}}$, i.e., $P_i \cap S_i \subseteq \bkt{u_{\min}}$.
Moreover, by the choice of $u^*_i$ we have that
\begin{equation}\label{eq:bucket calib eq1}
    \Vol{\bkt{u^*_i}}
    \geq \Vol{\bkt{u_{\min}}}
    \geq \Vol{P_i \cap S_i}
    \geq \Vol{S_i} - \Vol{P_i \triangle S_i}
    \geq \lp 1 - \frac{k \cdot \CGap}{3\Upsilon(k)}\rp \cdot \vol(S_i),
\end{equation}
where the last inequality follows by Lemma~\ref{lem:MS22+}.
We now write
\begin{align*}
    \Vol{\bkt{u_i^*} \cap S_i} 
    &=\Vol{\bkt{u^*_i}} - \Vol{\bkt{u^*_i} \setminus S_i}\\
    &\geq \lp 1 - \frac{k \cdot \CGap}{3\Upsilon(k)}\rp \cdot \Vol{S_i} - \Vol{P_i \setminus S_i}\\
    &\geq \lp 1 - \frac{2k \cdot \CGap}{3\Upsilon(k)}\rp \cdot \Vol{S_i},
\end{align*}
where the first inequality uses \eqref{eq:bucket calib eq1}.
The uniqueness of the heavy bucket  follows by noticing that $\Vol{B(u_i^*) \cap S_i} > \Vol{S_i}/2$.

The second statement essentially follows from the fact that $S_i$ has a large volume overlap with $\bkt{u_i^*}$.
Concretely, let $B_t \in \bkting{u_i^*} \setminus \left\{ \bkt{u_i^*}\right\}$ be an arbitrary bucket.
We have that
\begin{align*}
    \Vol{B_t} 
    &= \Vol{B_t \cap S_i} + \Vol{B_t \setminus S_i}\\
    &\leq \Vol{S_i} - \Vol{S_i \setminus B_t} + \Vol{P_i \setminus S_i} \\
    &\leq \Vol{S_i} - \Vol{S_i \cap \bkt{u_i^*}} + \Vol{P_i \triangle S_i} \\
    &\leq \frac{2k \cdot \CGap}{3\Upsilon(k)} \cdot \Vol{S_i}
    + \frac{k \cdot \CGap}{3\Upsilon(k)} \cdot \Vol{S_i} \\
    &= \frac{k \cdot \CGap}{\Upsilon(k)} \cdot \Vol{S_i}.
\end{align*}
This completes the proof.
\end{proof}

In order to analyse the final constructed tree $\TSCD$, we analyse the properties of the buckets in the context of the entire graph $G$, and not just on every set $P_i$.
We denote by $\buckets$ the set containing all the buckets obtained throughout all sets $P_i$, i.e., 
\[
    \buckets \triangleq \bigcup_{i=1}^k \bkting{u^*_i},
\]
where $u^*_i = \argmax_{u \in P_i} \Vol{\bkt{u}}$ is a vertex inducing the corresponding bucketing in set $P_i$.
We remark that $\buckets$ is a partition of $V$.
We use 
\[
    \ell \triangleq |\buckets|
\]
to denote the total number of buckets, and we know by Lemma~\ref{lemma:Bucketing props} that $\ell \leq 2k^2$.
For convenience, we label the buckets $\buckets = \{B_1, \dots, B_{\ell}\}$ in decreasing order of the sizes breaking ties arbitrarily, i.e., $|B_i| \geq |B_{i+1}|$, for all $i < \ell$. 

Based on Lemma~\ref{lem:Bucket calibration}, exactly $k$ buckets in $\buckets$ are heavy, meaning that they have  large overlap with the optimal clusters $S_i$.
We denote by $\heavyBkts$ the set of those $k$ buckets;  we emphasise once more that each heavy bucket corresponds to a unique cluster $S_i$ and vice-versa.
Additionally, we denote by $\lightBkts \triangleq \buckets \setminus \heavyBkts$ the set of the remaining \emph{light} buckets.

The next two results summarise the main properties of our constructed buckets.
We first present three properties of the heavy buckets which bound their size, volume and weight of the crossing edges.

\begin{lemma}\label{lem:Heavy buckets props}
    For every heavy bucket $B_j \in \heavyBkts$ corresponding to $S_i$, the following statements hold:
    \begin{enumerate}[label=(B\arabic*)]
        \item $\Vol{B_j} \leq \lp 1 + \frac{k \cdot \CGap}{3\Upsilon(k)} \rp \cdot \Vol{S_i}$;
        \item $|B_j| \leq \lp 1 + \degfracS \rp \cdot |S_i|$.
        \item $w(B_j, \overline{B_j}) \leq \frac{k \cdot \CGap + \lambda_{k+1}}{\Upsilon(k)} \cdot \Vol{S_i}$.
    \end{enumerate}
\end{lemma}

\begin{proof}
Let $B_j \in \heavyBkts$ be a heavy bucket corresponding to $S_i$.
For (B1), we have that
\[
    \Vol{B_j} 
    = \Vol{B_j \cap S_i} + \Vol{B_j \setminus S_i}
    \leq \Vol{S_i} + \Vol{P_i \triangle S_i}
    \leq \lp 1 + \frac{k \cdot \CGap}{3\Upsilon(k)}\rp \cdot \Vol{S_i}.
\]
    
    To prove (B2), suppose for contradiction that $|B_j| > \lp 1 + \degfracS\rp \cdot |S_i|$, which implies that $|B_j \setminus S_i| = |B_j| - |B_j \cap S_i| > \degfracS \cdot |B_j \cap S_i|$. 
Based on this, we have that
\begin{align*}
    \Vol{B_j \setminus S_i}
    &\geq |B_j \setminus S_i| \cdot \delta(B_j \setminus S_i)
    > \degfracS \cdot | B_j \cap S_i | \cdot \frac{\Delta(B_j \cap S_i)}{\degfracS}\\
    &\geq \Vol{B_j \cap S_i}
    \geq \lp 1 - \frac{2k \cdot \CGap}{3\Upsilon(k)}\rp \cdot \Vol{S_i},
\end{align*}
where the second inequality follows from the fact that, inside bucket $B_j$ all the degrees are within a factor of $\degfracS$ of each other and the last inequality follows by Lemma~\ref{lem:Bucket calibration}.
This implies that
\[
    \Vol{S_i} < \frac{1}{1 - \frac{2k \cdot \CGap}{3\Upsilon(k)}} \cdot \Vol{P_i \triangle S_i}
    \leq \frac{\Vol{S_i}}{\frac{3\Upsilon(k)}{k \cdot \CGap} - 2} 
    \leq \Vol{S_i},
\]
where the second inequality follows by Lemma~\ref{lem:MS22+}, and last inequality by our assumption on $\Upsilon(k)$.

    Finally, we prove (B3). We have that
\begin{align}
    w(B_j, \overline{B_j}) &\leq  w(B_j \cap S_i, \overline{B_j}) + w(B_j \setminus S_i, \overline{B_j})  \nonumber\\
    &\leq   w(B_j \cap S_i, \overline{B_j} \cap S_i) + w(B_j \cap S_i, \overline{B_j} \setminus S_i) + \Vol{B_j \setminus S_i}  \nonumber\\
    &\leq   \Vol{\overline{B_j} \cap S_i} + w(S_i, \overline{S_i}) + \Vol{P_i \triangle S_i}  \nonumber\\
    &\leq  \frac{2k \cdot \CGap}{3\Upsilon(k)} \cdot \Vol{S_i} + \rho(k) \cdot \Vol{S_i} + \frac{k \cdot \CGap}{3\Upsilon(k)} \cdot \Vol{S_i}  \label{eq:lem:Bucket main technical eq2}\\
    &\leq \frac{k \cdot \CGap + \lambda_{k+1}}{\Upsilon(k)} \cdot \Vol{S_i},\label{eq:lem:Bucket main technical eq3}
\end{align}
where~\eqref{eq:lem:Bucket main technical eq2} follows from Lemmas~\ref{lem:Bucket calibration} and~\ref{lem:MS22+}, and \eqref{eq:lem:Bucket main technical eq3} uses the fact that $\rho(k) = \frac{\lambda_{k+1}}{\Upsilon(k)}$.
\end{proof}

The final technical result is similar with Lemma~\ref{lem:Heavy buckets props}, and presents several useful upper bounds of the volume and size of the light buckets.

\begin{lemma}\label{lem:Light buckets props}
   For every \emph{light} bucket $B_t \in \lightBkts$, the following statements hold:
    \begin{enumerate}[label=(B\arabic*)]
        \setcounter{enumi}{3}
        \item $\Vol{B_t \cap S_j} \leq \degfracS \cdot \Vol{B_t \cap S_i}$, for all clusters $S_i, S_j$ satisfying $|B_t \cap S_j| \leq |B_t \cap S_i|$;
        \item $\Vol{B_t \cap S_i} \leq \frac{2 k \cdot \CGap}{3\Upsilon(k)} \cdot \Vol{S_i}$, for all clusters $S_i$;
        \item $\sum_{B_t \in \lightBkts}|B_t| \cdot \Vol{B_t} \leq \frac{4 k^2 \cdot \CGap \cdot \degfracS}{3\Upsilon(k)} \cdot \sum_{i=1}^k |S_i| \cdot \Vol{S_i}$.
    \end{enumerate}
\end{lemma}
\begin{proof}    
Let $B_t \in \lightBkts$ be a light bucket, and let $S_i, S_j$ be optimal clusters so that $|B_t \cap S_j| \leq |B_t \cap S_i|$. It holds that
\[
    \Vol{B_t \cap S_j} 
    \leq |B_t \cap S_j| \cdot \Delta(B_t \cap S_j)
    \leq |B_t \cap S_j| \cdot \degfracS \cdot \delta(B_t \cap S_i)
    \leq \degfracS \cdot \Vol{B_j \cap S_i},
\]
where the second inequality uses the fact that inside bucket $B_t$ all the degrees are within a factor of $\degfracS$. This proves (B4).

Next, we prove (B5). Suppose that $B_t$ is a bucket from some cluster $P_r$, and let $S_i$ be an arbitrary cluster. The proof continues by a case distinction:

\paragraph*{Case~$1$: $r = i$.} In this case, since $B_t$ is a light bucket in $P_r$, by Lemma~\ref{lem:Bucket calibration} we have that
\[
    \Vol{B_t \cap S_i} 
    = \Vol{B_t \cap S_r} 
    \leq \frac{2 k \cdot \CGap}{3\Upsilon(k)} \cdot \Vol{S_r}.
\]

\paragraph*{Case~$2$: $r \neq i$.} In this case we simply have that
\[
    \Vol{B_t \cap S_{i}} 
    \leq \Vol{P_r \cap S_{i}}
    \leq \Vol{P_{i} \triangle S_{i}}
    \leq \frac{k \cdot \CGap}{3\Upsilon(k)} \cdot \Vol{S_i}.
\]
Combining the two cases proves   (B5).

We finally prove (B6). We have that 
\begin{align}
    &\sum_{B_t \in \lightBkts} |B_t| \Vol{B_t} \nonumber\\
    &= \sum_{B_t \in \lightBkts} \lp \sum_{i=1}^{k} |B_t \cap S_i|\rp \lp \sum_{j=1}^k\Vol{B_t \cap S_j}\rp \nonumber\\
    &= \sum_{B_t \in \lightBkts} \sum_{i, j = 1}^k |B_t \cap S_i| \Vol{B_t \cap S_j} \nonumber\\
    &= \sum_{B_t \in \lightBkts} \lp \sum_{\substack{i, j = 1: \\ |B_t \cap S_i| \leq |B_t \cap S_j|}}^k |B_t \cap S_i| \Vol{B_t \cap S_j}
        + \sum_{\substack{i, j = 1: \\ |B_t \cap S_i| > |B_t \cap S_j|}}^k |B_t \cap S_i| \Vol{B_t \cap S_j}\rp \label{eq:Bucket main technical eq4}\\
    &\leq \sum_{B_t \in \lightBkts} \lp \sum_{\substack{i, j = 1: \\ |B_t \cap S_i| \leq |B_t \cap S_j|}}^k |B_t \cap S_j| \Vol{B_t \cap S_j}
        + \sum_{\substack{i, j = 1: \\ |B_t \cap S_i| > |B_t \cap S_j|}}^k |B_t \cap S_i| \cdot \degfracS \cdot \Vol{B_t \cap S_i}\rp \label{eq:Bucket main technical eq5} \\
    &\leq \sum_{B_t \in \lightBkts} \lp k \sum_{j=1}^k |B_t \cap S_j| \Vol{B_t \cap S_j} + k \sum_{i=1}^k |B_t \cap S_i| \cdot \degfracS \cdot \Vol{B_t \cap S_i}\rp  \nonumber\\ 
    &\leq 2k \cdot \degfracS \sum_{i=1}^k \sum_{B_t \in \lightBkts}
    |B_t \cap S_i| \Vol{B_t \cap S_i} \nonumber\\
    &\leq 2k \cdot \degfracS \cdot \frac{2 k \cdot \CGap}{3\Upsilon(k)} \cdot \sum_{i=1}^k \sum_{B_t \in \lightBkts}
    |B_t \cap S_i| \Vol{S_i} \label{eq:Bucket main technical eq6}\\
    &\leq \frac{4 k^2 \cdot \CGap \cdot \degfracS}{3\Upsilon(k)} \sum_{i=1}^k |S_i| \Vol{S_i}, \nonumber
\end{align}
where \eqref{eq:Bucket main technical eq4} follows by grouping the terms $|B_t \cap S_i| \Vol{B_t \cap S_j}$ into two categories, \eqref{eq:Bucket main technical eq5} follows by (B4), and \eqref{eq:Bucket main technical eq6} follows by (B5).
\end{proof}

Now we combine the properties shown earlier, and upper bound the cost of 
 our constructed tree $\TSCD$ through the following lemma.

\begin{lemma}\label{lemma:cost of the tree}
   It holds that 
   \[
        \COST_G \lp \TSCD \rp \leq \lp 1 + \frac{5k^4 \cdot \CGap}{\Upsilon(k)}\rp \cdot \frac{36 \cdot \degfracS^2}{\phiin} \cdot \OPT_G.
   \]
\end{lemma}

\begin{proof}
By construction, we know that $\TSCD$ is the caterpillar tree formed by merging the trees $\tree_{B_j}$, for all $B_j \in \buckets$.
Therefore, in order to bound the overall cost it suffices to bound the cost within each induced subtree $\tree_{B_j}$ as well as the cost of the edges crossing different buckets.
Formally, we have that
\begin{align}
    \COST_G \lp \TSCD \rp
    &= \sum_{j=1}^{\ell} \COST_{G[B_j]}\lp \tree_{B_j}\rp
      + \sum_{j=1}^{\ell - 1} \sum_{t = j+1}^{\ell} \sum_{e \in E(B_j, B_t)} \cost_G(e) \nonumber \\
    &\leq \sum_{j=1}^{\ell} |B_j| \Vol{B_j}
      + \sum_{j=1}^{\ell - 1} \sum_{t = j+1}^{\ell} \sum_{e \in E(B_j, B_t)} \cost_G(e). \label{eq:total cost}
\end{align}
W study the first term of \eqref{eq:total cost}, and have that
\begin{align}
    &\sum_{j=1}^{\ell} |B_j|\Vol{B_j}\nonumber\\
    &\leq \sum_{B_j \in \heavyBkts} |B_j|\Vol{B_j} +
    \sum_{B_t \in \lightBkts} |B_t|\Vol{B_t}\nonumber\\
    &\leq \lp 1 + \degfracS\rp \lp 1 + \frac{k \cdot \CGap}{3\Upsilon(k)}\rp \lp \sum_{i=1}^k |S_i| \Vol{S_i} \rp + \frac{4k^2 \cdot \CGap \cdot \degfracS}{3\Upsilon(k)} \cdot \lp \sum_{i=1}^k |S_i| \Vol{S_i}\rp \label{eq:internal cost eq1}\\
    &\leq \lp 2\degfracS \cdot \lp 1 + \frac{k \cdot \CGap}{3\Upsilon(k)}\rp + \frac{4k^2 \cdot \CGap \cdot \degfracS}{3\Upsilon(k)}\rp \sum_{i=1}^k |S_i| \Vol{S_i}\nonumber\\
    &\leq \lp 2\degfracS + \frac{2k^2 \cdot \CGap \cdot \degfracS}{\Upsilon(k)}\rp \cdot \frac{18\cdot \degfracS}{\phiin} \cdot \OPT_G\label{eq:total internal cost}\\
    &= \lp 1 + \frac{k^2 \cdot \CGap}{\Upsilon(k)}\rp \cdot \frac{36\cdot \degfracS^2}{\phiin} \cdot \OPT_G \nonumber
\end{align}
where \eqref{eq:internal cost eq1} follows from (B1), (B2),  (B6),   and  \eqref{eq:total internal cost} follows by Lemma~\ref{lem:lower_bound_sum_SiVolSi}.

Next  we look at the second term of \eqref{eq:total cost}.
Let $B_j \in \buckets$ be an arbitrary bucket, and we study  the cost of the edges between $B_j$ and all the other buckets $B_t$ positioned lower in the tree, i.e., $|B_t| \leq |B_j|$.
By construction, for every such edge $e$ we have that
\[
    \cost_G(e) 
    = w_e \sum_{t=j+1}^{\ell} |B_t|
    \leq w_e \cdot \ell \cdot |B_j|.
\]
Therefore, we can upper bound the total cost of the crossing edges adjacent to $B_j$ by 
\[
    \sum_{t=j+1}^{\ell} \sum_{e \in E(B_j, B_t)} \cost_G(e)
    \leq \ell \cdot |B_j| w(B_j, \overline{B_j}).
\]

Hence, by summing over all buckets $B_j$ we have that 
\begin{align}
    &\sum_{j=1}^{\ell - 1} \sum_{t = j+1}^{\ell} \sum_{e \in E(B_j, B_t)} \cost_G(e) \nonumber \\
    &\leq \sum_{j=1}^{\ell - 1} \ell \cdot |B_j|w(B_j, \overline{B_j}) \nonumber\\
    &\leq \ell \lp \sum_{B_j \in \heavyBkts} |B_j|w(B_j, \overline{B_j}) 
        + \sum_{B_t \in \lightBkts} |B_t|\Vol{B_t} \rp \nonumber\\
    &\leq \ell \lp \lp 1 + \degfracS\rp \cdot \frac{k \cdot \CGap + \lambda_{k+1}}{\Upsilon(k)} \cdot  \lp \sum_{i=1}^k |S_i|\Vol{S_i} \rp 
        + \frac{4 k^2 \cdot \CGap \cdot \degfracS}{3\Upsilon(k)} \cdot \lp \sum_{i=1}^k |S_i| \Vol{S_i}\rp \rp \label{eq:crossing cost eq1} \\
    &\leq \ell \lp \frac{2\degfracS \cdot 4k \cdot \CGap }{3\Upsilon(k)} + \frac{4 k^2 \cdot \CGap \cdot \degfracS}{3\Upsilon(k)}\rp \sum_{i=1}^k |S_i|\Vol{S_i} \nonumber\\
    &\leq \ell \lp \frac{4k^2 \cdot \CGap \cdot \degfracS}{\Upsilon(k)}\rp
    \cdot \frac{18 \cdot \degfracS}{\phiin} \cdot \OPT_G \label{eq:total crossing cost}\\
    &\leq \frac{4k^4 \cdot \CGap}{\Upsilon(k)} \cdot \frac{36 \cdot \degfracS^2}{\phiin} \cdot \OPT_G\nonumber
\end{align}
where \eqref{eq:crossing cost eq1} follows by (B2) and (B3) for the heavy buckets and (B6) for the light ones, and  \eqref{eq:total crossing cost} follows by Lemma~\ref{lem:lower_bound_sum_SiVolSi}.
Finally, combining~\eqref{eq:total cost}, \eqref{eq:total internal cost} and \eqref{eq:total crossing cost} we conclude that
\begin{align*}
    \COST_G \lp \TSCD \rp
    &\leq \lp 1 + \frac{k^2 \cdot \CGap}{\Upsilon(k)}\rp \cdot \frac{36\cdot \degfracS^2}{\phiin} \cdot \OPT_G
        + \frac{4k^4 \cdot \CGap}{\Upsilon(k)} \cdot \frac{36 \cdot \degfracS^2}{\phiin} \cdot \OPT_G \\
    &\leq \lp 1 + \frac{5k^4 \cdot \CGap}{\Upsilon(k)}\rp \cdot \frac{36 \cdot \degfracS^2}{\phiin} \cdot \OPT_G,
\end{align*}
which completes the proof.
\end{proof}

Finally, the following theorem summarises the performance of Algorithm~\ref{alg:spectral_degree_clustering}.

\begin{theorem}[Formal Statement of Theorem~\ref{thm:main2}]\label{thm:HSC_degree_bucket}
 Let $G=(V,E,w)$ graph of $n$ vertices, $m$ edges,  and optimal clusters $\{S_i\}_{i=1}^k$ corresponding to $\rho(k)$ satisfying $k=O(\mathrm{poly}\log (n))$, $\Phi_{G[S_i]} \geq \phiin$ for $1 \leq i \leq k$, and   $\Upsilon(k) \geq \CGap \cdot k$. Then,
there is an     $\tilde{O}(m)$ time algorithm that returns   both an  \textsf{HC} 
 tree  $\TSCD$ of $G$, and $\COST \lp \TSCD \rp $.
 Moreover, it holds that 
\begin{align*}
    \COST_G \lp \TSCD \rp \leq \lp 1 + \frac{5k^4 \cdot \CGap}{\Upsilon(k)}\rp \cdot \frac{36 \cdot \degfracS^2}{\phiin} \cdot \OPT_G
\end{align*}
where $\degfracS$ is an upper bound of  $ \max_i \left(\dmax(S_i) / \dmin(S_i)\right)$.
\end{theorem}

\begin{proof} 
The approximation guarantee of the returned tree $\TSCD$ from  Algorithm~\ref{alg:spectral_degree_clustering} follows from Lemma~\ref{lemma:cost of the tree}. By a similar proof as the proof of Lemma~\ref{lem:cost_constructing_treeDT}, it holds that the time complexity of Algorithm~\ref{alg:spectral_degree_clustering} is $\tilde{O}(m)$ and the depth of $\TSCD$ is $O(\polylog n)$. Therefore, by Lemma~\ref{lem:computing_cost_nearlylinear_dt}, we can compute $\COST_G(\TSCD)$ in nearly-linear time $\tilde{O}(m)$.

It remains to deal with our assumption that the algorithm knows the number of clusters $k$ and an upper bound
\[
    \degfracS \geq \max_{1\leq i\leq k} \frac{\Delta(S_i)}{\delta(S_i)}.
\]
If the number of clusters $k$ is unknown, we perform the technique described before   and run independent copies of Algorithm~\ref{alg:spectral_degree_clustering} with all possible values $k'$ ranging from $1$ to $O(\polylog(n))$.
By adding an extra $O(\polylog n) = \tilde{O}(1)$ factor in the overall time complexity, we ensure that one of the runs has the correct number of clusters $k' = k$.

Finally, in order to obtain an approximation of $\degfracS$,
we additionally run $O(\log(n))$ independent copies of Algorithm~\ref{alg:spectral_degree_clustering} with different values $\tilde{\degfracS} = 2^i$, for all $i \in \left[ \log \delta_G, \log \Delta_G \right]$.
This process ensures that, at least one such estimate satisfies that $\tilde{\degfracS}/2 \leq \degfracS \leq \tilde{\degfracS}$.
By introducing a factor of $O(\log n)$ to the overall running time, this ensures that at least one of the constructed trees of Algorithm~\ref{alg:spectral_degree_clustering} satisfies the promised approximation ratio, up to an additional factor of $\lp \degfracS/\tilde{\degfracS} \rp^2$.
As a result, we return the tree with the minimum cost among all different runs of our algorithm.
\end{proof}

\end{document}